\documentclass[a4paper,USenglish,cleveref,nameinlink,numberwithinsect,thm-restate]{lipics-v2021}
\pdfoutput=1
\nolinenumbers
\usepackage{microtype}
\usepackage{framed}
\usepackage{xstring}
\makeatletter
\newcommand{\seq}[1]{
\saveexpandmode\expandarg
\StrDel{#1}{&}
\restoreexpandmode
} 
\makeatother

\title{\texorpdfstring{From Chinese Postman to Salesman and Beyond I:\\ Approximating Shortest Tours $\delta$-Covering All Points on All Edges}{From Chinese Postman to Salesman and Beyond I: Approximating Shortest Tours $\delta$-Covering All Points on All Edges}}
\titlerunning{From Chinese Postman to Salesman and Beyond I}

\author{Fabian Frei}{CISPA Helmholtz Center for Information Security, Saarbrücken}{fabian.frei@cispa.de}{https://orcid.org/0000-0002-1368-3205}{}
\author{Ahmed Ghazy}{CISPA Helmholtz Center for Information Security and Saarland University, Saarbrücken}{ahmed.ghazy@cispa.de}{https://orcid.org/0009-0009-7414-5871}{}
\author{Tim A.~Hartmann}{CISPA Helmholtz Center for Information Security, Saarbrücken}{tim.hartmann@cispa.de}{https://orcid.org/0000-0002-1028-6351}{}
\author{Florian Hörsch}{CISPA Helmholtz Center for Information Security, Saarbrücken}{florian.hoersch@cispa.de}{https://orcid.org/0000-0002-5410-613X}{}
\author{Dániel Marx}{CISPA Helmholtz Center for Information Security, Saarbrücken}{marx@cispa.de}{https://orcid.org/0000-0002-5686-8314}{}

\authorrunning{F.~Frei, A.~Ghazy, T.~A.~Hartmann, F.~Hörsch, D.~Marx}

\Copyright{Fabian Frei, Ahmed Ghazy, Tim A.~Hartmann, Florian Hörsch, Dániel Marx}



\category{}

\relatedversion{}
\usepackage[utf8]{inputenc}
\usepackage[T1]{fontenc}
\usepackage{latexsym}

\usepackage{amssymb}
\usepackage{bibentry}
\usepackage{amsmath}
\usepackage{amsthm}
\usepackage{thmtools}
\usepackage{amsthm}
\usepackage{hyperref}
\usepackage[all]{hypcap}
\usepackage{mathtools}
\usepackage{xspace} 
\usepackage{xparse}
\usepackage{dsfont} 
\usepackage{tabularx} 
\usepackage{adjustbox}
\usepackage{changepage}
\usepackage{subcaption}
\usepackage{hyperref}
\usepackage[textsize=footnotesize,backgroundcolor=white]{todonotes}

\usepackage{algorithm}
\DeclareCaptionLabelFormat{nonumber}{%
  \kern0.05em{\color[rgb]{0.99,0.78,0.07}\rule{0.73em}{0.73em}}%
  \hspace*{0.67em}\bothIfFirst{#1}{~}} 

\let\oldtextsc\textsc
\renewcommand{\textsc}[1]{\mbox{\oldtextsc{#1}}}

\newtheorem{Case}{Case}
\crefname{claim}{Claim}{Claims}

\makeatletter
\newcommand{\xlabel}[2][]{\phantomsection\def\@currentlabelname{\ifthenelse{\equal{#1}{}}{#2}{#1}}\label{#2}}
\makeatother

\usepackage{tikz}
\usetikzlibrary{decorations.pathmorphing}
\usetikzlibrary{decorations.markings}
\usetikzlibrary{decorations,arrows.meta}
\pgfdeclaremetadecoration{many arrows}{initial}{ 
\state{initial}[width=0pt, next state=arrow] {
    \pgfmathdivide{100}{\pgfmetadecoratedpathlength}
    
    \pgfset{/pgf/decoration/segment length=4pt}
  }
  \state{arrow}[
switch if less than=\pgfmetadecorationsegmentlength to final, width=\pgfmetadecorationsegmentlength/3,
next state=end arrow]
  {
    \decoration{curveto}
    \beforedecoration
    {
      \pgfpathmoveto{\pgfpointmetadecoratedpathfirst}
} }
\state{end arrow}[width=\pgfmetadecorationsegmentlength/3, next state=move] {
    \decoration{curveto}
    \beforedecoration{\pgfpathmoveto{\pgfpointmetadecoratedpathfirst}}
    \afterdecoration
    {
      \pgfsetarrowsend{Latex[length=1pt,width=1pt]}
      \pgfusepath{stroke}
    }
}
\state{move}[width=\pgfmetadecorationsegmentlength/2, next state=arrow]{} \state{final}{}
}

\usetikzlibrary{shapes}
\usetikzlibrary{calc}
\usetikzlibrary{patterns}
\usetikzlibrary{fit}
\usepackage{pgfplots}

\newcommand{\p}{\textup{P}\xspace}
\newcommand{\np}{\textup{NP}\xspace}
\newcommand{\fpt}{\textup{FPT}\xspace}

\newcommand{\apx}{\textup{APX}\xspace}
\newcommand{\wone}{\textup{W[1]}\xspace}
\newcommand{\wtwo}{\textup{W[2]}\xspace}

\newcommand{\Oh}{\mathcal{O}}
\renewcommand{\Oh}{O}

\newcommand{\TSP}{\textup{TSP}\xspace}

\newcommand{\abs}[1]{\lvert #1\rvert}
\renewcommand{\setminus}{-}

\newcommand{\ceil}[1]{\ensuremath{\lceil #1 \rceil}}
\newcommand{\floor}[1]{\ensuremath{\lfloor #1 \rfloor}}

\DeclareMathOperator{\polylog}{polylog}

\newcommand{\half}{\ensuremath{{\tfrac{1}{2}}}}

\renewcommand{\epsilon}{\varepsilon}
\def\phi{\ensuremath{\varphi}}

\newcommand{\tour}{\textsc{Tour}\xspace}

\renewcommand{\textsc}[1]{\textnormal{\scshape #1}}
\newcommand{\VC}{\textsc{Vertex}\-\textsc{Cover}\xspace}
\newcommand{\domset}{\textsc{Dominating}\-\textsc{Set}\xspace}

\newcommand{\Center}{\textsc{Center}\xspace}
\newcommand{\covering}{\textsc{Covering}\xspace}
\newcommand{\dispersion}{\textsc{Dispersion}\xspace}
\newcommand{\IndependentSet}{\textsc{Independent}\-\textsc{Set}\xspace}

\newcommand{\CPP}{\textsc{Chinese}\-\textsc{Postman}\-\textsc{Problem}\xspace}

\newcommand{\MetricTSP}{\textsc{Metric}\-\textsc{TSP}\xspace}

\usepackage{tabularx, environ}

\makeatletter
\newcolumntype{\expand}{}
\long\@namedef{NC@rewrite@\string\expand}{\expandafter\NC@find}

\newcounter{myproblem}
\newcounter{temp}
\NewEnviron{myproblem}[1][]{%
\setcounter{temp}{\value{algorithm}}%
\setcounter{algorithm}{\value{myproblem}}%
\captionsetup[algorithm]{name=Optimization Problem, labelformat=nonumber, position=top, listformat=empty}%
  \noindent%
\begin{algorithm}[H]\caption{#1}%
\refstepcounter{myproblem}%
\vspace{1ex}%
\begin{tabularx}{\textwidth}{>{\bfseries}lXc}%
	\phantomsection%
    \BODY
  \end{tabularx}%
  \smallskip
\end{algorithm}
\setcounter{algorithm}{\value{temp}}
\LinkTargetOn
}

\crefname{algorithm}{Algorithm}{Algorithms}
\crefname{myproblem}{Problem}{Problems}
\Crefname{algorithm}{Algorithm}{Algorithms}
\Crefname{myproblem}{Problem}{Problems}

\makeatother

\NewEnviron{framedproblem}[1][]{%
  \noindent%
  
  \vspace{1ex}
  \begin{tabularx}{\textwidth}{|>{\bfseries}lX|c|}%
	\hline
    \BODY\\\hline
  \end{tabularx}%
  \smallskip
}
\makeatother

\newcommand{\dist}{\operatorname{{d}}}
\newcommand{\len}{\operatorname{\mathsf{\ell}}}

\newcommand{\opt}{\ensuremath{\mathsf{OPT}}}

\newcommand{\opttour}[1][\delta]{\ensuremath{\opt_{#1\textup{-tour}}}\xspace}

\newcommand\opttsp{\opt_{\textup{TSP}}\xspace}
\newcommand\optlp{\opt_{\textup{LP}}\xspace}

\newcommand{\deltatour}[1][\delta]{{\ensuremath{{#1\textup{-tour}}}}\xspace}
\newcommand{\deltatourprob}[1][\delta]{{\ensuremath{{#1\textup{-\textsc{Tour}}}}}\xspace}

\newcommand{\conn}[4]{\ensuremath{\ifthenelse{\equal{#4}{0}}{#1_{#2}^{#3}}{\bar{#1}_{#2}^{#3}}}}

\newcommand{\midpath}[2]{\ensuremath{\ifthenelse{\equal{#2}{0}}{M_{#1}}{\bar{M}_{#1}}}}

\input{restatable-helpers_2.tex}

\hideLIPIcs

\begin{document}
\maketitle
\begin{abstract}
A well-studied continuous model of graphs, introduced by Dearing and Francis [Transportation Science, 1974],
	considers each edge as a continuous unit-length interval of points.
        For \(\delta \geq 0\), we introduce the problem \(\delta\)-\tour, where the objective is to find the shortest tour that comes within a distance of \(\delta\) of every point on every edge. It can be observed that \(0\)-\tour is essentially equivalent to the Chinese Postman Problem, which is solvable in polynomial time. In contrast, \(1/2\)-\tour is essentially equivalent to the Graphic Traveling Salesman Problem (TSP), which is NP-hard but admits a constant-factor approximation in polynomial time. We investigate \(\delta\)-\tour for other values of \(\delta\),
		noting that the problem's behavior and the insights required to understand it differ significantly across various \(\delta\) regimes.
		We design polynomial-time approximation algorithms summarized as follows:
        \begin{enumerate}
          \item[(1)] For every fixed \(0 < \delta < 3/2\), the problem \(\delta\)-\tour admits a constant-factor approximation.
		  \item[(2)] For every fixed \(\delta \geq 3/2\), the problem admits an \(\Oh(\log n)\)-approximation.
		  \item[(3)] If $\delta$ is considered to be part of the input, then the problem admits an $\Oh(\log^3{n})$-approximation.
          \end{enumerate}
 This is the first of two articles on the \(\delta\)-\tour problem. In the second one we complement the approximation algorithms presented here with inapproximability results and  related to parameterized complexity.      
\bigskip
\end{abstract}
\newpage

\section{Introduction}

We consider a well-studied continuous model of graphs introduced by Dearing and Francis~\cite{Dearing1974}.
Each edge is seen as a continuous unit interval of points with its vertices as endpoints. 
For any given graph~$G$, this yields a compact metric space $(P(G),\dist)$ with a point set $P(G)$ and a distance function $\dist\colon P(G)^2\to \mathbb{R}_{\ge0}$.

A prototypical problem in this setting is $\delta$-\covering, 
introduced by Shier~\cite{Shier1977} for any positive real $\delta$.
The task is to find in $G$ a minimum set $S$ of points that \emph{$\delta$-covers} 
the entire graph, in the sense that each point in $P(G)$ has distance at most $\delta$ to some point in $S$.
This problem, which is also often referred to as the continuous $p$-\Center problem has been extensively studied; 
we cite only a few examples:~\cite{Kariv1979,ChandrasekaranTamir1980,MegiddoTamir1983}.
Observe that the problem differs from typical discrete graph problems 
in two ways: 
the solution has to $\delta$-cover every point of every edge 
(not just the vertices) and the solution may (and for optimality sometimes must) use points inside edges.
        How does the complexity of this problem depend on the distance $\delta$?
First, the problem is polynomial-time solvable when $\delta$ is a unit fraction, i.e., a rational with numerator $1$,
	and \np-hard for all other rational and irrational $\delta$~\cite{Hartmann2022,HartmannLW22}.
One can show that \VC is reducible to $2/3$-\covering and 
\domset is reducible to $3/2$-\covering. 
Thus $\delta$-\covering behaves very differently for different values of $\delta$ 
and can express problems of different nature and complexity: 
for example, while vertex cover is fixed-parameter tractable (\fpt) 
	when parameterized by the solution size, 
	dominating set is \wtwo-hard.
This is reflected also in the complexity of $\delta$-\covering: 
at the threshold of $\delta=3/2$, the parameterized complexity of the problem, parameterized by the size of the solution,
	jumps from \fpt to \wtwo-hard~\cite{HartmannLW22}.
Similarly, $\delta$-\covering allows a constant factor approximation
	for $\delta<3/2$ and becomes log-\apx-hard for $\delta \geq 
	3/2$~\cite{HartmannJanssen2024}.
The problem dual to $\delta$-\covering is $\delta$-\dispersion,
	as studied for example by Shier and Tamir~\cite{Shier1977,Tamir1991}.
The task is to place a maximum number of points in the input graph
	such that they pairwise have distance at least $\delta$ from each other.
For this problem, $\delta=2$ marks the threshold where the parameterized complexity
	for the solution size as the parameter
	jumps from \fpt to \wone-hard~\cite{HartmannL22}.
Furthermore, the problem is polynomial-time solvable when $\delta$ is a rational with numerator 1 or 2, 
	and \np-hard for all other rational and irrational $\delta$~\cite{GrigorievHLW21, HartmannL22}.
With $\delta$-\covering being the continuous version of \VC and \domset,
	and $\delta$-\dispersion being a continuous version of \IndependentSet, 
	we now turn to the natural continuous variant of another famous problem.

We study the Graphic Traveling Salesman Problem (TSP)
	with a positive real covering range $\delta$ in the continuous model,
	which we call \deltatourprob.
A \emph{\deltatour} $T$ is a tour
that may make U-turns at arbitrary points of the graph, even inside edges, 
and is $\delta$-covering, that is, every point in the graph is within distance $\delta$ from a point $T$ passes by.
The task in our problem \deltatourprob is to find a shortest \deltatour.
See~\cref{fig:example} for two examples of unique shortest \deltatour{}s that cannot be described as graph-theoretic closed walks. 

Note that computing a shortest $0$-tour is equivalent to computing a shortest Chinese Postman tour (a closed walk going through every edge),
	which is known to be polynomial-time solvable~\cite[Chapter~29]{schrijver-book}.
	Moreover, one can observe that if every vertex of the input graph has degree at least two,
	then there is a shortest $\frac12$-tour that stops at every vertex and, conversely,
	any tour stopping at every vertex is a $\frac12$-tour.
	Thus, finding a shortest $\frac12$-tour is essentially equivalent to
	solving a TSP instance, with some additional careful handling of degree-1 vertices.

\begin{figure}[t]
\label{fig:results}
\centering
\begin{tikzpicture}[
  declare function={
    ubone(\d)=
	    (\d == 0) * (1)   +
	    and(\d > 0, \d <= 1/6) * (1/(1-2*\d))   +
	    and(\d > 1/6, \d <= 1/2) * (1.5)
   ;
    ubtwo(\d)=
	    and(\d >= 1/2, \d <= 0.825) * (1.4/(2-2*\d))     +
	    and(\d >= 0.825, \d < 1) * (4)
   ;
    ubthr(\d)=
	    (\d == 1) * (3) +
	    and(\d > 1, \d <= 1.5) * (3/(3-2*\d))
   ;
   lb(\d) = 
	   and(\d >= 1, \d < 1.5) * (2.5-\d)
   ;
  }
]

\begin{axis}[width=\textwidth, height=7cm,
  axis x line=left, axis y line=left,
  ymajorgrids=true, grid style=dashed,
  ymin=0, ymax=5.1, ytick={1, 1.5, 3, 4}, ylabel=approximation ratio,
  xmin=0, xmax=1.55, xtick={0,1/6,1/2,0.825,1,1.5}, xlabel={covering range $\delta$},
  xticklabels={$0$,$1/6$,$1/2$,$33/40$,$1$,$3/2$},
  domain=0:1.5,samples=50,
]

\draw [draw=black!100!green, fill=black!100!green, very thick] (axis cs: 0.5, 1.4) circle (1.0pt);

\draw [draw=black!100!green, fill=black!100!green, very thick] (axis cs: 1, 3) circle (1.0pt);

\addplot [black!100!green,domain=0:.5] {ubone(x)};
\addplot [black!100!green,domain=.5:1] {ubtwo(x)};
\addplot [black!100!green,domain=1:1.2] {ubthr(x)};

\end{axis}
\end{tikzpicture} 

\caption{The approximation ratio of our algorithms for \deltatourprob plotted against $\delta$.}
\label{figure:approximation}
\end{figure}

\begin{table}[t]
\label{tab:results}
\centering
\caption{Approximation upper bounds (UB) for \deltatourprob.}
\setlength\tabcolsep{3.13pt}
\begin{adjustwidth}{-0.0cm}{}
	\begin{tabular}{|c|c|c|c|c|c|c|c|c|c|}
		\hline
			$\boldsymbol{\delta}$ &
		$(0,1/6]$ &
		\multicolumn{2}{c|}{ $(1/6,1/2)$ } &
		$1/2$ &
		$(1/2,33/40)$ &
		$[33/40,1)$ &
		$[1,3/2)$ & $[3/2,\infty)$ 
		\\
		\hline
			\textbf{UB} &
		$1/(1{-}2\delta)$ &
		\multicolumn{2}{c|}{ 1.5 } &
		$1.4$ &
		$1.4/(2-2\delta)$ &
		$4$ &
		$3/(3{-}2\delta)$ &
		$\min\{2\delta, 64\log^2{n}\} \log{n}$ \\
			&
		Thm.~\ref{thm:approx:ub:zero_sixth} &
		\multicolumn{2}{c|}{Thm.~\ref{thm:approx:ub:sixth_half} } &
		Thm.~\ref{thm:approx:ub:half} &
		Thm.~\ref{thm:approx:ub:half:threequarters}&
		Thm.~\ref{thm:approx:ub:threequarters:one} &
		Thm.~\ref{thm:approx:ub:one:threehalves} &
		Thms.~\ref{thm:approx:ub:threehalves:lognpthree} and~\ref{thm:approx:ub:threehalves:logn}\\
		\hline
	\end{tabular}
\end{adjustwidth}
\label{table:approximation}
\end{table}

\subparagraph{Our Results.}
It turns out that finding a shortest \deltatour is \np-hard for all $\delta>0$.
In fact, in~\cite{FreiGHHM24}, we show \apx-hardness for every fixed $\delta\in (0,3/2)$.
In this work, we therefore present approximation algorithms.
As is standard, an \emph{$\alpha$-approximation algorithm} is one that runs
	in polynomial time and finds a solution of value within a factor $\alpha$ of the optimum.
As our main approximation result, for every fixed $\delta \in (0,3/2)$,
	we give constant-factor approximation algorithms for finding a shortest \deltatour.
We list our results in \cref{table:approximation}
	and plot the approximation ratio against $\delta$ in \cref{figure:approximation}.

\begin{theorem}[\textbf{Constant-Factor Approximation}]\label{thm:approxmain}
For every fixed $\delta \in (0,3/2)$, the problem \deltatourprob admits a polynomial-time constant-factor approximation algorithm.
\end{theorem}          
        
The problem behaves very differently in the various regimes of $\delta$, even within the range of $(0,3/2)$, 
and we exploit connections to different problems for different values of $\delta$:

\begin{description}
\item[{Case \boldmath$\delta\in(0,1/6]$.}]
	There is a close relation between our problem and the \CPP in this range, 
	which gives a good approximation ratio.
When $\delta$ approaches $0$, our approximation ratio approaches $1$.
See \cref{thm:approx:ub:zero_sixth}.

\item[{Case \boldmath$\delta\in(1/6,33/40)$.}]
	When $1/6 < \delta < 1/2$, the problem can be reduced to solving TSP on metric instances, for
	which we can use Christofides' 3/2-approximation
	algorithm~\cite{Christofides2022} to obtain the same approximation
	ratio for our problem. See \cref{thm:approx:ub:sixth_half}.

	A simplification of this approach for $\delta = 1/2$ allows us to
	use the better $7/5$-approximation for TSP due to Seb\H o and Vygen \cite{SeboV14}. See \cref{thm:approx:ub:half}. 

	Finally, for $1/2 < \delta < 33/40$, it turns out that a \deltatour[\frac12]
	is a good approximation of a \deltatour. See \cref{thm:approx:ub:half:threequarters}.

\item[{Case \boldmath$\delta\in[33/40,3/2)$.}]
	The problem here is closely related to a variation of the \VC problem,  
	some results on which we exploit in our approximation algorithms~\cite{ArkinHH1993,KonemannKPS03}.
See Theorems~\ref{thm:approx:ub:one:threehalves}~and~\ref{thm:approx:ub:threequarters:one}.
      \end{description}

      Once $\delta$ reaches $3/2$, the problem \deltatourprob suddenly changes character: it becomes similar to \domset,
      where only a logarithmic-factor approximation is known; see \cref{thm:approx:ub:threehalves:logn}.
	  In~\cite{FreiGHHM24}, we show that this is best possible unless $\p = \np$.

\begin{theorem}[\textbf{Logarithmic Approximation}]\label{thm:logmain}
	For every fixed $\delta \geq 3/2$, the problem \deltatourprob admits a polynomial-time $\Oh(\log n)$-approximation algorithm.
\end{theorem}

The above approximation ratio in fact depends on $\delta$, which the big-$\Oh$ notation hides.
Thus, if $\delta$ is not fixed and is rather given as an input,
this approximation guarantee can be arbitrarily bad. We show that a
polylogarithmic-factor approximation is fortunately still possible in that setting.

\begin{restatable}[Polylogarithmic Approximation]{linkedtheorem}{ThmApproxUbThreeHalvesLogNpThree}
\label{ThmApproxUbThreeHalvesLogNpThree}
\label{thm:approx:ub:threehalves:lognpthree}
	There is a polynomial-time algorithm that, given $\delta > 0$ and a graph $G$ of order
	$n$, computes a $64(\log n)^3$-approximation of a shortest \deltatour
	of $G$.
\end{restatable}

\Cref{section:prel} begins with formal notions including a thorough definition of a \deltatour.
Then \cref{section:overview} gives an extended overview of our results.

\section{Formal Definitions}
\label{section:prel}

\subparagraph{General Definitions.}
For a positive integer $n$, we denote the set $\{1, \dots, n\}$ by $[n]$.
All graphs in this article are undirected, unweighted and do not contain parallel edges or loops. Let $G$ be a graph.
For a subset of vertices $V' \subseteq V(G)$, we denote by $G[V']$ the subgraph induced by $V'$.
The neighborhood of a vertex $u$ is $N_G(u) \coloneqq \{ v \in V(G) \mid uv \in E(G)\}$. 
We write $uv$ for an edge $\{u,v\} \in E(G)$.
We denote by $\log$ the binary logarithm.

\subparagraph{Problem Related Definitions.}
For a graph~$G$, we define a metric space whose 
	point set $P(G)$ contains, somewhat informally speaking, all points on the continuum of each edge, which has unit length. 
We use the word \emph{vertex} for the elements in $V(G)$, 
	whereas we use the word \emph{point} to denote elements in $P(G)$. Note however, that each vertex of $G$ is also a point of $G$.

The set $P(G)$ is the set of points $p(u,v,\lambda)$ for every edge $uv \in E(G)$ and every $\lambda \in [0,1]$
	where $p(u,v,\lambda)=p(v,u,1-\lambda)$; $p(u,v,0)$ coincides with $u$
	and $p(u,v,1)$ coincides with $v$.

The \emph{edge segment} $P(p,q)$ of $p$ and $q$ then is the subset of points
	$\{ p(u,v,\mu) \mid \min\{\lambda_p,\lambda_q\} \leq \mu \leq \max\{\lambda_p,\lambda_q\} \}$.
A $\seq{p&q}$-\emph{walk} $T$ between points $p_0 \coloneqq p$ and $p_z \coloneqq q$
	is a finite sequence of points $\seq{p_0&p_1&\dots&p_z}$
	where every two consecutive points are distinct and lie on the same edge,
	that is, formally, for every $i\in[z]$
	there are an edge $u_i v_i \in E(G)$ and $\lambda_i,\mu_i \in [0,1]$ with $\lambda_i \neq \mu_i$
	such that $p_{i-1}=p(u_i,v_i,\lambda_i)$ and $p_{i}=p(u_i,v_i,\mu_i)$.
When $p$ and $q$ are not specified, we may simply write \emph{walk} instead of $\seq{p&q}$-walk.
The points in the sequence defining a walk are called its \emph{stopping points}. The point set of $T$ is $P(T)= \bigcup_{i \in [z]} P(p_{i-1},p_i)$. For some $p\in P(T)$, we say that $T$ {\it passes} $p$.

For two points $p,q \in P(G)$, we now define their so-called {\it distance} $\dist(p,q)$. 
If $p$ and $q$ are on the same edge $uv$,
say $p=p(u,v,\lambda_p)$ and $q=p(u,v,\lambda_q)$, we define  $\dist(p,q)=\abs{\lambda_q-\lambda_p}$. In order to extend this definition to arbitrary point pairs, we first need the definition of the length of a walk. 
Namely, the \emph{length} $\len(T)$ of a walk $T=\seq{p_0&p_1&\dots&p_z}$ is $\sum_{i\in[z]}\dist(p_{i-1},p_i)$.
A $\seq{p&q}$-walk with minimum length among all $\seq{p&q}$-walks is called \emph{shortest}. We can now extend the notion of distance to arbitrary point pairs.
The \emph{distance between two points} $p,q \in P(G)$, denoted $\dist(p,q)$,
	is either the length of a shortest $\seq{p&q}$-walk or, if no such walk exists, $\infty$ . It is easy to see that both definitions coincide for points located on the same edge.
Further, let $\dist(p,Q) = \inf\{\dist(p,q) \mid q \in Q \}$ for any point $p \in P(G)$ and point set $Q \subseteq P(G)$.
For some $p\in P(G)$ and a walk $T$, we write $\dist(p,T)$ as an abbreviation for $\dist(p,P(T))$.
We later show that $\dist(p, T)$ is in fact a minimum taken over the set of stopping points of $T$ (see \cref{nearstop}).

A \emph{tour} $T$ is a $\seq{p_0&p_z}$-walk with $p_0=p_z$.
For a real $\delta>0$, a $\delta$-\emph{tour} is a tour
where $\dist(p,T) \leq \delta$ for every point $p \in P(G)$.
We study the following minimization problem.
\begin{myproblem}[\deltatourprob, where $\delta \geq 0$]
\label{prob:deltatour}%
Instance&A connected simple graph $G$.\\
Solution&Any \deltatour $T$.\\
Goal& Minimize the length $\len(T)$.
\end{myproblem}

Further, we use the following notions for a tour $T=\seq{p_0&p_1&\dots&p_z}$.
A \emph{tour segment} of $T$ is a walk
	given by a contiguous subsequence of $\seq{p_0&p_1&\dots&p_z}$.
The tour $T$ \emph{stops} at a point $p \in P(G)$ if $p \in \{p_0,p_1,\dots,p_z\}$ 
and \emph{traverses} an edge $uv$ if $uv$ or $vu$ is a tour segment of $T$.
The \emph{discrete length} of a tour is $z$, that is, the length 
of the finite sequence of points representing it. We denote the discrete length
of a tour $T$ by $\alpha(T)$.

A point $p\in P(G)$ is \emph{integral} if it coincides with a vertex.
Similarly, $p=p(u,v,\lambda)$ is \emph{half-integral} if $\lambda \in \{0,\half,1\}$. A tour is {\it integral} ({\it half-integral}) if all its stopping points are integral (half-integral).

The \emph{extension} of a tour $T= \seq{p_0&p_1&\dots&p_z}$, denoted as $\ceil{T}$,
	is the integral tour where,
	for every edge $uv \in E(G)$ and every $\lambda<1$,
	every tour segment~$\seq{u&p(u, v, \lambda)&u}$
	is replaced by $\seq{u&v&u}$.

We note that $P(T) \subseteq P(\ceil{T})$.

\begin{figure}[htb]
\begin{subfigure}{.64\textwidth}
\centering
\begin{tikzpicture}[scale=2.1,
		vertex/.style={
			draw,
			circle, 
			inner sep=.0pt, 
			minimum size=0.62cm},
		every edge/.append style={}
			]
\newcommand{\inpointset}{very thick}
	\coordinate (cc1) at (90+0:.5);
	\coordinate (cc2) at (90+120:.5);
	\coordinate (cc3) at (90+240:.5);
	\coordinate (ca1) at (90+0:1.9);
	\coordinate (ca2) at (90+120:1.9);
	\coordinate (ca3) at (90+240:1.9);
	\coordinate (cb1) at ($ (ca1)!.5!(ca2) $);
	\coordinate (cb2) at ($ (ca2)!.5!(ca3) $);
	\coordinate (cb3) at ($ (ca3)!.5!(ca1) $);

	\coordinate (cv1) at ($ (ca1)!.2!(ca2) $);
	\coordinate (cv2) at ($ (ca2)!.2!(ca3) $);
	\coordinate (cv3) at ($ (ca3)!.2!(ca1) $);

	\coordinate (cx1) at ($ (ca1)!.4!(ca2) $);
	\coordinate (cx2) at ($ (ca2)!.4!(ca3) $);
	\coordinate (cx3) at ($ (ca3)!.4!(ca1) $);

	\coordinate (cy1) at ($ (ca1)!.6!(ca2) $);
	\coordinate (cy2) at ($ (ca2)!.6!(ca3) $);
	\coordinate (cy3) at ($ (ca3)!.6!(ca1) $);

	\coordinate (cz1) at ($ (ca1)!.8!(ca2) $);
	\coordinate (cz2) at ($ (ca2)!.8!(ca3) $);
	\coordinate (cz3) at ($ (ca3)!.8!(ca1) $);

	\coordinate (mac1) at ($ (ca1)!.5!(cc1) $);
	\coordinate (mac1l) at ($ (mac1)!.1!90:(cc1) $);
	\coordinate (mac1r) at ($ (mac1)!.1!-90:(cc1) $);
	\coordinate (mac2) at ($ (ca2)!.5!(cc2) $);
	\coordinate (mac2l) at ($ (mac2)!.1!90:(cc2) $);
	\coordinate (mac2r) at ($ (mac2)!.1!-90:(cc2) $);
	\coordinate (mac3) at ($ (ca3)!.5!(cc3) $);
	\coordinate (mac3l) at ($ (mac3)!.1!90:(cc3) $);
	\coordinate (mac3r) at ($ (mac3)!.1!-90:(cc3) $);

	\node[vertex] (c1) at (cc1) {$c_1$};
	\node[vertex] (c2) at (cc2) {$c_2$};
	\node[vertex] (c3) at (cc3) {$c_3$};
	\node[vertex,\inpointset] (a1) at (ca1) {$a_1$};
	\node[vertex,\inpointset] (a2) at (ca2) {$a_2$};
	\node[vertex,\inpointset] (a3) at (ca3) {$a_3$};

	\node[vertex,\inpointset] (v1) at (cv1) {$v_1$};
	\node[vertex,\inpointset] (v2) at (cv2) {$v_2$};
	\node[vertex,\inpointset] (v3) at (cv3) {$v_3$};

	\node[vertex,\inpointset] (x1) at (cx1) {$x_1$};
	\node[vertex,\inpointset] (x2) at (cx2) {$x_2$};
	\node[vertex,\inpointset] (x3) at (cx3) {$x_3$};

	\node[vertex,\inpointset] (y1) at (cy1) {$y_1$};
	\node[vertex,\inpointset] (y2) at (cy2) {$y_2$};
	\node[vertex,\inpointset] (y3) at (cy3) {$y_3$};

	\node[vertex,\inpointset] (z1) at (cz1) {$z_1$};
	\node[vertex,\inpointset] (z2) at (cz2) {$z_2$};
	\node[vertex,\inpointset] (z3) at (cz3) {$z_3$};

	\node[minimum size=0.2cm,inner sep=0,circle,fill=black]  (m1) at (mac1) {};
	\node[minimum size=0.2cm,inner sep=0,circle,fill=black]  (m2) at (mac2) {};
	\node[minimum size=0.2cm,inner sep=0,circle,fill=black]  (m3) at (mac3) {};

	\node[minimum size=0.2cm,inner sep=0,circle,fill=black]  (m1) at (mac1) {};
	\node[minimum size=0.2cm,inner sep=0,circle,fill=black]  (m2) at (mac2) {};
	\node[minimum size=0.2cm,inner sep=0,circle,fill=black]  (m3) at (mac3) {};

\newcommand{\vertexangle}{15}
\newcommand{\vertexdistance}{.19}
\newcommand{\midpointangle}{30}
\newcommand{\midpointdistance}{.09}

	\draw[\inpointset] (a1)--(mac1)--(a1)--(v1)--(x1)--(y1)--(z1)--(a2)--(mac2)--(a2)--(v2)--(x2)--(y2)--(z2)--(a3)--(mac3)--(a3)--(v3)--(x3)--(y3)--(z3)--(a1);
	\draw (c1)--(mac1)--(a1)--(mac1)--(c1)--(c2)--(mac2)--(a2)--(mac2)--(c2)--(c3)--(mac3)--(a3)--(mac3)--(c3)--(c1);
\begin{scope}[/pgf/fpu/install only={reciprocal}]
	\draw [ultra thick,dash pattern={on 6pt off 2pt}, rounded corners=.5] 
  ([shift={(-90+\vertexangle:\vertexdistance)}]ca1)--([shift={(90-\midpointangle:\midpointdistance)}]mac1) arc (90-\midpointangle:90-360+\midpointangle:\midpointdistance)--([shift={(-90-\vertexangle:0\vertexdistance)}]ca1)--([shift={(60-\vertexangle:\vertexdistance)}]cv1) arc (60-\vertexangle:-120+\vertexangle:\vertexdistance)--([shift={(60-\vertexangle:\vertexdistance)}]cx1) arc (60-\vertexangle:-120+\vertexangle:\vertexdistance)--([shift={(60-\vertexangle:\vertexdistance)}]cy1) arc (60-\vertexangle:-120+\vertexangle:\vertexdistance)--([shift={(60-\vertexangle:\vertexdistance)}]cz1) arc (60-\vertexangle:-120+\vertexangle:\vertexdistance)
--([shift={(-90+120+\vertexangle:\vertexdistance)}]ca2)--([shift={(120+90-\midpointangle:\midpointdistance)}]mac2) arc (120+90-\midpointangle:120+90-360+\midpointangle:\midpointdistance)--([shift={(-90+120-\vertexangle:\vertexdistance)}]ca2)--([shift={(180-\vertexangle:\vertexdistance)}]cv2) arc (180-\vertexangle:0+\vertexangle:\vertexdistance)--([shift={(180-\vertexangle:\vertexdistance)}]cx2) arc (180-\vertexangle:0+\vertexangle:\vertexdistance)--([shift={(180-\vertexangle:\vertexdistance)}]cy2) arc (180-\vertexangle:0+\vertexangle:\vertexdistance)--([shift={(180-\vertexangle:\vertexdistance)}]cz2) arc (180-\vertexangle:0+\vertexangle:\vertexdistance)
--([shift={(-90+240+\vertexangle:\vertexdistance)}]ca3)--([shift={(240+90-\midpointangle:\midpointdistance)}]mac3) arc (240+90-\midpointangle:240+90-360+\midpointangle:\midpointdistance)--([shift={(-90+240-\vertexangle:\vertexdistance)}]ca3)--([shift={(300-\vertexangle:\vertexdistance)}]cv3) arc (300-\vertexangle:120+\vertexangle:\vertexdistance)--([shift={(300-\vertexangle:\vertexdistance)}]cx3) arc (300-\vertexangle:120+\vertexangle:\vertexdistance)--([shift={(300-\vertexangle:\vertexdistance)}]cy3) arc (300-\vertexangle:120+\vertexangle:\vertexdistance)--([shift={(300-\vertexangle:\vertexdistance)}]cz3) arc (300-\vertexangle:120+\vertexangle:\vertexdistance)
--cycle;
\end{scope}

\end{tikzpicture}
\caption{A graph and a \deltatour for $\delta=1$.  
The tour $\delta$-covers the inner part of this graph by peeking into three edges up to 
the midpoint. These three peek points are highlighted as the thick dots. 
The depicted tour (the thick dashed line) has length $18$, which is shortest.%
}
\label{fig:example1}
\end{subfigure}
\hfill
\begin{subfigure}{.34\textwidth}
\begin{tikzpicture}[scale=4,rotate=30,
		vertex/.style={
			draw,
			circle, 
			inner sep=.0pt, 
			minimum size=0.62cm},
		every edge/.append style={}
			]

	\coordinate (cu) at ($(0,1)+(-30:.8)$);
	\coordinate (cv) at (0,1);
	\coordinate (cm1) at (0,1/3);
	\coordinate (cm2) at (0,1/6);
	\coordinate (cx) at (0,0);
	\coordinate (cy) at (-120-30:.8);
	\coordinate (cz) at (-120+30:.8);

	\node[vertex] (u) at (cu) {$u$};
	\node[vertex] (v) at (cv) {$v$};
	\node[vertex] (x) at (cx) {$x$};
	\node[vertex] (y) at (cy) {$y$};
	\node[vertex] (z) at (cz) {$z$};

	\node[minimum size=0.2cm,inner sep=0,outer sep=-1mm,circle,fill=black]  (m1) at (cm1) {};
	\node[minimum size=0.2cm,inner sep=0,outer sep=-1mm,circle,fill=black]  (m2) at (cm2) {};

\newcommand{\midpointangle}{29}
\newcommand{\midpointdistance}{.045}
	\draw[very thick] (m1)--(m2);
	\draw (u)--(v)--(m1)--(m2)--(x)--(y)--(z)--(x);
	\draw [ultra thick,dash pattern={on 6pt off 2pt}, rounded corners=.5] 
	([shift={(-90+\midpointangle:\midpointdistance)}]cm1) arc (-90+\midpointangle:360-90-\midpointangle:\midpointdistance)--([shift={(90+\midpointangle:\midpointdistance)}]cm2) arc (90+\midpointangle:360+90-\midpointangle:\midpointdistance)--cycle;
	\end{tikzpicture}
\caption{The depicted shortest \deltatour for $\delta=5/3$ of length $2\cdot1/6$ travels between the two points on edge $vx$ at distances $1/3$ and $1/6$ from $x$.}
\label{fig:example2}
\end{subfigure}
\caption{Two examples of a \deltatour in a graph. On the right, see the special case of a tour fully contained in an edge.}
\label{fig:example}
\end{figure}


\section{Overview of Results}
\label{section:overview}

\Cref{section:overview:structural:results} provides key technical insights.
We present our approximation algorithms in \cref{section:overview:approximation:algorithms}.

\subsection{Structural Results}
\label{section:overview:structural:results}

\begin{figure}[t]
\centering
\begin{subfigure}{.24\textwidth}
\begin{tikzpicture}[scale=2,rotate=0,
		vertex/.style={draw,circle,inner sep=.0pt,minimum size=0.62cm},]
    \clip(-.22,-.3) rectangle (1.3,.3);
	\coordinate (cx) at (0,1);
	\coordinate (cy) at ($(1,0)+(0:1)$);
	\coordinate (cz) at ($(1,0)+(-90:1)$);
	\coordinate (cu) at (0,0);
	\coordinate (cv) at (1,0);
	\node[vertex,very thick] (u) at (cu) {$u$};
	\node[vertex,very thick] (v) at (cv) {$v$};
	\node[vertex] (x) at (cx) {$x$};
	\node[vertex] (y) at (cy) {$y$};
	\node[vertex] (z) at (cz) {$z$};
\newcommand{\vertexangle}{12}
\newcommand{\vertexdistance}{.21}
\newcommand{\midpointangle}{30}
\newcommand{\midpointdistance}{.08}
	\draw[very thick] (x)--(u);
	\draw[very thick] (y)--(v)--(z);
	\draw (u)--(v)--(y)--(z)--(v);
	\draw [ultra thick,dash pattern={on 6pt off 2pt}, rounded corners=.5] 
	([shift={(-90+\vertexangle:\vertexdistance)}]cx) arc (-90+\vertexangle:360-90-\vertexangle:\vertexdistance)--([shift={(90+\vertexangle:\vertexdistance)}]cu) arc (90+\vertexangle:450-\vertexangle:\vertexdistance)--cycle;
	\draw [ultra thick,dash pattern={on 6pt off 2pt}, rounded corners=.5] 
	([shift={(-180+\vertexangle:\vertexdistance)}]cy) arc (-180+\vertexangle:180-\vertexangle:\vertexdistance)--([shift={(\vertexangle:\vertexdistance)}]cv) arc (\vertexangle:270-\vertexangle:\vertexdistance)--([shift={(90+\vertexangle:\vertexdistance)}]cz) arc (90+\vertexangle:90-\vertexangle:\vertexdistance);
	\end{tikzpicture}
\caption{Staying out.}
\end{subfigure}
\hfill
\begin{subfigure}{.24\textwidth}
\begin{tikzpicture}[scale=2,rotate=0,
		vertex/.style={draw,circle,inner sep=.0pt,minimum size=0.62cm},]
    \clip(-.22,-.3) rectangle (1.3,.3);
	\coordinate (cx) at (0,1);
	\coordinate (cy) at ($(1,0)+(0:1)$);
	\coordinate (cz) at ($(1,0)+(-90:1)$);
	\coordinate (cu) at (0,0);
	\coordinate (cv) at (1,0);
	\coordinate (cm) at (.65,0);
	\node[vertex,very thick] (u) at (cu) {$u$};
	\node[vertex] (v) at (cv) {$v$};
	\node[vertex] (x) at (cx) {$x$};
	\node[vertex] (y) at (cy) {$y$};
	\node[vertex] (z) at (cz) {$z$};
	\node[minimum size=0.2cm,inner sep=0,outer sep=-1mm,circle,fill=black]  (m) at (cm) {};
\newcommand{\vertexangle}{12}
\newcommand{\vertexdistance}{.21}
\newcommand{\midpointangle}{30}
\newcommand{\midpointdistance}{.08}
	\draw[very thick] (x)--(u)--(m);
	\draw (u)--(v)--(y)--(z)--(v);
	\draw [ultra thick,dash pattern={on 6pt off 2pt}, rounded corners=.5] 
	([shift={(-90+\vertexangle:\vertexdistance)}]cx) arc (-90+\vertexangle:360-90-\vertexangle:\vertexdistance)--([shift={(90+\vertexangle:\vertexdistance)}]cu) arc (90+\vertexangle:360-\vertexangle:\vertexdistance)--([shift={(-180+\midpointangle:\midpointdistance)}]cm) arc (-180+\midpointangle:180-\midpointangle:\midpointdistance)--([shift={(\vertexangle:\vertexdistance)}]cu) arc (\vertexangle:90-\vertexangle:\vertexdistance)--cycle;
	\end{tikzpicture}
\caption{Peeking in.}
\end{subfigure}
\hfill
\begin{subfigure}{.24\textwidth}
\begin{tikzpicture}[scale=2,rotate=0,
		vertex/.style={draw,circle,inner sep=.0pt,minimum size=0.62cm},]
    \clip(-.22,-.3) rectangle (1.3,.3);
	\coordinate (cx) at (0,1);
	\coordinate (cy) at ($(1,0)+(0:1)$);
	\coordinate (cz) at ($(1,0)+(-90:1)$);
	\coordinate (cu) at (0,0);
	\coordinate (cv) at (1,0);
	\node[vertex,very thick] (u) at (cu) {$u$};
	\node[vertex,very thick] (v) at (cv) {$v$};
	\node[vertex] (x) at (cx) {$x$};
	\node[vertex] (y) at (cy) {$y$};
	\node[vertex] (z) at (cz) {$z$};
\newcommand{\vertexangle}{12}
\newcommand{\vertexdistance}{.21}
\newcommand{\midpointangle}{30}
\newcommand{\midpointdistance}{.08}
	\draw[very thick] (x)--(u)--(v);
	\draw[very thick] (y)--(v)--(z);
	\draw (u)--(v)--(y)--(z)--(v);
	\draw [ultra thick,dash pattern={on 6pt off 2pt}, rounded corners=.5] 
	([shift={(-90+\vertexangle:\vertexdistance)}]cx) arc (-90+\vertexangle:360-90-\vertexangle:\vertexdistance)--([shift={(90+\vertexangle:\vertexdistance)}]cu) arc (90+\vertexangle:360-\vertexangle:\vertexdistance)--([shift={(180+\vertexangle:\vertexdistance)}]cv) arc (180+\vertexangle:270-\vertexangle:\vertexdistance)--([shift={(90+\vertexangle:\vertexdistance)}]cz) arc (90+\vertexangle:90-\vertexangle:\vertexdistance);
	\draw [ultra thick,dash pattern={on 6pt off 2pt}, rounded corners=.5] 
	([shift={(90-\vertexangle:\vertexdistance)}]cz)--([shift={(-90+\vertexangle:\vertexdistance)}]cv) arc (-90+\vertexangle:-\vertexangle:\vertexdistance)--([shift={(-180+\vertexangle:\vertexdistance)}]cy);
	\end{tikzpicture}
\caption{Traversing once.}
\end{subfigure}
\hfill
\begin{subfigure}{.24\textwidth}
\begin{tikzpicture}[scale=2,rotate=0,
		vertex/.style={draw,circle,inner sep=.0pt,minimum size=0.62cm},]
    \clip(-.22,-.3) rectangle (1.3,.3);
	\coordinate (cx) at (0,1);
	\coordinate (cy) at ($(1,0)+(0:1)$);
	\coordinate (cz) at ($(1,0)+(-90:1)$);
	\coordinate (cu) at (0,0);
	\coordinate (cv) at (1,0);
	\node[vertex,very thick] (u) at (cu) {$u$};
	\node[vertex,very thick] (v) at (cv) {$v$};
	\node[vertex] (x) at (cx) {$x$};
	\node[vertex] (y) at (cy) {$y$};
	\node[vertex] (z) at (cz) {$z$};
\newcommand{\vertexangle}{12}
\newcommand{\vertexdistance}{.21}
\newcommand{\midpointangle}{30}
\newcommand{\midpointdistance}{.08}
	\draw[very thick] (x)--(u)--(v);
	\draw[very thick] (y)--(v)--(z);
	\draw (u)--(v)--(y)--(z)--(v);
	\draw [ultra thick,dash pattern={on 6pt off 2pt}, rounded corners=.5] 
	([shift={(-90+\vertexangle:\vertexdistance)}]cx) arc (-90+\vertexangle:360-90-\vertexangle:\vertexdistance)--([shift={(90+\vertexangle:\vertexdistance)}]cu) arc (90+\vertexangle:360-\vertexangle:\vertexdistance)--([shift={(180+\vertexangle:\vertexdistance)}]cv) arc (180+\vertexangle:270-\vertexangle:\vertexdistance)--([shift={(90+\vertexangle:\vertexdistance)}]cz) arc (90+\vertexangle:90-\vertexangle:\vertexdistance);
	\draw [ultra thick,dash pattern={on 6pt off 2pt}, rounded corners=.5] 
	([shift={(-180+\vertexangle:\vertexdistance)}]cy) arc (-180+\vertexangle:180-\vertexangle:\vertexdistance)--([shift={(\vertexangle:\vertexdistance)}]cv) arc (\vertexangle:180-\vertexangle:\vertexdistance)--([shift={(\vertexangle:\vertexdistance)}]cu) arc (\vertexangle:90-\vertexangle:\vertexdistance)--([shift={(-90+\vertexangle:\vertexdistance)}]cx) arc (-90+\vertexangle:360-90-\vertexangle:\vertexdistance);
	\end{tikzpicture}
\caption{Traversing twice.}
\end{subfigure}
\caption{The four ways a nice \deltatour defined by at least $3$ points can interact with an edge $uv$.}
\label{fig:fourbehaviors}
\end{figure}

Because TSP in the continuous model of graphs is studied in this article for the first time, 
we need to lay a substantial amount of groundwork. 
Due to the continuous nature of the problem, it is not clear a priori how to check if a solution is a valid \deltatour, 
or whether it is possible to compute a shortest \deltatour by a brute force search over a finite set of plausible tours. 
We clarify these issues in this section. 
While some of the arguments are intuitively easy to accept, 
the formal proofs (presented in \cref{section:structural}) are delicate with many corner cases to consider;
the reader might want to skip these proofs at first reading.

Sometimes, a \deltatour has to make U-turns inside edges to be shortest; see~\cref{fig:example1}.
Indeed, it can be checked that an optimal $1$-tour 
for the graph in~\cref{fig:example1} must look exactly as depicted.
Except for a single case, it is unnecessary for a tour to make more than one U-turn inside an edge.
Indeed, the only case where a shortest tour is forced to make two U-turns
in an edge is when the tour remains entirely within a single edge; see~\cref{fig:example2} for an example.
Note also that there are degenerate cases in which a shortest \deltatour consists of a single point.

However, unless a tour is completely contained in a single edge,
	we can see that there are only four reasonable ways for a \deltatour to interact with 
	the interior of any given edge, 
	illustrated in~\cref{fig:fourbehaviors}: 
\begin{description}
\item[(a)] completely avoiding the interior, 
\item[(b)] peeking into the edge from one side, 
\item[(c)] fully traversing the edge exactly once from one vertex to the other, or
\item[(d)] traversing the edge exactly twice.
\end{description}
\newcommand{\textDefNiceTour}{%
a tour $T=\seq{p_0&p_1&\dots&p_z}$ in a connected graph $G$ with $z \geq 3$ is \emph{nice} if
\begin{itemize}
\item $\{p_{i-1},p_i\}\cap V(G)\neq \emptyset$ for every $i \in [z]$
(i.e., $T$ has no two conecutive stopping points inside an edge),

\item
for every $i \in [z]$ with $p_i \notin V(G)$, we have $p_{i-1} = p_{(i+1)\bmod z}$\\
(i.e., whenever $T$ stops inside an edge, the previous and next stopping point are the same)
\item every edge $uv \in E(G)$ has at most one index $i \in [z]$ with
	$p_i \in \{ p(u,v,\lambda) \mid \lambda \in (0,1)\}$\\
(i.e., $T$ stops at most once inside any given edge),
\item if $T$ traverses an edge $uv \in E(G)$, there is no index $i \in [z]$ with $p_i \in \{ p(u,v,\lambda) \mid \lambda \in (0,1)\}$\\
(i.e., $T$ does not stop inside traversed edges), and 
\item every edge $uv \in E(G)$ is traversed at most twice by $T$.
\end{itemize}
}
We call a tour that restricts itself to this reasonable behavior a \emph{nice} tour. The following result allows us to restrict our search space to nice tours.

\begin{restatable}[Nice Tours]{linkedlemma}{tournice}
\label{lem:tournice}
Let $G$ be a connected graph. Further, let a tour~$T$ in $G$ be given.
Then, in polynomial time, we can compute a tour~$T'$ in $G$ with $\len(T')\leq \len(T)$, such that,
$T'$ is either nice or has at most two stopping points, and if $T$ is a
\deltatour for some $\delta \geq 0$, then so is $T'$.
\end{restatable}

Despite the continuous nature of \deltatourprob,
we show that we can actually study the problem in a discrete setting instead.
More precisely, we prove that
there is a nice shortest \deltatour~$T$ defined by points
	whose edge positions $\lambda$ come from a small explicitly defined set.
To show this, the idea is that there are only three scenarios for the edge position of a non-vertex stopping point $p$ of $T$.
\begin{description}
\item[1.] 
It has distance exactly $\delta$ to a vertex $u$.
An example is that $G$ is a long path with an endvertex $u$,
	and $p$ is the closest point of $P(T)$ to $u$.
Then $p$ has an edge position $\lambda$
	that is the fractional part of $\delta$.
\item[2.]
It has distance exactly $\delta$ to a half-integral point $p(u,v,\half)$.
An example is that $G$ is a long $\seq{w&w'}$-path with a triangle $\seq{u&v&w}$ attached to $w$,
	and $p$ is a closest point of $P(T)$ to $p(u,v,\half)$.
Then $p$ has an edge position $\lambda$
	which is the fractional part of $\delta+\frac{1}{2}$.
\item[3.]
It has distance exactly $2\delta$ to a vertex $u$.
An example is that $G$ contains a long $\seq{u&v}$-path $P$,
	$p$ is the closest point of $P(T) \setminus \{u\}$ on the path $P$ to $u$,
	and $T$ stops at $u$.
Then $p$ has an edge position $\lambda$
	which is the fractional part of $2\delta$.
\end{description}

Any \deltatour can efficiently be modified into one that is nice and whose stopping points have only such edge positions.
Our technical proof uses some theory of linear programming, in particular some results on the vertex cover polytope.

\begin{restatable}[Discretization Lemma]{linkedlemma}{TheoremDiscretization}
\label{TheoremDiscretization}
\label{lemma:discretization}

For every $\delta \geq 0$ and every connected graph $G$, there is a shortest \deltatour
	that is either nice or contains at most two stopping points and such that each stopping point of the tour can be described as $p(u,v,\lambda)$ with $\lambda \in S_\delta$
	where
$ S_\delta = \big \{0, \delta-\floor{\delta}, \half+\delta - \floor{\half+\delta}, 2\delta -\floor{2\delta} \big\}.$
\end{restatable}
As a consequence, we can find a shortest \deltatour by a brute-force algorithm. Using some related arguments, we can check whether a given tour actually is a \deltatour in polynomial time.

\subsection{Approximation Algorithms}
\label{section:overview:approximation:algorithms}

Here, we overview the approximation algorithms we design for different ranges of $\delta$.
Most of our algorithms follow a general paradigm; our approach is to design
	a collection of \textit{core} approximation algorithms for certain key values
	of $\delta$ and rely on one of the following two ideas to extrapolate 
	the approximation ratios we get to previous and subsequent intervals.
The first main idea uses the simple fact that a \deltatour is also \deltatour[(\delta+x)] for all $x>0$.
	Having an approximation algorithm for \deltatourprob,
	if we are able to reasonably bound the ratio between the lengths of a shortest $\delta$-tour and a shortest $(\delta+x)$-tour,
	we obtain an approximation algorithm for \deltatourprob[(\delta+x)]
	essentially for free.
The second main idea is complementary to the first. Namely, we show
	that we may also extend a given \deltatour[\delta]
	to obtain a \deltatour[(\delta-x)] where $x>0$.
	Again, having an approximation algorithm for the \deltatourprob,
	if we have a good bound on the total length of the extensions we add,
	we obtain an approximation algorithm for \deltatourprob[(\delta-x)].

\subparagraph{Approximation for \boldmath$\delta \in (0, 1/6]$.} 
The main idea is that a shortest Chinese Postman tour, that is, a tour which traverses every edge,
is a good approximation of a \deltatour.
Let us denote the length of a shortest \deltatour of a given graph by $\opttour$
	and the length of a shortest Chinese Postman tour by $\opt_{\text{CP}}$.
To bound the ratio ${\opt_{\text{CP}}}/{\opttour}$,
	we observe that there is
	a shortest \deltatour that, for every edge $uv$,
	either traverses $uv$ or
contains the segment of the form $up(u,v,\lambda)u$ for some $\lambda\in\{1-\delta,1-2\delta\}$.
We obtain a Chinese Postman tour
by replacing every such segment by a tour segment~$\seq{u&v&u}$.
This bounds ${\opt_{\text{CP}}}/{\opttour}$ by
	$1/(1-2\delta)$.
Hence, outputting a shortest Chinese Postman tour, which can be computed in polynomial time \cite{3134208}, yields an approximation ratio of 
	$1/(1-2\delta)$.

\begin{restatable}{linkedtheorem}{ThmApproxUbZeroSixth}
\label{ThmApproxUbZeroSixth}
\label{thm:approx:ub:zero_sixth}
	For every $\delta \in (0, 1/6]$,
	\deltatourprob admits a polynomial-time $1/(1-2\delta)$-approximation algorithm. 
\end{restatable}

\subparagraph{Approximation for \boldmath$\delta \in (1/6, 1/2)$.}
In this range, we
	rely on shortest \deltatour{}s that satisfy certain desirable discrete properties.
	In the following more precise description, we focus on the case $\delta \leq \frac{1}{4}$, the construction needing to be slightly altered if $\frac{1}{4}<\delta \leq \frac{1}{2}$.
Here, we prove the existence
of a nice shortest \deltatour~$T$ such that
\begin{description}
\item[(P1)] $T$ contains the tour segment~$\seq{u&p(u, v, 1-\delta)&u}$ for
every edge $uv \in E(G)$ incident to a leaf vertex $v$ (that is, $\deg_G(v) = 1$) and
\item[(P2)] for every edge $uv$ not incident to a leaf, either $T$ traverses $uv$ or the interaction of $T$ with $uv$ consists of one of the segments~$\seq{u&p(u, v, 1-2\delta)&u}$ or $\seq{v&p(v, u, 1-2\delta)&v}$.
\end{description}
We construct an auxiliary graph~$G'$ on the above listed points
	 with edge weights corresponding to their distance in $G$.
It turns out that \TSP tours in $G'$
	are in a one-to-one correspondence with \deltatour{}s in $G$
	satisfying properties (P1--P2).
More precisely, we prove that an $\alpha$-approximate \TSP tour~$T'$ of $G'$ can
be efficiently turned into a \deltatour~$T$ of $G$ of at most the same length which yields 
$\len(T) \leq \len(T') \leq \alpha \cdot \opttsp$, where
$\opttsp$ denotes the length of a shortest \TSP tour of $G'$.
Then, noting that a given \deltatour of $G$ satisfying properties (P1) and (P2) can be converted
to a TSP tour of $G'$ of at most the same length, we get that $T$ is a \deltatour with $\len(T) \leq \alpha \cdot \opttour$.
The first part, that is, proving that a \TSP tour~$T'$ of $G'$ can be turned into a $\delta$-tour of $G$ of the same length,
is based on the fact that there is a limited number of ways a reasonable TSP
tour interacts with the points corresponding to a certain edge. More precisely,
any TSP tour in $G'$ can easily be transformed into one which is not longer and
whose interaction with the points in any edge is in direct correspondence with
the interaction of a certain \deltatour with this edge in $G$.

This lets us transfer known positive approximation results for metric TSP to $\delta$-tour.
We may use the algorithm of Christofides~\cite{Christofides2022}, yielding
the following theorem.

\begin{restatable}{linkedtheorem}{ThmApproxUbSixthHalf}
\label{ThmApproxUbSixthHalf}
\label{thm:approx:ub:sixth_half}
	For every $\delta \in (1/6, 1/2)$,
	\deltatourprob admits a polynomial-time $1.5$-approximation algorithm. 
\end{restatable}

\subparagraph{Approximation for \boldmath$\delta = 1/2$.}
Even though the idea from the previous section still applies when $\delta = 1/2$,
interestingly, we obtain a better approximation ratio observing that the
problem further reduces to computing a graphic TSP tour on the non-leaf
vertices, which admits a $1.4$-approximation algorithm due to Seb\H o and Vygen~\cite{SeboV14}.

\begin{restatable}{linkedtheorem}{ThmApproxUbHalf}
\label{ThmApproxUbHalf}
\label{thm:approx:ub:half}
	\deltatourprob[1/2] admits a polynomial-time $1.4$-approximation algorithm.
\end{restatable}

\subparagraph{Approximation for \boldmath$\delta \in (1/2, 33/40)$.}
In this range, we show that computing a $\frac12$-tour
via \Cref{thm:approx:ub:half} is a good approximation of a \deltatour.
To that end, we characterize \deltatour{}s for $\delta \in [1/2, 1]$,
showing that, in particular, the existence of a
shortest \deltatour~$T$ such that one of the following conditions hold for
every edge $uv$. 

\begin{description}
\item[(P1)] $T$ stops at both $u$ and $v$, or
\item[(P2)] $T$ stops at one of the endpoints, say $u$, and
additionally stops at the point $p(u, v, \lambda)$ for some $\lambda \in [1-\delta, 1]$, or
$T$ stops at the two points $p(u, v, \lambda_1)$ and $p(x, v, \lambda_2)$ for
some $x \in N_G(v)$ where $\lambda_1 + \lambda_2 \geq 2-2\delta$, or
\item[(P3)] $T$ stops at neither $u$ nor $v$ but stops at two points
$p(x, v, \lambda_1)$ and $p(y, u, \lambda_2)$ for some $x \in N_G(v)$ and
$y \in N_G(u)$ where $\lambda_1 + \lambda_2 \geq 3-2\delta$.
\end{description}

Let $\opt_{1/2}$ and $\opttour$ be the lengths of a shortest
$\frac12$-tour and \deltatour in $G$, respectively. To bound
	the ratio ${\opt_{1/2}}/{\opttour}$, we observe that a given \deltatour~$T_\delta$
	can be transformed into a $\frac12$-tour~$T_{1/2}$ 
	by an appropriate replacement of every tour segment of $T_\delta$
	corresponding to one of the cases (P1) through (P3).

It can be shown that these modifications then result in a tour~$T_{1/2}$ stopping at every non-leaf
	vertex of $G$ and covering leaves by tour segments of length $1$,
	so $T_{1/2}$ is a $\frac12$-tour.
These modifications increase the tour length by at most a multiplicative factor of
	$\max\{\frac{1}{2(1-\delta)}, \frac{2}{3-2\delta}\} = 1/(2-2\delta)$,
	so we have the following theorem.

\begin{restatable}{linkedtheorem}{ThmApproxUbHalfThreeQuarters}
\label{ThmApproxUbHalfThreeQuarters}
\label{thm:approx:ub:half:threequarters}
	For every $\delta \in (1/2, 33/40)$,
	\deltatourprob admits a polynomial-time $1.4/(2-2\delta)$-approximation algorithm. 
\end{restatable}

\subparagraph{Approximation for \boldmath$\delta \in (33/40, 3/2)$.} 
Here we design a constant-factor approximation for computing a shortest $1$-tour and show how to tweak the obtained tour to obtain a constant-factor for any $\delta$ in the considered domain.
For $\delta > 1$, as in the previous range, due to a similar
characterization of \deltatour{}s, we can show an
$\alpha$-approximation algorithm for $1$-\tour to imply 
an $\frac{\alpha}{3-2\delta}$-approximation algorithm.
For $\delta < 1$, we show that starting from a $1$-tour and augmenting it with
some tour segments results in a \deltatour of a bounded length.
The $3$-approximation algorithm for a $1$-\tour works as follows.
We exploit a connection to the problem of computing a
	shortest \emph{vertex cover tour},
	which is a closed walk in a graph such that the vertices this tour stops at form a vertex cover.
This problem, introduced in~\cite{ArkinHH1993}, admits a $3$-approximation algorithm
	using linear programming (LP) techniques~\cite{KonemannKPS03}.
It is easy to see that a vertex cover tour forms a $1$-tour; however,
a shortest $1$-tour can be shorter than a shortest vertex cover tour (e.g., this is the case in \Cref{fig:example1}).
Thus, the \deltatour[1] we get from an arbitrary $3$-approximation for vertex cover tour
is in general not a $3$-approximate \deltatour[1].
Instead, we closely examine the LP formulated by Könemann
	et al.~\cite{KonemannKPS03}, showing the optimum for this LP to be a lower bound on the length of a \deltatour[1],
	which means that the vertex cover tours we get using this approach yield $3$-approximate \deltatour[1]{}s.

Given a connected graph~$G$, let $\mathcal{F}(G)$ be the set of subsets $F$ of $V(G)$ such
that both $G[F]$ and $G[V(G) \setminus F]$ induce at least one edge.
For a set $F \in \mathcal{F}(G)$,
let $C_G(F)$ denote the set of edges in $G$ with exactly one endpoint
in $F$.
The LP can then be formulated as follows:
\medskip
\begin{center}
\boxed{
\begin{aligned}
\begin{array}{ll@{}ll}
\text{Minimize}  & \displaystyle\sum\limits_{e \in E(G)} z_e &\\
\text{subject to}& \displaystyle\sum\limits_{e \in C_G(F)}z_e\geq 2 &\text{ for all $F \in \mathcal{F}(G)$}\text{ and }\\
0 \leq z_e \leq 2 &\text{ for all $e \in E(G)$.}
\end{array}
\end{aligned}
}
\end{center}
\medskip
Denoting the optimal objective value of the above LP defined for a fixed graph $G$ by $\optlp(G)$, the corollary below follows from \cite{KonemannKPS03}.

\begin{restatable}[Consequence of {\cite[Thms.~2 and 3]{KonemannKPS03}}]{linkedtheorem}{KoenemannCor}
\label{KoenemannCor}
\label{cor:koenemann_ovrvw}
Given a connected graph~$G$ of order $n$,
in polynomial time, we can compute a vertex cover tour~$T$
of $G$ with $\len(T) \leq 3 \cdot \optlp(G)$.
\end{restatable}

It remains to show that $\optlp$ lower-bounds
$\opttour[1]$, the length of a shortest \deltatour[1].
Let $T_{\deltatour[1]} = p_0 \dots p_k$
$p_{k} = p_0$
be a nice \deltatour[1] of $G$. For every edge $uv \in E(G)$, we
define
$\Lambda_{uv} \coloneqq \sum\limits_{i \in [k]\colon  P(p_{i-1},p_{i}) \subseteq P(u,v)} d_G(p_{i-1}, p_{i})$,
indicating how much the tour $T_{\deltatour[1]}$ spends inside every edge $uv$.
The vector $\left(\min(2,\Lambda_e)\right)_{e \in E(G)}$ can then be shown to be feasible for
the above LP.
We observe the length of $T_{\deltatour[1]}$ to be at least $\sum_{e \in E(G)} \Lambda_e$, 
	yielding  
	$\opttour[1] \geq \optlp(G)$ and with \Cref{cor:koenemann_ovrvw}, we obtain the following theorem.

\begin{restatable}{linkedtheorem}{ThmApproxUbOneThreeHalves}
\label{ThmApproxUbOneThreeHalves}
\label{thm:approx:ub:one:threehalves}
	For any $\delta \in [1,3/2)$,
	\deltatourprob admits a polynomial-time $3/(3-2\delta)$-approximation algorithm.
\end{restatable}

For the remaining range $\delta \in (33/40, 1)$,
our algorithm first uses \Cref{thm:approx:ub:one:threehalves}
to obtain a $3$-approximate $1$-tour~$T$ that is a vertex cover tour.
Then, for every vertex $v \not \in V(T)$, we choose an arbitrary neighbor $u$. Observe that $u \in V(T)$.
Then we extend $T$ into a tour~$T'$ by replacing an arbitrary occurrence of $u$ in $T$ by the segment
$\seq{u& p(u, v, 1-\delta)& u}$ if $v$ is a  leaf vertex and by  the segment
$\seq{u& p(u, v, 2(1-\delta))& u}$, otherwise.
As $T'$ fulfills the characterizing properties of a \deltatour,
$T'$ is a \deltatour.
To bound its length, using our characterization,
we observe that, given an arbitrary $\delta$-tour $T''$, each non-leaf vertex $v$ of $G$ can be associated to a segment of $T''$ of cost
at least $4(1-\delta)$ as $T''$ either stops at $v$ by traversing an edge,
incurring a cost of at least $1$, or makes two non-vertex stops that can be associated to $v$ with a total cost of at least $2(2-\delta)$.
The previous observation can be used to show that the \deltatour
we construct achieves an approximation ratio of $4$.

\begin{restatable}{linkedtheorem}{ThmApproxUbThreeQuartersOne}
\label{ThmApproxUbThreeQuartersOne}
\label{thm:approx:ub:threequarters:one}
For any $\delta \in [33/40,1)$, there is a polynomial-time $4$-approximation algorithm for \deltatourprob. 
\end{restatable}

\subparagraph{Approximation for \boldmath$\delta > 3/2$.}
Here we design $\polylog(n)$-approximation algorithms.
We consider two different settings:
	one where $\delta$ is fixed and another where $\delta$ is part of the input.
We show that each of the two problems can be reduced to an appropriate
dominating set problem in an auxiliary graph.
Recall that the discretization lemma (\cref{lemma:discretization}) shows, at a high-level, that there
is a shortest \deltatour~$T$ of $G$ whose stopping points on every edge come
from a constant-sized set. Let $P_\delta(G)$ be the set of all such
points in $G$.

In order to define our auxiliary graph, we first describe a collection $\mathcal{I}_{G,\delta}$ of edge segments of $G$. Namely, $\mathcal{I}_{G,\delta}$ is the collection of minimal edge segments each of whose endpoints is either a vertex of $V(G)$ or is of distance exactly $\delta$ to a point in $P_\delta(G)$ in $G$. This definition is suitable due to three properties of $\mathcal{I}_{G,\delta}$: 
\begin{description}
\item[(P1)] If $T$ is a \deltatour in $G$ whose stopping points are all contained in $P_{\delta}(G)$, then every $I \in \mathcal{I}_\delta(G)$ is fully covered by one stopping point of $T$, 
\item[(P2)] every point in $P(G)$ is contained in some $I \in \mathcal{I}_{G,\delta}$, and 
\item[(P3)] the number of segments in $\mathcal{I}_{G,\delta}$ is polynomial in $n$.
\end{description}
We are now ready to describe the auxiliary graph~$\Gamma(G,\delta)$.
We let $V(\Gamma(G,\delta))$ consist of $P_\delta(G)$ and a vertex $x_I$ for every $I \in \mathcal{I}_{G,\delta}$. We further let $E(\Gamma(G,\delta))$ contain edges such that $\Gamma(G,\delta)[P_\delta(G)]$ is a clique and let it contain an edge $px_I$ for $p \in P_\delta(G)$ and $I \in \mathcal{I}_{G,\delta}$ whenever $p$ covers all the points in $I$.
The main connection between $\delta$-tours in $G$ and dominating sets in $\Gamma(G,\delta)$
is due to the following lemma, which we algorithmically exploit in both
settings, when $\delta \geq 3/2$ is fixed and when $\delta$ is part of
the input.

\begin{restatable}{linkedlemma}{DomSetTour}
\label{DomSetTour}
\label{lem:domsettour}
Let $G$ be a graph and $\delta>1$. Further, let $T$ be a tour in $G$ whose stopping points are all in $P_{\delta}(G)$. Then $T$ is a $\delta$-tour in $G$ if and only if the stopping points of $T$ are a dominating set in $\Gamma(G,\delta)$.
\end{restatable}

\subparagraph{Approximation for Fixed \boldmath$\delta > 3/2$.}

By computing a dominating set $Y$ in the auxiliary graph~$\Gamma(G, \delta)$
using a standard $\log{n}$-approximation algorithm
and connecting it into a tour of length $\Oh(\delta |Y|)$, we obtain the
main result in this setting.

\begin{restatable}{linkedtheorem}{ThmApproxUbThreeHalvesLogN}
\label{ThmApproxUbThreeHalvesLogN}
\label{thm:approx:ub:threehalves:logn}
	For any $\delta \geq 3/2$,
	\deltatourprob admits a polynomial-time $\Oh(\log n)$-approximation algorithm.
\end{restatable}

\subparagraph{Approximation for \boldmath$\delta$ as Part of the Input.}
The approach from the previous section does not yield any non-trivial approximation guarantee in this setting
mainly because we get an additional factor of roughly $\delta$ when connecting the dominating set into a tour.
However, we are able to show that a $\polylog(n)$-approximation is attainable when $\delta$ is part of the input.
The algorithm for this is based on a reduction to another problem related to dominating sets. Namely, a dominating tree $U$ of a given graph~$H$ is a subgraph of $H$ which is a tree and such that $V(U)$ is a dominating set of $H$. Kutiel~\cite{Kutiel18} proves that given an edge-weighted graph~$H$, we can compute a dominating tree of $H$ of weight at most $\log^{3}{n}$ times the minimum weight of a  dominating tree of $H$.

In order to make use of this result, we now endow $E(\Gamma(G,\delta))$ with a weight function $w$. For all $p,p' \in P_\delta(G)$, we set $w(pp')=\dist_G(p,p')$, and all other edges get a very large weight. We now compute an approximate dominating tree $U$ of $\Gamma(G,\delta)$ with respect to $w$. By the definition of $w$, we obtain that $U$ does not contain any vertex of $V(\Gamma(G,\delta))-P_\delta(G)$. It follows that we can obtain a tour~$T$ from $U$ that stops at all points of $V(U)$ and whose weight is at most $2 w(U)$. By \Cref{lem:domsettour}, we obtain that $T$ is a $\delta$-tour in $G$. 

Finally, in order to determine the quality of $T$, consider a shortest \deltatour~$T^*$ in $G$. It follows from \Cref{lem:domsettour} that the set $P_{T^*}$ of points of $P_\delta(G)$ passed by $T^*$ forms a dominating set of $\Gamma(G,\delta)$. Further, we can easily find a tree in $\Gamma(G,\delta)$ spanning $P_{T^*}$ whose weight is at most the length of $T^*$. Hence, this tree is a dominating tree in $\Gamma(G,\delta)$, and \Cref{thm:approx:ub:threehalves:lognpthree} follows.

\section{Structural Results}\label{section:structural}
In this section, we provide a collection of structural results on shortest $\delta$-tours, 
which will be used several times throughout the main proofs of our algorithmic results 
in \cref{section:approximation-ub}. 
In \cref{section:general:observations}, we give some simple results preliminary to the 
ones in later sections.  In \cref{nicesec} we then show that, when finding shortest 
$\delta$-tours, we can restrict our search space to tours satisfying certain 
non-degeneracy conditions, more precisely nice tours. \Cref{sec:char} gives 
necessary and sufficient conditions for a given tour to be a $\delta$-tour. 
These conditions are of local nature and can be checked efficiently. 
Next, in Section~\ref{discsec}, we use the results from \cref{sec:char} to show that 
there are shortest $\delta$-tours satisfying some more restrictive conditions. 
Namely, we show that we can find a shortest $\delta$-tour whose stopping points all 
come from a small set. Finally, we give in \cref{algosec} some simple algorithmic 
corollaries of the previous results. While the results in \cref{discsec} rely on a 
connection to the vertex cover polytope, the proofs in the remaining sections are 
conceptually transparent albeit somewhat technical.
\subsection{General Observations}
\label{section:general:observations}
We first give a well-known result of Euler, that will be used several times throughout the article.
A multigraph~$G$ is called {\it Eulerian} if $\deg_G(v)$ is even for all $v \in V(G)$. 
An {\it Euler tour} of $G$ is a sequence $v_0,\ldots, v_z$ of vertices in $G$
such that for each pair $u,v \in V(G)$, the number of indices 
$i \in [z]$ with $\{v_{i-1},v_i\}=\{u,v\}$ is exactly the number of edges in $E(G)$ linking $u$ and~$v$.
\begin{proposition}[Euler's Theorem]\label{euler}
A multigraph~$G$ admits an Euler tour if and only if it is connected and Eulerian. Moreover, if an Euler tour exists, it can be computed in polynomial time.
\end{proposition} 
We begin with two simple results.
 As a little warm-up, we start with a simple observation giving an upper bound on the maximum length of $\delta$-tours if $\delta$ is not too small.
\begin{proposition}\label{trgvftzuh}
Let $G$ be a connected graph and $\delta \geq 1/2$ a constant. Then $G$ admits a \deltatour of length at most $2n-2$ and this tour can be found in polynomial time. 
\end{proposition}
\begin{proof}
Let $U$ be an arbitrary spanning tree of $G$ and let $U'$ be obtained from $U$ by doubling every edge. Observe that $\deg_{U'}(v)=2\deg_U(v)$ is even for all $v \in V(G)$, hence by \cref{euler}, there is an Euler tour~$T$ of $U'$. Observe that $\len(T)=\abs{E(U')}=2n-2$. Further, let $p=(u,e,\lambda) \in P(G)$ for some $e=uv \in E(G)$ and $\lambda \in [0,1]$. As $T$ stops at $u$ and $v$, we have  $\dist_G(p,T)\leq \min \{\dist_G(p,u),\dist_G(p,v)\}=\min\{\lambda,1-\lambda\}\leq \frac{1}{2}(\lambda+(1-\lambda))=\frac{1}{2}\leq \delta$. Hence $T$ is a \deltatour in $G$. It follows from Proposition~\ref{euler} that $T$ can be computed in polynomial time.
\end{proof}

We next give one proposition that will be helpful in several places throughout this section.

\begin{proposition}\label{nearstop}
Let $G$ be a connected graph, $T$ a tour in $G$ and $p \in P(G)$ a point which is not passed by $T$.
Then $\dist_G(p,T)=\dist_G(p,q)$ for some point $q \in P(G)$ at which $T$ stops.
\end{proposition}
\begin{proof}
Suppose otherwise, so there is some $q^*=(u,v,\lambda)\in P(G)$ which is passed by $T$ and such that $\dist_G(p,q)>\dist_G(p,q^*)$ for every  point $q \in P(G)$ which $T$ stops at. As $T$ passes but does not stop at $q^*$, we obtain that $\lambda \in (0,1)$ and that there are $\lambda_0,\lambda_1 \in [0,1]$ such that $\lambda_0<\lambda<\lambda_1$, the point $(u,v,\lambda')$ is passed by $T$ for all $\lambda'\in [\lambda_0,\lambda_1]$ and $T$ stops at $q_0$ and $q_1$ where $q_i=(u,v,\lambda_i)$ for $i \in \{0,1\}$. First suppose that $p$ is also on the edge $uv$, so $p=(u,v,\lambda_p)$ for some $\lambda_p \in [0,1]$. As $p$ is not passed by $T$ and by symmetry, we may suppose that $\lambda_p<\lambda_0$. This yields $\dist_G(p,q_0)=\lambda_0-\lambda_p<\lambda-\lambda_p=\dist_G(p,q^*)$, a contradiction to the assumption. Now suppose that $p$ is not on $uv$. Then, by symmetry, we may suppose that $\dist_G(p,q^*)=\dist_G(p,u)+\dist_G(u,q^*)$. It follows that $\dist_G(p,q_0)\leq \dist_G(p,u)+\dist_G(u,q_0)=\dist_G(p,u)+\lambda_0<\dist_G(p,u)+\lambda'=\dist_G(p,u)+\dist_G(u,q^*)<\dist_G(p,q^*)$, again a contradiction to the assumption.
\end{proof}

\subsection{Making a Tour Nice}
\label{nicesec}

This section proves that in all relevant cases, there is a nice shortest \deltatour.
Recall that \textDefNiceTour

We prove \cref{lem:tournice}, which we restate here:

\tournice*\label\thisthm

We prove \cref{lem:tournice} by showing how to transform any given tour lacking any of the niceness properties stated in \cref{lem:tournice} into one of shorter or equal length that has fewer stopping points. 
We do this for all niceness properties listed.
\begin{proposition}\label{no2mp}
Let $G$ be a connected graph and $T=\seq{p_0&\ldots&p_z}$ be a tour in $G$ for some $z \geq 3$ such that  $\{p_{i-1},p_i\}\cap V(G)= \emptyset$ for some $i \in [z]$. Then we can compute in polynomial time a tour~$T'$ in $G$ with $\len(T')\leq \len(T)$ and $\alpha(T')<\alpha(T)$ and such that if $T$ is a \deltatour for some $\delta \geq 0$, then so is $T'$.
\end{proposition}
\begin{proof}
By symmetry, we may suppose that there is some $k \geq 1$ such that $\seq{p_0&\ldots&p_k}$ is a segment of $T$ that contains the maximum number of points among all those that are disjoint from $V(G)$.
Hence there are some $u,v\in V(G)$ with $uv \in E(G)$ and for all $i \in [k]\cup \{0\}$ a $\lambda_i\in (0,1)$ such that $p_i=p(u,v,\lambda_i)$.

First suppose that $k=z$; that is, the entire tour is contained inside the edge $uv$.  We now choose some $i_1,i_2 \in [z]$ such that $\lambda_{i_1}=\min\{\lambda_i:i \in [z]\}$ and $\lambda_{i_2}=\max\{\lambda_i:i \in [z]\}$. Let $T'=p_{i_1}p_{i_2}p_{i_1}$. Clearly, $T'$ can be computed in polynomial time. Further, we have $\len(T')=2\dist_G(p_{i_1},p_{i_2})\leq \len(T)$ and $\alpha(T')=2<3\leq \alpha(T)$.

Now consider some $\delta \geq 0$ such that $T$ is a \deltatour. If a point $p \in P(G)$ is passed by $T$, then we have $p=(u,v,\lambda)$ for some $\lambda \in [\lambda_1,\lambda_2]$ and hence $p$ is also passed by $T'$. It follows that $T'$ is a \deltatour in $G$. 

Now suppose that $k<z$. 
As $p_0=p_z$ and by the maximality of $k$, we obtain that $k<z-1$.
By the choice of the segment~$\seq{p_0&\ldots&p_k}$ and as $T$ is a tour, we obtain that $\{p_{z-1},p_{k+1}\}\subseteq \{u,v\}$.

First suppose that $p_{z-1}=p_{k+1}$, say $p_{z-1}=p_{k+1}=u$. We then choose some $i_0 \in [k]\cup \{0\}$ such that $\lambda_{i_0}=\max\{\lambda_i \mid i \in [z]\}$.  Now let $T'=p_{i_0}p_{k+1}\ldots p_{z-1}p_{i_0}$. Clearly, $T'$ can be computed in polynomial time.  Further, we have $\len(T')\leq \len(T)$ and $\alpha(T')=z-k<z=\alpha(T)$. 

Now consider some $\delta \geq 0$ such that $T$ is a \deltatour in $G$. If a point $p \in P(G)$ not on $uv$ is passed by $T$, then it clearly is also passed by $T'$. If a point $p \in P(G)$ on $uv$ is passed by $T$, then either $p$ is also passed by $T'$ by construction or we have $p=(u,v,\lambda)$ for some $\lambda \in [0,\lambda_{i_0}]$ and hence $p$ is also passed by $T'$. It follows that $T'$ is a \deltatour in $G$.

Now suppose that $\{p_{z-1},p_{k+1}\}=\{u,v\}$, say $p_{z-1}=u$ and $p_{k+1}=v$. Let $T'=p_{z-1}p_{k+1}\ldots p_{z-1}$. Clearly, $T'$ can be computed in polynomial time. Further, we have $\len(T')\leq \len(T)$ and $\alpha(T')=z-(k+1)<z=\alpha(T)$. 

Now consider some $\delta \geq 0$ such that $T$ is a \deltatour in $G$. If a point $p \in P(G)$ which is not on $uv$ is passed by $T$, then it clearly is also passed by $T'$. If a point $p \in P(G)$ on $uv$ is passed by $T$, then we have $p=(u,v,\lambda)$ for some $\lambda \in [0,1]$ and hence $p$ is also passed by $T'$. It follows that $T'$ is a \deltatour in $G$. 
\end{proof}

Note that applying \cref{no2mp} repeatedly to a given tour yields a tour that is not longer but either contains at most two stopping points or satisfies the first niceness property. 

\begin{proposition}\label{nointermediate}
Let $G$ be a connected graph, and $T=\seq{p_0&\ldots&p_z}$ be a tour in $G$ for some $z \geq 3$ such that either there is some $i \in [z-1]$ with $p_i \notin V(G)$ and $p_{i-1}\neq p_{i+1}$ or  $p_0 \notin V(G)$ and  $p_{1}\neq p_{z-1}$. Then, in polynomial time, we can compute a tour~$T'$ in $G$ with $\alpha(T')<\alpha(T)$ and $\len(T')\leq \len(T)$ and such that if $T$ is a \deltatour for some $\delta \geq 0$, then so is $T'$.
\end{proposition}
\begin{proof}
By symmetry, we may suppose that $p_0 \notin V(G)$ and  $p_{1}\neq p_{z-1}$. Further, by Proposition~\ref{no2mp}, we may suppose that $p_{z-1}=u, p_0=p(u,v,\lambda)$, and $p_1=v$ for some $uv \in E(G)$ and some $\lambda \in (0,1)$. Let $T'=\seq{p_1&\ldots&p_{z-1}&p_1}$. Clearly, $T'$ can be computed in polynomial time. Further, we have $\len(T')\leq \len(T)$ and $\alpha(T')=\alpha(T)-1<\alpha(T)$. 

Now consider some $\delta \geq 0$ such that $T$ is a \deltatour. If a point $p \in P(G)$ not on $uv$ is passed by $T$, then it clearly is also passed by $T'$. If a point $p \in P(G)$ on $uv$ is passed by $T$, then, as $T'$ traverses $uv$, it is also passed by $T'$. It follows that $T'$ is a \deltatour in $G$.
\end{proof}
We now show how to transform a tour into one with only two stopping points or one that has the third niceness property without increasing its length.
\begin{proposition}\label{atmost1}
Let $G$ be a connected graph and $T=\seq{p_0&\ldots&p_z}$ be a tour in $G$ for some $z \geq 3$ that stops at least twice in the interior of $uv$ for some $uv \in E(G)$.  Then, in polynomial time, we can compute a tour~$T'$ in $G$ with $\alpha(T')<\alpha(T)$ and $\len(T')\leq \len(T)$ and such that if $T$ is a \deltatour for some $\delta \geq 0$, then so is $T'$.
\end{proposition}
\begin{proof}
Let $uv \in E(G)$ and suppose, by symmetry, that there are distinct $i,j \in [z-1]$ such that $p_i=p(u,v,\lambda_1)$ and $p_j=p(u,v,\lambda_2)$ for some $uv \in E(G)$ and some $\lambda_1,\lambda_2 \in (0,1)$. By Propositions~\ref{no2mp} and~\ref{nointermediate}, we obtain that $p_{i-1}=p_{i+1}, p_{j-1}=p_{j+1}$ and $\{p_{i-1},p_{j-1}\}\subseteq \{u,v\}$.

First suppose that $p_{i-1}=p_{j-1}$, say $T$ contains the two segments~$\seq{u&p(u,v,\lambda_1)&u}$ and $\seq{u&p(u,v,\lambda_2)&u}$.
Further, by symmetry, we may suppose that $\lambda_2\geq \lambda_1$.
We now obtain $T'$ by replacing the segment~$\seq{u&p(u,v,\lambda_1)&u}$ by $u$.
Clearly, $T'$ can be computed in polynomial time. Next, we have $\alpha(T')<\alpha(T)$ and $\len(T')\leq \len(T)$. Now consider some $\delta \geq 0$ such that $T$ is a \deltatour in $G$. As every point that is passed by $T$ is also passed by $T'$, we obtain that $T'$ is a \deltatour in $G$.

It hence suffices to consider the case that $p_{i-1}\neq p_{j-1}$, say $p_{i-1}=u$ and $p_{j-1}=v$. We obtain that $T$ contains the segments~$\seq{u&p_i&u}$ and $\seq{v&p_j&v}$.
 By the above argument, we have that $p_i$ and $p_j$ are the only points in the interior of $uv$ that $T$ stops at.

First suppose that $\lambda_1>\lambda_2$. Then, let $T'$ be obtained from replacing the segment~$\seq{u&p_i&u}$ by $\seq{u&v&u}$ and the segment~$\seq{v&p_j&v}$ by $v$. Clearly, $T'$ can be computed in polynomial time.  Further, we have $\len(T')=\len(T)-2(\lambda_1+(1-\lambda_2))+2<\len(T)$ and $\alpha(T')\leq\alpha(T)$.

Now consider some $\delta \geq 0$ such that $T$ is a \deltatour in $G$. Observe that  every point passed by $T$ is also passed by $T'$, so $T'$ is a \deltatour.

We may hence suppose that $\lambda_1\leq \lambda_2$.  Now let $T'$ be obtained from $T$ by replacing the segment  $\seq{u&p_i&u}$ by $\seq{u&p(u,v,1-(\lambda_2-\lambda_1))&u}$ and replacing the segment~$\seq{v&p_j&v}$ by $v$. Clearly, $T'$ can be computed in polynomial time. As $\delta \geq \frac{1}{2}(\lambda_2-\lambda_1)$, we have that $\len(T')-\len(T)=2(1- (\lambda_2-\lambda_1))-2(\lambda_1+(1-\lambda_2))= 0$. Further, it holds $\alpha(T')<\alpha(T)$.

Now consider some $\delta \geq 0$ such that $T$ is a \deltatour.
Let $q_0=p(u,v,\lambda_1+\frac{1}{2}(\lambda_2-\lambda_1))$. As $T$ is a \deltatour and $p_i$ and $p_j$ are the only points in the interior of $uv$ stopped at by $T$, we obtain that $\delta\geq \dist_G(q_0,T)=\min\{\dist_G(q_0,p_i),\dist_G(q_0,p_j)\}=\frac{1}{2}(\lambda_2-\lambda_1)$. In order to see that $T'$ is a \deltatour, consider some $q \in P(G)$ not passed by $T'$. As $T$ is a \deltatour and by Proposition~\ref{nearstop}, there is a point $p$ stopped at by $T$ with $\dist_G(p,q)\leq \delta$. If $p \notin \{p_i,p_j\}$, we obtain $\dist_G(q,T')\leq \dist_G(q,p)\leq \delta$. We may hence suppose that $p \in \{p_i,p_j\}$. Next, if $q$ is not on $uv$, we obtain $\delta\geq \dist_G(q,p)=\min\{\dist_G(q,u)+\dist_G(u,p), \dist_G(q,v)+\dist_G(v,p)\}\geq \min\{\dist_G(q,u), \dist_G(q,v))\}\geq \dist_G(q,T')$.  We may hence suppose that $q=(u,v,\lambda)$ for some $\lambda \in (0,1)$. As $q$ is not passed by $T'$  and by construction,  we obtain $\lambda \geq 1-(\lambda_2-\lambda_1)$. This yields $\dist_G(q,T')\leq \min\{\dist_G(q,p(u,v,1-(\lambda_2-\lambda_1))),\dist_G(q,v)\}=\min\{\lambda-(1-(\lambda_2-\lambda_1)),1-\lambda\}\leq\min\{\lambda-(1-2\delta),1-\lambda\}\leq \frac{1}{2}((\lambda-(1-2\delta))+(1-\lambda))= \delta$. Hence $T'$ is a \deltatour in $G$. 
\end{proof}

\begin{proposition}\label{atmost2}
Let $G$ be a connected graph and $T=\seq{p_0&\ldots&p_z}$ be a tour in $G$ that stops at an edge $uv \in E(G)$ that it traverses.  Then, in polynomial time, we can compute a tour~$T'$ in $G$ with $\alpha(T')<\alpha(T)$ and $\len(T')<\len(T)$ and such that if $T$ is a \deltatour for some $\delta \geq 0$, then so is $T'$.
\end{proposition}
\begin{proof}
By Propositions~\ref{no2mp}~and~\ref{nointermediate} and by symmetry, we may suppose that $T$ contains the segment~$\seq{u&p(u,v,\lambda)&u}$ for some $\lambda \in (0,1)$. Let $T'$ be obtained by replacing $\seq{u&p(u,v,\lambda)&u}$ by $u$. Clearly, $T'$ can be computed in polynomial time.  Further, we clearly have $\len(T')<\len(T)$ and $\alpha(T')<\alpha(T)$.

Now consider some $\delta \geq 0$ such that $T$ is a \deltatour. As $T$ traverses $uv$, so does $T'$. It follows that every point passed by $T$ is also passed by $T'$. Hence $T'$ is a \deltatour.
\end{proof}

\begin{proposition}\label{edgetwice}
Let $G$ be a connected graph and $T=\seq{p_0&\ldots&p_z}$ a tour in $G$ that traverses an edge $uv \in E(G)$ at least 3 times. Then, in polynomial time, we can compute a tour~$T'$ in $G$ with $\alpha(T')<\alpha(T)$ and $\len(T') < \len(T)$ and such that if $T$ is a \deltatour for some $\delta \geq 0$, then so is $T'$.
\end{proposition}
\begin{proof}
We create a multigraph~$H$ in the following way: We let $V(H)$ consist of the points in $P(G)$ which were stopped at at least once by $T$ and we let $E(H)$ contain an edge linking two points $p,p' \in V(H)$ for every $i \in [z]$ with $\{p_{i-1},p_i\}=\{p,p'\}$. Observe that $H$ can be computed in polynomial time. By construction, $T$ is an Euler tour in $H$, so by \cref{euler} and as $H$ clearly does not contain isolated vertices, we have that $H$ is Eulerian and connected. Further, by assumption and construction, we have that $H$ contains at least three parallel edges linking $u$ and $v$. Now let $H'$ be obtained from $H$ by deleting two of the edges linking $u$ and $v$. We have $\deg_{H'}(u)=\deg_H(u)-2, \deg_{H'}(v)=\deg_H(v)-2$, and $\deg_{H'}(w)=\deg_H(w)$ for all $w \in V(H)-\{u,v\}$. Hence $H'$ is Eulerian. Further, as $H$ is connected, so is $H'$. It hence follows by Proposition~\ref{euler} that $H'$ has an Euler tour~$T'=p'_0\ldots p'_{z-2}$ and $T'$ can be computed in polynomial time.  Further, we have 
\begin{align*}
\len(T')&=\sum_{i \in [z-2]}\dist_G(p'_{i-1},p'_i)\\
&=\sum_{pp'\in E(H')}\dist_G(p,p')\\
&<\sum_{pp'\in E(H)}\dist_G(p,p')\\
&=\sum_{i \in [z]}\dist_G(p_{i-1},p_i)\\
&=\len(T).
\end{align*}
Next, we have $\alpha(T')=z-2<z=\alpha(T)$.

Finally, consider some $\delta \geq 0$ such that $T$ is a \deltatour in $G$. We will show that $T'$ is a \deltatour 
in $G$. First observe that $T'$ is a tour in $G$. Let $q \in P(G)$. As $T$ is a \deltatour in $G$, there is some 
$p^* \in P(G)$ that $T$ passes and that satisfies $\dist_G(q,p^*)\leq \delta$. By construction, we obtain that $T'$ 
also passes $p^*$. This yields $\dist_G(q,T')\leq \dist_G(q,p^*)\leq \delta$. Hence $T'$ is a \deltatour in~$G$.
\end{proof}

We are now ready to prove \cref{lem:tournice}.
\begin{proof}
Let a tour~$T$ in $G$ be given. As long as $T$ is not nice, we can recursively apply one of Propositions~\ref{no2mp},~\ref{nointermediate},~\ref{atmost1},~\ref{atmost2}, and~\ref{edgetwice} and update $T$. In the end of this procedure, we obtain a tour~$T'$ which is either nice or contains at most two stopping points. Also, for any $\delta \geq 0$ such that $T$ is a \deltatour, we have that $T'$ is a \deltatour. Further, in every iteration, the modifications can be executed in polynomial time. Finally, as the discrete length of the tour is decreased in every iteration, we obtain that the algorithm terminates after at most $\alpha(T)$ iterations, so its total running time is polynomial. 
\end{proof}

\subsection{Characterization}\label{sec:char}
In this section, we give a collection of results allowing to decide whether a given tour is actually a \deltatour for a certain $\delta$.
More precisely, we describe when a tour covers all points on a certain edge.
The result is cut into three parts: one for the case that the tour stops at no vertex of the edge, one for the case that it stops at exactly one vertex of the edge and for the case that it stops at both vertices of the edge.
In the following, we say that a tour {\it $\delta$-covers} an edge if it covers all the points on the edge. We speak about {\it covering} when $\delta$ is clear from the context.
We first deal with the case that the tour stops at no vertex of the edge. 
\begin{lemma}\label{nostopchar}
Let $G$ be a connected graph, $\delta\geq 0$ a real, $T$ a tour in $G$ and $x_1x_2 \in E(G)$ such that $T$ stops at none of $x_1$ and $x_2$. Then $T$ covers $x_1x_2$ if and only if one of the following holds:
\begin{enumerate}[(i)]
\item for $i \in [2]$, there are stopping points $p_i=p(u_i,v_i,\lambda_i)$ of $T$ with $\lambda_i \in [0,1)$ 
such that $\lambda_1+\lambda_2\geq \dist_G(x_1,v_1)+\dist_G(x_2,v_2)+3-2 \delta$ holds,
\item $T$ stops at points $p(x_1,x_2,\lambda_1)$ and $p(x_1,x_2,\lambda_2)$ for some $\lambda_1 \in (0,\delta]$ and some $\lambda_2 \in [1-\delta,1)$.
\end{enumerate}
\end{lemma}
\begin{proof}
First suppose that $(i)$ holds. Consider $p=p(x_1,x_2,\lambda)$ for some $\lambda \in [0,1]$. 
We obtain \begin{align*}
\dist_G(T,p)&\leq \min\{\dist_G(p,x_1)+\dist_G(x_1,T),\dist_G(p,x_2)+\dist_G(x_2,T)\}\\
&\leq \frac{1}{2}(\dist_G(p,x_1)+\dist_G(x_1,T)+\dist_G(p,x_2)+\dist_G(x_2,T))\\
&\leq \frac{1}{2}(1+\dist_G(x_1,p_1)+\dist_G(x_2,p_2))\\
&\leq \frac{1}{2}(1+\dist_G(x_1,v_1)+\dist_G(v_1,p_1)+\dist_G(x_2,v_2)+\dist_G(v_2,p_2))\\
&= \frac{1}{2}(1+\dist_G(x_1,v_1)+\dist_G(x_2,v_2)+(1-\lambda_1)+(1-\lambda_2))\\
&\leq  \frac{1}{2}(3+(\lambda_1+\lambda_2+2 \delta-3)-(\lambda_1+\lambda_2))\\
&=\delta.
\end{align*}
Hence $x_1x_2$ is covered by $T$.

Now suppose that $(ii)$ holds. Consider $p=p(x_1,x_2,\lambda)$ for some $\lambda \in [0,1]$. If $\lambda_1 \leq \lambda \leq \lambda_2$, then, as $T$ stops at none of $x_1$ and $x_2$, we obtain that $T$ passes $p$. Next suppose that $\lambda<\lambda_1$. Then we have $\dist_G(p,T)\leq \dist_G(p,p(x_1,x_2,\lambda_1))=\lambda_1-\lambda\leq \lambda_1\leq \delta$. A similar argument shows that $\dist_G(p,T)\leq \delta$ if $\lambda>\lambda_2$. It follows that $T$ covers $x_1x_2$.

Now suppose that $T$ covers $x_1x_2$. First suppose that $T$ does not stop at any point on $x_1x_2$ and for $i \in 
[2]$, let $p_i=p(u_i,v_i,\lambda_i)$ be a stopping point of $T$ with $\dist_G(p_i,x_i)=\dist_G(x_i,T)$. As $T$ 
stops at neither $x_1$ nor $x_2$, we may choose the labeling of $p_1$ and $p_2$ such that 
$\dist_G(x_i,p_i)=\dist_G(x_i,v_i)+\dist_G(p_i,v_i)$ for $i \in [2]$ and $\max\{\lambda_1,\lambda_2\}<1$ holds. Suppose for the sake of a contradiction that $\lambda_1+\lambda_2< 
\dist_G(x_1,v_1)+\dist_G(x_2,v_2)+3-2 \delta$. Consider the point 
$p=p(x_1,x_2,\frac{1}{2}(1+\dist_G(x_2,p_2)-\dist_G(x_1,p_1))$. By the choice of $p_1$ and $p_2$ and as $x_1x_2 \in 
E(G)$, we have
\begin{align*}
1+\dist_G(x_2,p_2)-\dist_G(x_1,p_1)&\geq 1+\dist_G(x_2,p_2)-\dist_G(x_1,p_2)\\
&\geq 1+\dist_G(x_2,p_2)-(\dist_G(x_1,x_2)+\dist_G(x_2,p_2))\\
&=1-\dist_G(x_1,x_2)\\
&=0.
\end{align*}

We further have
\begin{align*}
1+\dist_G(x_2,p_2)-\dist_G(x_1,p_1)&\leq 1+\dist_G(x_2,p_1)-\dist_G(x_1,p_1)\\
&\leq 1+(\dist_G(x_1,p_1)+\dist_G(x_1,x_2))-\dist_G(x_1,p_1)\\
&=1+\dist_G(x_1,x_2)\\
&=2.
\end{align*}

Hence $p$ is well-defined.

As $T$ stops at no point on $x_1x_2$ and by assumption, we obtain
\begin{align*}
\dist_G(p,T)&=\min\{\dist_G(p,x_1)+\dist_G(x_1,T),\dist_G(p,x_2)+\dist_G(x_2,T)\}\\
&=\min\{(1+\dist_G(x_2,p_2)-\dist_G(x_1,p_1))/2+\dist_G(x_1,p_1),\\
&\qquad1-((1+\dist_G(x_2,p_2)-\dist_G(x_1,p_1)))/2+\dist_G(x_2,p_2)\}\\
&=(1+\dist_G(x_1,p_1)+\dist_G(x_2,p_2))/2\\
&=(1+(\dist_G(x_1,v_1)+1-\lambda_1)+(\dist_G(x_2,v_2)+1-\lambda_2))/2\\
&=(3+\dist_G(x_1,v_1)+\dist_G(x_2,v_2)-(\lambda_1+\lambda_2))/2\\
&>\delta,
\end{align*}
a contradiction to $p$ being covered by $T$. We hence obtain $\lambda_1+\lambda_2< \dist_G(x_1,v_1)+\dist_G(x_2,v_2)+3-2 \delta$ and so $(i)$ holds.

Finally suppose that $T$ stops at some point on $x_1x_2$. Let $p_1=p(x_1,x_2,\lambda_1)$ and $p_2=p(x_1,x_2,\lambda_2)$ be stopping points of $T$ on $x_1x_2$ which are chosen so that $\lambda_1$ is minimized and $\lambda_2$ is maximized. As $T$ stops at none of $x_1$ and $x_2$ and $T$ is a tour, we obtain that all stopping points of $T$ are on $x_1x_2$. As $T$ covers $x_1x_2$, this yields $\delta\geq \dist_G(x_1,T)=\dist_G(x_1,p_1)=\lambda_1$. A similar argument shows that $\lambda_2 \leq 1-\delta$. Hence $(ii)$ holds.
\end{proof}
We next deal with the case that the tour stops at exactly one vertex of the edge.
\begin{lemma}\label{char1stop}
Let $G$ be a connected graph, $\delta\geq 0$ a real, $T$ a tour in $G$ and $x_1x_2 \in E(G)$ such that $T$ stops at $x_1$ but not at $x_2$. Then $T$ covers $x_1x_2$ if and only if one of the following holds:
\begin{enumerate}[(i)]
\item There are stopping points $p_1=p(x_1,x_2,\lambda_1)$ and $p_2=(v,x_2,\lambda_2)$ of $T$ with $\lambda_1,\lambda_2 \in [0,1)$ and $v \in N_G(x_2)$ such that  $\lambda_1+\lambda_2\geq 2-2 \delta$ holds.
\item $T$ stops at the point $p(x_1,x_2,\lambda)$ for some $\lambda\in [1-\delta,1]$.
\end{enumerate}
\end{lemma}
\begin{proof}
First suppose that $(i)$ holds. Consider $p=p(x_1,x_2,\lambda)$ for some $\lambda \in [0,1]$. If $\lambda\leq \lambda_1$, then, as $T$ does not stop at $x_2$, we have that $T$ passes $p$. If $\lambda>\lambda_1$,
we obtain \begin{align*}
\dist_G(T,p)&\leq \min\{\dist_G(p,p_1),\dist_G(p,x_2)+\dist_G(x_2,p_2)\}\\
&\leq \frac{1}{2}(\dist_G(p,p_1)+\dist_G(p,x_2)+\dist_G(x_2,p_2))\\
&= \frac{1}{2}((\lambda-\lambda_1)+(1-\lambda)+(1-\lambda_2))\\
&= \frac{1}{2}(2-(\lambda_1+\lambda_2))\\
&\leq\delta.
\end{align*}
Hence $x_1x_2$ is covered by $T$.

Next suppose that $(ii)$ holds. Consider $p=p(x_1,x_2,\lambda)$ for some $\lambda \in [0,1]$. If $\lambda\leq\lambda_1$, then, as $T$ does not stop at $x_2$, we have that $T$ passes $p$. If $\lambda>\lambda_1$, we have $\dist_G(p,T)\leq \dist_G(p,p_1)=\lambda-\lambda_1\leq 1-\lambda \leq \delta$. Hence $T$ covers $x_1x_2$.

Now suppose that $T$ covers $x_1x_2$. Let $\lambda_1$ be the greatest real such that $T$ stops at $p(x_1,x_2,\lambda_1)$. If $\lambda_1 \geq 1-\delta$, then $(ii)$ holds, so there is nothing to prove. We may hence suppose that $\lambda_1<1-\delta$, so in particular $\delta<1$. Now let $p_2=(v,u,\lambda_2)$ be a point stopped at by $T$ with $\dist_G(x_2,p_2)=\dist_G(x_2,T)$. If $x_2 \notin \{u,v\}$, we obtain $\dist_G(x_2,T)=\dist_G(x_2,p_2)\geq 1>\delta$, a contradiction. We may hence suppose that $u=x_2$. If $\lambda_2<1-\delta$, we have $\dist_G(x_2,T)=\dist_G(x_2,p_2)=1-\lambda_2>\delta$, a contradiction to $T$ covering $x_1x_2$. We may hence suppose that $\lambda_2>1-\delta$. Suppose for the sake of a contradiction that $\lambda_1+\lambda_2<2-2\delta$. We now consider $p=(x_1,x_2,1+\frac{1}{2}(\lambda_1-\lambda_2))$. Observe that $\lambda_1-\lambda_2\geq -\lambda_2\geq -1$ and $\lambda_1-\lambda_2\leq (1-\delta)-(1-\delta)=0$, so $p$ is well-defined and, by the choice of $\lambda_1$, we have that $T$ does not pass $p$. Next, by assumption, we have $\dist_G(p,p_1)=1+\frac{1}{2}(\lambda_1-\lambda_2)-\lambda_1=1-\frac{1}{2}(\lambda_1+\lambda_2)>1-\frac{1}{2}(2-2\delta)=\delta$. Further,  by assumption, we have $\dist_G(p,x_2)+\dist_G(x_2,p_2)=1-(1+\frac{1}{2}(\lambda_1-\lambda_2))+(1-\lambda_2)=1-\frac{1}{2}(\lambda_1+\lambda_2)>1-\frac{1}{2}(2-2\delta)=\delta$. By construction and the choice of $p_2$, this yields $\dist_G(p,T)=\min\{\dist_G(p,p_1),\dist_G(p,x_2)+\dist_G(x_2,p_2)\}>\delta$. This contradicts $T$ covering $x_1x_2$. We obtain that $\lambda_1+\lambda_2\geq 2-2\delta$, so $(i)$ holds.
\end{proof}
We finally handle the case that the tour stops at both vertices of the edge.
\begin{lemma}\label{char2stops}
Let $G$ be a connected graph, $\delta\geq 0$ a real, $T$ a nice tour in $G$ and $x_1x_2 \in E(G)$ such that $T$ stops at both $x_1$ and $x_2$. Then $T$ covers $x_1x_2$ if and only if one of the following holds:
\begin{enumerate}[(i)]
\item $T$ traverses $x_1x_2$,
\item $\delta \geq \frac{1}{2}$, and
\item $T$ contains the segment~$\seq{x_i&p(x_i,x_{3-i},\lambda)&x_i}$ for some $i \in \{1,2\}$ and some $\lambda 
\geq 1-2 \delta$.
\end{enumerate}
\end{lemma}
\begin{proof}
If $(i)$ holds, then $T$ covers $x_1x_2$ by construction. Next suppose that $(ii)$ holds and let $p=(x_1,x_2,\lambda)$ for some $\lambda \in [0,1]$. We obtain $\dist_G(p,T)\leq \min\{\dist_G(p,x_1),\dist_G(p,x_2)\}=\min\{\lambda,1-\lambda\}\leq \frac{1}{2}(\lambda+(1-\lambda))=\frac{1}{2}\leq \delta$. Finally suppose that $(iii)$ holds. By symmetry, we may suppose that $T$ contains the segment~$\seq{x_1&p(x_1,x_{2},\lambda)&x_1}$ for some $\lambda \geq 1-2 \delta$. Consider $p=(x_1,x_2,\lambda')$ for some $\lambda' \in [0,1]$. If $\lambda'<\lambda$, then $T$ passes $p$. If $\lambda'\geq \lambda$, we have $\dist_G(p,T)\leq \min\{\dist_G(p,p(u,v,\lambda)),\dist_G(p,x_2)\}=\min\{\lambda'-\lambda,1-\lambda'\}\leq \frac{1}{2}((\lambda'-\lambda)+(1-\lambda'))=\frac{1}{2}(1-\lambda)\leq\delta$. Hence $T$ covers $x_1x_2$.

Now suppose that $T$ covers $x_1x_2$. If $T$ traverses $x_1x_2$ or $\delta \leq \frac{1}{2}$, then $(i)$ or $(ii)$ 
holds, so there is nothing to prove. We may hence suppose that $T$ does not traverse $x_1x_2$ and 
$\delta<\frac{1}{2}$. If $T$ does not stop at any point of the form $p(x_i,x_{3-i},\lambda)$ for some $i \in 
\{1,2\}$ and $\lambda \in (0,1)$, then, as $\delta<\frac{1}{2}$, for $p=p(x_1,x_2,\frac{1}{2})$, we obtain 
$\dist_G(p,T)= \min\{\dist_G(p,x_1),\dist_G(p,x_2)\}=\frac{1}{2}>\delta$, a contradiction. Hence, as $T$ is nice and 
by symmetry, we may suppose that $T$ contains the subsequence $x_1p(x_1,x_2,\lambda)x_1$ for some $\lambda \in 
(0,1)$ and $T$ does not stop at any other point on $x_1x_2$. If $\lambda <1-2 \delta$, then consider 
$p=p(x_1,x_2,\frac{1}{2}(1+\lambda))$. By construction, we have 
$\dist_G(p,T)=\min\{\dist_G(p,p(x_1,x_2,\lambda)),\dist_G(p,x_2)\}=\min\{\frac{1}{2}(1+\lambda)-\lambda,1-\frac{1}{2}(1+\lambda)\}=\frac{1}{2}(1-\lambda)>\delta$,
 a contradiction to $T$ covering $x_1x_2$. It follows that $(iii)$ holds.
\end{proof}

\subsection{Discretization}
\label{discsec}
We now show that we can restrict ourselves to tours whose stopping points come from a small set of points. In particular, we aim to prove \Cref{lemma:discretization}, which we restate
below for convenience.

\TheoremDiscretization*\label\thisthm

Again, we distinguish two domains of values for $\delta$. 
For the main proof, we need some polyhedral results related to the vertex cover polytope.
Given a graph~$G$, potentially containing loops, we consider the vertex-cover LP of $G$:
\begin{equation}\tag{$\diamond$}
\boxed{\begin{gathered}
		\text{Minimize }\sum_{v \in V}z_v\text{, subject to the following:}\\
		z_{u}+z_{v}\geq 1 \text{ for all $e=uv \in E(G)$},\\
		z_v \in [0,1] \text{ for all $v \in V(G)$.}
	\end{gathered}}
\end{equation}
Observe that, if $e$ is a loop at a vertex $v$ of $G$, then the corresponding inequality is $2z_v\geq 1$.
We use that there is a half-integral solution to the vertex-cover LP.

\begin{proposition}[\cite{NemhauserT75}]
\label{halfint}
For every graph~$G$, there is an optimal solution $(z_v\colon v \in V(G))$ of $(\diamond)$ with $z_v \in \{0,\frac{1}{2},1\}$ for all $v \in V(G)$.
\end{proposition}

We further consider the following LP for some graph~$G$ and some nonnegative constant~$\gamma$:
\begin{equation}\tag{$\diamond \diamond$}
\boxed{\begin{gathered}
		\text{Minimize }\sum_{v \in V}z_v\text{, subject to the following:}\\
		z_{u}+z_{v}\geq \gamma \text{ for all $e=uv \in E(G)$},\\
		z_u \in [0,1] \text{ for all $u \in V(G)$.}
	\end{gathered}}
\end{equation}
We next show that for all reasonable values of $\gamma$ for which ($\diamond \diamond$) has a solution, this LP has a feature similar to $(\diamond)$, namely that it has an optimal solution all of whose values come from a small set.
\begin{proposition}\label{halfintgen01}
For every graph~$G$ and every $\gamma \in (0,1]$, there is an optimal solution $(z_v\colon v \in V(G))$ of $(\diamond \diamond)$ with $z_v \in \{0,\frac{1}{2}\gamma,\gamma\}$ for all $v \in V(G)$.
\end{proposition}
\begin{proof}
Let $(z_v^1\colon v \in V(G))$ be an optimal solution for $(\diamond \diamond)$ and let $(z_v^2\colon v \in V(G))$ be defined by $z_v^2=\frac{1}{\gamma}z_v^1$ for all $v \in V(G)$. Observe that, as $(z_v^1\colon v \in V(G))$ is an optimal solution for $(\diamond \diamond)$, we have $z_v^1\leq \gamma$ for all $v \in V(G)$. This yields $z_v^2=\frac{1}{\gamma}z_v^1\leq 1$ for all $v \in V(G)$. Further, we have $z_v^2=\frac{1}{\gamma}z_v^1\geq 0$ for all $v \in V(G)$. Finally, for every $e=uv \in E(G)$, we have $z_{u}^2+z_{v}^2=\frac{1}{\gamma}z_u^1+\frac{1}{\gamma}z_v^1=\frac{1}{\gamma}(z_u^1+z_v^1)\geq 1$. Hence $(z_v^2\colon v \in V(G))$ is a feasible solution for $(\diamond)$.

Next, by Proposition~\ref{halfint}, there is an optimal solution $(z_v^3\colon v \in V(G))$ for $(\diamond)$ with $z_v^3 \in \{0,\frac{1}{2},1\}$ for all $v \in V(G)$. Let $(z_v^4\colon v \in V(G))$ be defined by $z_v^4=\gamma z_v^3$ for all $v \in V(G)$. Observe that $z_v^4 \in \{0,\frac{1}{2}\gamma,\gamma\}$ for all $v \in V(G)$. Next, we have $z_v^4=\gamma z_v^3\geq 0$ and $z_v^4=\gamma z_v^3\leq \gamma \leq 1$ for all $v \in V(G)$. Finally, for every $e=uv \in E(G)$, we have $z_{u}^4+z_{v}^4=\gamma z_u^3+\gamma z_v^3=\gamma(z_u^3+z_v^3)\geq \gamma$. Hence $(z_v^4\colon v \in V(G))$ is a feasible solution for $(\diamond \diamond)$.

We further have $\sum_{v \in V(G)}z_v^4 = \gamma \sum_{v \in V(G)}z_v^3 \leq \gamma \sum_{v \in V(G)}z_v^2 =  \sum_{v \in V(G)}z_v^1$. Hence $(z_v^4\colon v \in V(G))$ is an optimal solution for $(\diamond \diamond)$.

\end{proof}

\begin{proposition}\label{halfintgen12}
For every graph~$G$ without isolated vertices and every $\gamma \in (1,2]$, there is an optimal solution $(z_v\colon v \in V(G))$ of $(\diamond \diamond)$ with $z_v \in \{\gamma-1,\frac{1}{2}\gamma,1\}$ for all $v \in V(G)$.
\end{proposition}
\begin{proof}
If $\gamma=2$, it is easy to see that the only feasible optimal solution is obtained by setting $z_v=1$ for all $v \in V(G)$. We may hence suppose that $\gamma<2$.
Let $(z_v^1\colon v \in V(G))$ be an optimal solution for $(\diamond \diamond)$ and let $(z_v^2\colon v \in V(G))$ be defined by $z_v^2=\frac{z_v^1-(\gamma-1)}{2-\gamma}$ for all $v \in V(G)$. Observe that, as $(z_v^1\colon v \in V(G))$ is a feasible solution for $(\diamond \diamond)$ and $G$ does not have any isolated vertex, every $v \in V(G)$ is incident to an edge $uv \in E(G)$. 
Hence we have $z_v^1\geq \gamma-z_u^1-1\geq \gamma-1$ for all $v \in V(G)$. This yields $z_v^2=\frac{z_v^1-(\gamma-1)}{2-\gamma} \geq 0$. We further have  $z_v^2=\frac{z_v^1-(\gamma-1)}{2-\gamma}\leq \frac{2-\gamma}{2-\gamma}=1$ for all $v \in V(G)$. Finally, for every $e=uv \in E(G)$, we have $z_{u}^2+z_{v}^2=\frac{z_u^1+z_v^1-2(\gamma -1)}{2-\gamma}\geq \frac{\gamma-2(\gamma -1)}{2-\gamma} = 1$. Hence $(z_v^2\colon v \in V(G))$ is a feasible solution for $(\diamond)$.

Next, by Proposition~\ref{halfint}, there is an optimal solution $(z_v^3\colon v \in V(G))$ for $(\diamond)$ with $z_v^3 \in \{0,\frac{1}{2},1\}$ for all $v \in V(G)$. Let $(z_v^4\colon v \in V(G))$ be defined by $z_v^4=(2-\gamma) z_v^3+(\gamma-1)$ for all $v \in V(G)$. Observe that $z_v^4 \in \{\gamma-1,\frac{1}{2}\gamma,1\}$ for all $v \in V(G)$. Next, we have $z_v^4=(2-\gamma) z_v^3+(\gamma-1)\geq 0$ and $z_v^4=(2-\gamma) z_v^3+(\gamma-1)\leq (2-\gamma)+(\gamma-1) \leq 1$ for all $v \in V(G)$. Finally, for every $e=uv \in E(G)$, we have $z_{u}^4+z_{v}^4=(2-\gamma)(z_u^3+ z_v^3)+2(\gamma-1)\geq (2-\gamma)+2(\gamma-1)=\gamma$. Hence $(z_v^4\colon v \in V(G))$ is a feasible solution for $(\diamond \diamond)$.

We further have $\sum_{v \in V(G)}z_v^4 = (2-\gamma) \sum_{v \in V(G)}z_v^3 +(\gamma-1)\abs{V(G)}\leq (2-\gamma) \sum_{v \in V(G)}z_v^2 +(\gamma-1)\abs{V(G)} = \sum_{v \in V(G)}z_v^1$. Hence $(z_v^4\colon v \in V(G))$ is an optimal solution for $(\diamond \diamond)$.
\end{proof}

We now extend this result to a somewhat more general polyhedron with two different kinds of equations. Namely, for two graphs $G_1,G_2$ on a common vertex set $V$ and some $\gamma \in [0,1]$, we define the following polyhedron.
\begin{equation}\tag{$\clubsuit$}
\boxed{\begin{gathered}
		\text{Minimize }\sum_{v \in V}z_v\text{, subject to the following:}\\
z_{u}+z_{v}\geq \gamma + 1 \text{ for all $e=uv \in E(G_1)$},\\
z_{u}+z_{v}\geq \gamma \text{ for all $e=uv \in E(G_2)$},\\
z_u \in [0,1] \text{ for all $u \in V$.}
	\end{gathered}}
\end{equation}
\begin{lemma}\label{genvc}
Let $G_1$ and $G_2$ be graphs on a common vertex set $V$ and let $\gamma \in [0,1]$. Then there is an optimal solution $(z_v\colon v \in V)$ of $(\clubsuit)$ with $z_v \in S_\gamma$ for all $v \in V$ where $S_\gamma=\{0,\frac{1}{2}\gamma,\gamma,\frac{1}{2}\gamma+\frac{1}{2},1\}$.
\end{lemma}
\begin{proof}
Clearly, we may suppose that every $v \in V$ is incident to at least one edge in $E(G_1)\cup E(G_2)$.
Let $E'_2$ be the set of edges $uv \in E(G_2)$ such that none of $u$ and $v$ is incident to an edge of $E(G_1)$.
We next consider the following LP:
\begin{equation}\tag{$\clubsuit\clubsuit$}
\boxed{\begin{gathered}
		\text{Minimize }\sum_{v \in V}z_v\text{, subject to the following:}\\
z_{u}+z_{v} \geq \gamma + 1 \text{ for all $e=uv \in E(G_1)$},\\
z_{u}+z_{v} \geq \gamma \text{ for all $e=uv \in E'_2$},\\
z_f \in [0,1] \text{ for all $v \in V$.}
	\end{gathered}}
\end{equation}
\begin{claim}\label{redundant}
A vector $(z_v\colon v \in V)$ is a feasible solution for $(\clubsuit \clubsuit)$ if and only if it is a feasible solution for $(\clubsuit)$.
\end{claim}

\begin{claimproof}
As the constraints in $(\clubsuit \clubsuit)$ are a subset of the constraints in $(\clubsuit)$, a feasible solution for $(\clubsuit)$ is also feasible for $(\clubsuit \clubsuit)$. Now let $(z_v\colon v \in V)$ be a feasible solution for $(\clubsuit \clubsuit)$. In order to show that $(z_v\colon v \in V)$ is feasible for $(\clubsuit)$, it suffices to prove that $z_{u}+z_{v}\geq \gamma$ holds for all $uv \in E(G_2)-E_2'$. Let $uv \in E(G_2)-E_2'$. By symmetry, we may suppose that $u$ is incident to an edge in $E(G_1)$, so there is some $w \in V$ such that $uw \in E(G_1)$. As $(z_v: v \in V)$ is feasible for $(\clubsuit \clubsuit)$, we obtain $z_{u}+z_{v}\geq z_{u}\geq (\gamma +1)-z_{w}\geq (\gamma +1)-1=\gamma$.
\end{claimproof}
In the following, we define $V_1$ to be the vertices which are incident to at least one edge in $E(G_1)$ and we set $V_2=V-V_1$. 
We further consider the following two LPs.
\begin{center}
\begin{equation}\tag{$\clubsuit\clubsuit\clubsuit$}
\boxed{\begin{gathered}
		\text{Minimize }\sum_{v \in V_1}z_v\text{, subject to the following:}\\
z_{u}+z_{v}\geq \gamma+1 \text{ for all $uv \in E(G_1)$},\\
z_v \in [0,1] \text{ for all $v \in V_1$.}
	\end{gathered}}
\end{equation}
\begin{equation}\tag{$\clubsuit\clubsuit\clubsuit\clubsuit$}
\boxed{\begin{gathered}
		\text{Minimize }\sum_{v \in V_2}z_v\text{, subject to the following:}\\
z_{u}+z_{v}\geq \gamma \text{ for all $uv \in E_2'$},\\
z_v\in [0,1] \text{ for all $v \in V_2$.}
	\end{gathered}}
\end{equation}
\end{center}
\begin{claim}\label{21}
A vector $(z_v\colon v \in V)$ is a feasible solution for $(\clubsuit\clubsuit)$ if and only if  $(z_v\colon v \in V_1)$ is a feasible solution for $(\clubsuit\clubsuit\clubsuit)$ and $(z_v\colon v \in V_2)$ is a feasible solution for $(\clubsuit\clubsuit\clubsuit\clubsuit)$.
\end{claim}
\begin{claimproof}
This is a direct consequence of the fact that the constraints in $(\clubsuit\clubsuit\clubsuit)$ are exactly the constraints imposed on $(z_v\colon v \in V_1)$ in $(\clubsuit\clubsuit)$ and the constraints in $(\clubsuit\clubsuit\clubsuit\clubsuit)$ are exactly the constraints imposed on $(z_v\colon v \in V_2)$ in $(\clubsuit\clubsuit)$. 
\end{claimproof}
We are now ready to finish the proof of Lemma~\ref{genvc}. By Proposition~\ref{halfintgen12}, there is an optimal solution $(z^*_v\colon v \in V_1)$ of $(\clubsuit\clubsuit\clubsuit)$ such that $z^*_v \in \{\gamma,\frac{1}{2}\gamma+\frac{1}{2},1\}$ for all $v \in V_1$. By Proposition~\ref{halfintgen01}, there is an optimal solution $(z^*_v\colon v \in V_2)$ of $(\clubsuit\clubsuit\clubsuit\clubsuit)$ such that $z^*_v \in \{0,\frac{1}{2}\gamma,\gamma\}$ for all $v \in V_2$. By Claim~\ref{21}, we obtain that $(z^*_v\colon v \in V)$ is an optimal solution for $(\clubsuit\clubsuit)$. Next, by Claim~\ref{redundant}, we obtain that $(z^*_v\colon v \in V)$ is an optimal solution for $(\clubsuit)$. This finishes the proof.
\end{proof}

We are now ready to prove our main results. We first give the result for $\delta \geq 1/2$, relying  on a connection to the vertex cover polytope.  Given a graph~$G$ and a tour~$T$ in $G$, we use $V_T$ for the set of vertices in $V(G)$ that $T$ stops at.
\begin{lemma}\label{lem:discretization:gt:one}
For every $\delta \geq \frac{1}{2}$ and every connected graph~$G$, there is a shortest \deltatour~$T$  such that all stopping points of $T$ are of the form $p(u,v,\lambda)$ for some $uv \in E(G)$ for some $\lambda \in S'_\delta$ where $S'_\delta=\{0,\frac{1}{2}(\floor{2 \delta}+1-2 \delta),\floor{2 \delta}+1-2 \delta,\frac{1}{2}(\floor{2 \delta}+2-2 \delta)\}$. Moreover, either $\alpha(T)\leq 2$ or $T$ is nice and for every non-integral stopping point $p$ of $T$, there are some $uv \in E(G)$ and some $\lambda \in S'_\delta$ such that $p=p(u,v,\lambda)$ and $T$ contains the sequence $upu$.
\end{lemma}
\begin{proof}
Let $T_0$ be a shortest \deltatour in $G$. Further, let $k=\floor{2 \delta}$. For technical reasons, we need to distinguish two cases, with the second one being degenerate. 
\setcounter{Case}{0}
\begin{Case}\label{nondeg}
$T_0$ stops at a vertex of $V(G)$.
\end{Case}
By \cref{lem:tournice}, we may suppose that $T_0$ is nice.
For $j \in \mathbb{Z}$, let $E_j$ be the set of  edges $x_1x_2 \in E(G)$ with 
$\dist_G(x_1,V_{T_0})+\dist_G(x_2,V_{T_0})-k=j$. Now consider some $e=x_1x_2 \in E_0 \cup E_1$. Further, for $i \in 
[2]$, let $p_i=p(u_i,v_i,\lambda_i)$ be a stopping point of $T_0$ such that $\dist_G(x_i,p_i)=\dist_G(x_i,T_0)$. 
First suppose that $T_0$ stops at none of $x_1$ and $x_2$.  By Lemma~\ref{nostopchar}, we may suppose that 
$\lambda_1+\lambda_2 \geq \dist_G(x_1,v_1)+\dist_G(x_2,v_2 )+ 3-2 \delta$ hold and that $T_0$ stops at none of $v_1$ 
and $v_2$ by the choice of $p_1$ and $p_2$. We say that $(f_1,f_2)$ is associated to $e$ where $f_i=u_iv_i$ for $i 
\in [2]$. Now suppose that $T_0$ stops at exactly one of $x_1$ and $x_2$, say $x_1$. We obtain by 
Lemma~\ref{char1stop} that one of $(i)$ and $(ii)$ of Lemma~\ref{char1stop} holds. First suppose that $(i)$ holds, 
so $T_0$ stops at some points $p(x_1,x_2,\lambda_1)$ and $p(v,x_2,\lambda_2)$ with $\lambda_1+\lambda_2 \geq 2-2 
\delta$. We say that $(f_1,f_2)$ is associated to $e$ where $f_1=e$ and $f_2=vx_2$. We use $F$ for the set of edges 
contained in a pair associated to some $e \in E_0\cup E_1$ and for $j \in \{0,1\}$, we use $F_j$ for the set of 
edges contained in a pair associated to some $e \in E_j$. Further, let $F^*$ be the set of edges in $E_0 \cup E_1$ 
exactly one of whose endpoints is stopped at by $T_0$ and for which $(ii)$ of Lemma~\ref{char1stop} holds. Observe 
that for every $f \in F\cup F^*$, exactly one of its vertices is contained in $V_{T_0}$. We now consider the 
following linear program. For some $f=uv \in F \cup F^*$ with $u \in V_{T_0}$, we refer by $\lambda_f$ to the 
largest value such that $T_0$ stops $(u,v,\lambda_f)$. Observe that $\lambda_f$ is uniquely defined as $T_0$ is 
nice. We consider the following linear program.

\medskip

\begin{equation}\tag{$\ast$}
\boxed{\begin{gathered}
		\text{Minimize }\sum_{f \in F \cup F^*}z_f\text{, subject to the following:}\\
2z_{f}\geq 2-2 \delta \text{ if $f \in F^*$},\\
z_{f_1}+z_{f_2}\geq k+1-2 \delta \text{ if $(f_1,f_2)$ is associated to some $e \in E_0$},\\
z_{f_1}+z_{f_2}\geq k+2-2 \delta \text{ if $(f_1,f_2)$ is associated to some $e \in E_1$},\\
z_f \in [0,1] \text{ for all $f \in F\cup F^*$.}
	\end{gathered}
}
\end{equation}
\begin{claim}\label{lambdafeas}
$(\lambda_f\colon f \in F)$ is a feasible solution for $(\ast)$.
\end{claim}
\begin{claimproof}
By definition, we have $\lambda_f\in [0,1]$ for all $f \in F$. Now let $(f_1,f_2)$ be a pair associated to some $e \in E_j$ for some $j \in \{0,1\}$ with $f_i=u_iv_i$ with $u_i \in V_{T_0}$. By the choice of $(f_1,f_2)$, we have $\lambda_{f_1}+\lambda_{f_2} \geq \dist_G(x_1,v_1)+\dist_G(x_2,v_2)+3-2\delta\geq \dist_G(x_1,V_{T_0})+\dist_G(x_2,V_{T_0})+1-2 \delta=(k+j+1)-2 \delta$. Finally, for every $f \in F^*$, we have $z_f \geq 1-\delta$ by definition.
\end{claimproof}
In the following, for every $f=uv \in F \cup F^*$ with $u \in V_{T_0}$, let $S_f$ be defined by $\seq{u&p(u,v,\lambda_f)&u}$ if $\lambda_f>0$ and by $S_f=\seq{u}$ if $\lambda_f=0$. Observe that $S_f$ is a segment of $T_0$ for all $f \in F \cup F^*$ as $T_0$ is nice. Now let $z=(z_f\colon f \in F \cup F^*)$ be a vector in $[0,1]^{F\cup F^*}$. We  let $S^z_f$ be defined by $S^z_f=up(u,v,z_f)u$ if $z_f>0$ and by $S'_f=u$ if $z_f=0$. We further let the tour~$T^z$ be the tour obtained from $T_0$ by, for all $f=uv \in F\cup F^*$ with $u \in V_{T_0}$, replacing the sequence $S_f$ by $S^z_f$.
\begin{claim}\label{feaslambda}
Let $(z_f\colon f \in F \cup F^*)$ be a feasible solution for $(\ast)$. Then  $T^z$  is a nice \deltatour in $G$. 
\end{claim}

\begin{claimproof}
As $T_0$ is nice and stops at most one vertex of every $f \in F \cup F^*$ and by construction, we obtain that $T^z$ is nice.

Let $e=x_1x_2 \in E(G)$. First suppose that $T_0$ stops at both $x_1$ and $x_2$. Then $T^z$ stops at both $x_1$ and $x_2$. As $T^z$ is nice, by Lemma~\ref{char2stops}, we obtain that $e$ is covered by $T^z$. We may hence in the following suppose that $T_0$ stops at at most one of $x_1$ and $x_2$. Let $j=\dist_G(x_1,V_{T_0})+\dist_G(x_2,V_{T_0})-k$.

 First suppose that $j \geq 2$. As $k \geq 0$, we have $\dist_G(x_1,V_{T_0})+\dist_G(x_2,V_{T_0})\geq 2$, so in 
 particular, we obtain that $T_0$ stops at none of $x_1$ and $x_2$. For $i \in \{1,2\}$, let 
 $p_i=p(u_i,v_i,\lambda_i)$ be stopping points of $T_0$ with $\lambda_i \in [0,1)$. We then have 
 $\lambda_1+\lambda_2 < 2\leq \dist_G(x_1,V_{T_0})+\dist_G(x_2,V_{T_0})-k< 
 \dist_G(x_1,v_1)+\dist_G(x_2,v_2)+3-2\delta$. By Lemma~\ref{nostopchar}, this contradicts $T_0$ being a \deltatour.
 
  Now suppose that $j \leq -1$. If $T_0$ stops at one of $x_1$ and $x_2$, say $x_1$, we have 
  $k\geq\dist_G(x_1,V_{T_0})+\dist_G(x_2,V_{T_0})+1=2$. As $k=\floor{2 \delta}$, we obtain that $\delta\geq 1$ and 
  hence $0 \geq 1-\delta$. Since $T^z$ stops at $x_1$, it follows that $(ii)$ of Lemma~\ref{char1stop} holds, so $e$ 
  is covered by $T^z$. Now suppose that $T_0$ stops at none of $x_1$ and $x_2$. For $i \in \{1,2\}$, let $v_i \in 
  V_{T_0}$ with $\dist_G(v_i,x_i)=\dist_G(x_i,V_{T_0})$ and observe that $T$ stops at $v_i$. Further, let $w_i$ be 
  the first vertex on a shortest $v_ix_i$-path in $G$. We have $0+0= 0 = \dist_G(x_1,v_1)+\dist_G(x_2,v_2)-(k+j)\geq 
  \dist_G(x_1,w_1)+\dist_G(x_2,w_2)+3-k\geq \dist_G(x_1,w_1)+\dist_G(x_2,w_2)+3-2 \delta$. As $T^z$ stops at $v_1$ 
  and $v_2$, it follows that $e$ is covered by $T^z$ by Lemma~\ref{nostopchar}.
  
   We may hence suppose that $j \in \{0,1\}$. If $e \in F^*$, we obtain $z_e \geq 1-\delta$, so $e$ is covered by 
   $T^z$ by Lemma~\ref{char1stop} $(ii)$. Otherwise, let $(f_1,f_2)$ be the pair associated to $e$ such that 
   $f_i=u_iv_i$ with $u_i \in V_{T_0}$ for $i \in \{1,2\}$. As $(z_f\colon f \in F)$ is a feasible solution for $(\ast)$, 
   we have $z_{f_1}+z_{f_2}\geq k+j+1-2 
   \delta=\dist_G(x_1,V_{T_0})+\dist_G(x_2,V_{T_0})+1-2\delta=\dist_G(x_1,v_1)+\dist_G(x_2,v_2)+3-2\delta$. We now 
   obtain that $T^z$ is a \deltatour by Lemmas~\ref{nostopchar},~\ref{char1stop}, and~\ref{char2stops}. 
\end{claimproof}
We are now ready to conclude Case~\ref{nondeg}. By Lemma~\ref{genvc}, there is an optimal solution $z^*=(z^*_f\colon f \in F \cup F^*)$ of $(\ast)$ such that $z^*_f\in S'_\delta$ for all $f \in F \cup F^*$. By Claim~\ref{feaslambda}, it follows that $T^{z^*}$ is a nice $\delta$-tour in $G$. It then follows from Claim~\ref{lambdafeas} that $T^{z^*}$ is a nice shortest \deltatour in $G$. 
\begin{Case}\label{degen}
$T_0$ does not stop at a vertex in $V(G)$.
\end{Case}
As $T_0$ is a nice tour, we obtain that $T_0$ has at most two stopping points and there is a unique edge $uv \in E(G)$ such that all stopping points of $T$ are on $uv$.

If $\delta<1$, then for any $w \in V(G)-\{u,v\}$, we have $\dist_G(w,T_0)>\dist_G(w,\{u,v\})\geq 1>\delta$. It hence follows that $V(G)=\{u,v\}$. We obtain that the tour consisting of $p(u,v,\frac{1}{2})$ is a \deltatour in $G$ which is obviously a shortest one. In the following, we may hence suppose that $\delta\geq 1$.

For $j \in \mathbb{Z}$, let $E_j$ be the set of  edges $xy \in E(G)$ with 
$\dist_G(x,\{u,v\})+\dist_G(y,\{u,v\})-k=j$. Let $\lambda_u$ be the largest real such that $p_u=p(u,v,\lambda_u)$ is 
a stopping point of $T_0$ and let $\lambda_v$ be the largest real such that $p_v=p(v,u,\lambda_v)$ is a stopping 
point of $T_0$. For every $w \in V(G)-\{u,v\}$, we define $\alpha(w)=u$ if $\dist_G(w,T_0)=\dist_G(w,p_u)$ and 
$\alpha(w)=v$, otherwise. We now consider the following linear program.
\begin{equation}\tag{$\Box$}
\boxed{\begin{gathered}
		\text{Minimize }z_u+z_v\text{, subject to the following:}\\
z_{\alpha(x)}+z_{\alpha(y)}\geq k+1-2 \delta \text{ for every $xy \in E_{-2}$},\\
z_{\alpha(x)}+z_{\alpha(y)}\geq k+2-2 \delta \text{ for every $xy \in E_{-1}$},\\
z_u,z_v \in [0,1].
	\end{gathered}}
\end{equation}

\begin{claim}\label{lambdafeas2}
$(\lambda_u,\lambda_v)$ is a feasible solution for $(\Box)$.
\end{claim}
\begin{claimproof}
By definition, we have $\lambda_u,\lambda_v\in [0,1]$. Now consider some $e=xy \in E_j$ for some $j \in \{-2,-1\}$. By the definition of $\alpha$ and Lemma~\ref{nostopchar}, we have $\lambda_{\alpha(x)}+\lambda_{\alpha(y)} \geq \dist_G(x,\alpha(x))+\dist_G(y,\alpha(y))+3-2 \delta=\dist_G(x,\{u,v\})+\dist_G(y,\{u,v\})+3-2 \delta=(k+j+3)-2 \delta$. 
\end{claimproof}
\begin{claim}\label{1point}
Let $(z_u,z_v)$ be a feasible solution for $(\Box)$ with $z_u+z_v \leq 1$. Then the tour~$T$ consisting of $\seq{p(u,v,z_u)}$ is a $\delta$-tour in $G$.
\end{claim}
\begin{claimproof}
Let $z'_u=z_u$ and $z'_v=1-z_u$. Let $xy \in E(G)$ and let $j=\dist_G(x,\{u,v\})+\dist_G(y,\{u,v\})-k$.  If $T$ stops at one of $x$ and $y$, we have that Lemma~\ref{char1stop} $(ii)$ holds as $0 \geq 1-\delta$. We may hence suppose that $T$ stops at none of $x$ and $y$. First suppose that $j \geq 0$. This yields $\lambda_{\alpha(x)}+\lambda_{\alpha(y)}\leq 2\leq j+2=\dist_G(x,\{u,v\})+\dist_G(y,\{u,v\})-k+2<\dist_G(x,\{u,v\})+\dist_G(x,\{u,v\})+3-2 \delta$, a contradiction to $T_0$ being a $\delta$-tour in $G$ by Lemma~\ref{nostopchar}. We may hence suppose that $j \leq -1$. If $j \leq -3$, we have $z'_{\alpha(x)}+z'_{\alpha(y)}\geq z_{\alpha(x)}+z_{\alpha(y)}\geq 0+0=0\geq j+3=\dist_G(x,\{u,v\})+\dist_G(y,\{u,v\})-k+3\geq \dist_G(x_1,\alpha(x))+\dist_G(y,\alpha(y))+3-2 \delta$, so $e$ is covered by $T$ by Lemma~\ref{nostopchar}. We may hence suppose that $j \in \{-2,-1\}$.

As $z$ is feasible for $(\Box)$,we obtain $z'_{\alpha(x)}+z'_{\alpha(y)}\geq z_{\alpha(x)}+z_{\alpha(y)}\geq k+(j+3)-2 \delta\geq \dist_G(x,\{u,v\})+\dist_G(y,\{u,v\})+3-2 \delta$. It follows by Lemma~\ref{nostopchar} that $e$ is covered by $T$.
\end{claimproof}
\begin{claim}\label{2points}
Let $(z_u,z_v)$ be a feasible solution for $(\Box)$ with $z_u+z_v \geq 1$. Then the tour~$p(u,v,z_u)p(v,u,z_v)p(u,v,z_u)$ is a $\delta$-tour in $G$.
\end{claim}

\begin{claimproof}
Let $xy \in E(G)$ and let $j=\dist_G(x,\{u,v\})+\dist_G(y,\{u,v\})-k$.  If $T$ stops at one of $x$ and $y$, we have that Lemma~\ref{char1stop} $(ii)$ holds as $0 \geq 1-\delta$. We may hence suppose that $T$ stops at none of $x$ and $y$. First suppose that $j \geq 0$. This yields $\lambda_{\alpha(x)}+\lambda_{\alpha(y)}\leq 2\leq j+2=\dist_G(x,\{u,v\})+\dist_G(y,\{u,v\})-k+2<\dist_G(x,\{u,v\})+\dist_G(x,\{u,v\})+3-2 \delta$, a contradiction to $T_0$ being a $\delta$-tour in $G$ by Lemma~\ref{nostopchar}. We may hence suppose that $j \leq -1$. If $j \leq -3$, we have $z_{\alpha(x)}+z_{\alpha(y)}\geq 0+0=0\geq j+3=\dist_G(x,\{u,v\})+\dist_G(y,\{u,v\})-k+3\geq \dist_G(x,\alpha(x))+\dist_G(y,\alpha(y))+3-2 \delta$, so $e$ is covered by $T$ by Lemma~\ref{nostopchar}. We may hence suppose that $j \in \{-2,-1\}$.

As $z$ is feasible for $(\Box)$, we obtain $ z_{\alpha(x)}+z_{\alpha(y)}\geq k+(j+3)-2 \delta\geq \dist_G(x,\{u,v\})+\dist_G(y,\{u,v\})+3-2 \delta$. It follows that $e$ is covered by $T$.
\end{claimproof}

We are now ready to finish the proof of Case~\ref{degen}. By Lemma~\ref{genvc}, we have that there is an optimal solution $z^*=(z_u^*,z_v^*)$ of $(\Box)$ with $z_u^*,z_v^* \in S'_\delta$. If $z_u^*+z_v^*\leq 1$, let $T^*$ be the tour consisting of $p(u,v,z_u^*)$. By Claim~\ref{1point}, it follows that $T^*$ is a \deltatour in $G$. Clearly, it follows that $T^*$ is a shortest $\delta$-tour in $G$.

If $z_u^*+z_v^*\leq 1$, let $T^*=\seq{p(u,v,z_u^*)&p(v,u,z_v^*)&p(u,v,z_u^*)}$. By Claim~\ref{2points}, we obtain that $T^*$ is a \deltatour in $G$. It follows by Claim~\ref{lambdafeas2} that $T^*$ is a shortest $\delta$-tour in $G$.
\end{proof}

We finally give a somewhat stronger result for $\delta<1/2$. Its proof is elementary.

\begin{lemma}\label{discreteklein}
Let $\delta \in [0,1/2]$ and let $G$ be a connected graph. Then there is a shortest \deltatour~$T$ that either consists of only two stopping points or is nice and such that for every $x_1x_2 \in E(G)$, one of the following holds:
\begin{enumerate}[(a)]
\item $V(G)=\{x_1,x_2\}$ and \[
    T= 
\begin{cases}
    p(x_1,x_2,\delta)p(x_1,x_2,1-\delta)p(x_1,x_2,\delta),& \text{if } \delta<\frac{1}{2}\\
    p(x_1,x_2,\frac{1}{2}),              & \text{if $\delta=\frac{1}{2}$,}
\end{cases}
\]
\item $\deg_G(x_i)=1$ and $\deg_G(x_{3-i})\geq 2$ for some $i \in [2]$, and $T$ contains the tour segment~$\seq{x_{3-i}&p(x_{3-i},x_i,1-\delta)&x_{3-i}}$ and does not stop at any other points on $x_1x_2$,
\item $\min\{\deg_G(x_1),\deg_G(x_2)\}\geq 2$, and $T$ traverses $x_1x_2$,
\item $\min\{\deg_G(x_1),\deg_G(x_2)\}\geq 2$, $T$ contains the two tour segments\\
 \[
 \seq{x_{3-i}}\text{ and }
 \begin{cases}{\seq{x_{i}&p(x_{i},x_{3-i},1-2\delta)&x_{i}}}, &\text{if $\delta<\frac{1}{2}$}\\{\seq{x_i}}, &\text{if $\delta=\frac{1}{2}$}\end{cases},
 \]
 and $T$ does not stop at any other points on $x_1x_2$.
\end{enumerate}
\begin{proof}
First suppose that $G$ admits a \deltatour $T$ that consists of at most two stopping points. It follows directly from Lemmas \ref{char1stop} and \ref{char2stops} and the fact that $\delta \leq \frac{1}{2}$ that $E(G)$ consists of a single edge $uv$. It is easy to see that $(a)$ holds for this edge when $T$ is a shortest tour.

We may hence suppose that every \deltatour of $G$ consists of at least 3 stopping points. Hence, by \cref{lem:tournice}, we may suppose that $T$ is nice. It follows that $|V(G)|\geq 3$.
Among all nice shortest tours of $G$, we choose $T$ to be one which minimizes the number of edges for which none of $(b),(c)$, and $(d)$ hold. We will show that one of $(b),(c)$, and $(d)$ holds for $T$ for every edge in $E(G)$.

First consider an edge $x_1x_2\in E(G)$ with $\deg(x_1)=1$ and $\deg(x_2)\geq 2$. By Lemma~\ref{nostopchar} and as $T$ is a tour, we obtain that $T$ stops at $x_2$. We obtain, as $\deg(x_1)=1$ and as $T$ is nice tour that Lemma~\ref{char1stop} $(ii)$ with $i=2$ holds for $x_1x_2$. Hence $T$ contains the segment~$\seq{x_2&p(x_2,x_1,\lambda)&x_2}$ for some $\lambda \in [1-\delta,1]$. If $\lambda>1-\delta$, let $T'$ be obtained from $T$ by replacing $p(x_2,x_1,\lambda)$ by $p(x_2,x_1,1-\delta)$. It follows that $x_1x_2$ satisfies Lemma~\ref{char1stop} $(ii)$ for $T'$. Further, clearly every $e \in E(G)-x_1x_2$ satisfies one of $(b),(c)$, and $(d)$ for $T'$. As $T$ is nice, so is $T'$. It follows by Lemma~\ref{char2stops} that $T'$ is a \deltatour in $G$. As $\len(T')<\len(T)$, we obtain a contradiction to the choice of $T$.

Next suppose for the sake of a contradiction that there is some $x \in V(G)$ with $\deg_G(x)\geq 2$ that is not stopped at by $T$. Let $Y=N_G(x)$. By the above, we obtain that $\deg_G(y)\geq 2$ for all $y \in Y$. We obtain by Lemma~\ref{char1stop} that $yx$ satisfies one of $(i)$ and $(ii)$ for $T$ for all $y \in Y$. As $T$ is nice, this yields that $T$ contains the segment~$\seq{y&p_y&y}$ for every $y \in Y$ where $p_y=p(y,x,\lambda_y)$ for some $\lambda_y\in (0,1)$. We now choose $y_1 \in Y$ such that $\lambda_{y_1}=\max\{\lambda_y\colon y \in Y\}$ and some arbitrary $y_2 \in Y-y_1$. Observe that $y_2$ is well-defined as $\deg(x)\geq 2$. If $y_2x$ satisfies $(i)$ for $T$, we obtain that $\lambda_{y_1}+\lambda_{y_2}\geq 2-2\delta$ by the choice of $y_1$. If $y_2x$ satisfies $(ii)$ for $T$, by the choice of $y_1$, we obtain $\lambda_{y_1}\geq \lambda_{y_2}\geq 1-\delta$. In either case, we have $\lambda_{y_1}+\lambda_{y_2}\geq 2-2\delta$. Now let $T'$ be obtained from $T$ by replacing $p_{y_1}$ by $x$ and replacing $y_2p_{y_2}y_2$ by $y_2p(y_2,x,1-2 \delta)y_2$ if $\delta<\frac{1}{2}$ and by $y_2$ if $\delta=\frac{1}{2}$. Observe that $xy_1$ satisfies Lemma~\ref{char2stops} $(i)$ for $T'$ and $xy_2$ satisfies Lemma~\ref{char2stops} $(ii)$ or $(iii)$ for $T'$. Now consider some $y_3 \in Y-\{y_1,y_2\}$. If $y_3x$ satisfies Lemma~\ref{char1stop} $(ii)$ for $T$, we have $\lambda_{y_3}\geq 1-\delta\geq 1-2 \delta$. If $y_3x$  satisfies Lemma~\ref{char1stop} $(i)$ for $T$, by the choice of $y_1$ and $\lambda_{y_1}<1$, we have $\lambda_{y_3}\geq 2-2\delta-\lambda_{y_1}\geq 1-2 \delta$. This yields that $y_3x$ satisfies Lemma~\ref{char2stops} $(iii)$ for $T'$. Finally, clearly any $e \in E(G)$ which is not incident to $x$ is covered by $T'$. Further, as $T$ is nice, so is $T'$. It hence follows by Lemma~\ref{char2stops} that $T'$ is a \deltatour in $G$. Further, we have $\len(T')-\len(T)=2(1-\lambda_{y_1})-2(\lambda_{y_2}-(1-2\delta))=2(2-2\delta-(\lambda_1+\lambda_2))\leq 0$, so $T'$ is a shortest \deltatour.  Finally, $y_1x$ satisfies $(b)$ for $T'$ and none of $(b),(c)$, and $(d)$ for $T$. As every $e \in E(G)$ that satisfies one of $(b),(c)$, and $(d)$ for $T$, also satisfies one of $(b),(c)$, and $(d)$ for $T'$, we obtain a contradiction to the choice of $T$. Hence $T$ stops at every $x \in V(G)$ with $\deg(x)\geq 2$.

In particular, by Lemma~\ref{char2stops}, every $x_1x_2 \in E(G)$ with $\min\{\deg_G(x_1),\deg_G(x_2)\}\geq 2$ satisfies $(i),(ii)$ or $(iii)$. Consider some  $x_1x_2 \in E(G)$ with $\min\{\deg_G(x_1),\deg_G(x_2)\}\geq 2$. If $x_1x_2$ satisfies $(i)$, it also satisfies $(c)$, so there is nothing to prove. If $x_1x_2$ satisfies $(ii)$ and not $(iii)$, then, as $T$ is nice, we obtain that $T$ stops at $x_1$ and $x_2$ and at no other points on $x_1x_2$. Then $(d)$ holds. We may hence suppose that $x_1x_2$ satisfies $(iii)$. By symmetry, we may suppose that $T$ contains the segments~$\seq{x_1&p(x_1,x_2,\lambda)&x_1}$ for some $\lambda \in [1-2\delta,1)$ and $x_2$. If $\lambda>1-2\delta$, then let $T'$ be obtained from $T$ by replacing $x_1p(x_1,x_2,\lambda)x_1$ by $x_1p(x_1,x_2,1-2\delta)x_1$ if $\delta<\frac{1}{2}$ and by $x_1$ if $\delta=\frac{1}{2}$. Clearly, $x_1x_2$ satisfies Lemma~\ref{char2stops} $(iii)$. Further, as $T$ is nice, so is $T'$. Hence, as $T$ is a \deltatour and by Lemma~\ref{char2stops}, we have that $T'$ is a \deltatour in $G$. As $\len(T')<\len(T)$, we obtain a contradiction to the choice of $T$. It hence follows that $\lambda=1-2\delta$ and so $x_1x_2$ satisfies $(c)$.

We hence obtain that every $e \in E(G)$ satisfies one of $(a)$, $(b)$, $(c)$, and $(d)$ and so the statement holds for~$T$.
\end{proof}
\end{lemma}
We are now ready to conclude \Cref{lemma:discretization}.

\begin{proof}[Proof of \Cref{lemma:discretization}]
Let $G$ be a connected graph and $\delta \geq 0$ a constant. Further, let $S_{G,\delta}$ be the set of points that can be expressed as $p(u,v,\lambda)$ for some $uv \in E(G)$ and $\lambda \in S_\delta$. It suffices to prove that there is a shortest \deltatour in $G$ all of whose stopping points are in $S_{G,\delta}$.

First suppose that $\delta \geq \frac{1}{2}$. It follows from \cref{lem:discretization:gt:one} that there is a shortest \deltatour of $G$ that is either nice or contains at most two stopping points and such that all stopping points of $T$ can be expressed in the form $p(u,v,\lambda)$ for some $uv \in E(G)$ and $\lambda \in S_\delta'$. 

Let $p=p(u,v,\lambda)$ be a stopping point of $T$ with $\lambda \in S_\delta'$. If $\lambda=0$, we clearly have $p \in S_{G,\delta}$. If $\lambda=\frac{1}{2}(\lfloor 2 \delta\rfloor+1-2 \delta)$ and $\delta-\lfloor \delta\rfloor< \frac{1}{2}$, we have $1-\lambda=\frac{1}{2}+\delta-\frac{1}{2}\lfloor 2 \delta\rfloor=\frac{1}{2}+\delta-\lfloor \delta\rfloor=\frac{1}{2}+\delta-\lfloor \delta+\frac{1}{2}\rfloor$, so $p \in S_{G,\delta}$. If $\lambda=\frac{1}{2}(\lfloor 2 \delta\rfloor+1-2 \delta)$ and $\delta-\lfloor \delta\rfloor\geq  \frac{1}{2}$, we have $1-\lambda=\frac{1}{2}+\delta-\frac{1}{2}\lfloor 2 \delta\rfloor=\frac{1}{2}+\delta-(\lfloor \delta\rfloor+\frac{1}{2})=\delta-\lfloor  \delta\rfloor$, so $p \in S_{G,\delta}$.

If $\lambda=\lfloor 2 \delta\rfloor+1-2 \delta$, we have $1-\lambda=2\delta-\lfloor 2\delta\rfloor$, so $p \in S_{G,\delta}$.

If $\lambda=\frac{1}{2}(\lfloor 2 \delta\rfloor+2-2 \delta)$ and $\delta-\lfloor \delta\rfloor< \frac{1}{2}$, we have $1-\lambda=\delta-\frac{1}{2}\lfloor 2 \delta\rfloor=\delta-\lfloor \delta\rfloor$, so $p \in S_{G,\delta}$. If $\lambda=\frac{1}{2}(\lfloor 2 \delta\rfloor+2-2 \delta)$ and $\delta-\lfloor \delta\rfloor\geq  \frac{1}{2}$, we have $1-\lambda=\delta-\frac{1}{2}\lfloor 2 \delta\rfloor=\delta-(\lfloor \delta\rfloor+\frac{1}{2})=\frac{1}{2}+\delta-(\lfloor \delta\rfloor+1)=\frac{1}{2}+\delta-(\lfloor \delta+\frac{1}{2}\rfloor)$, so $p \in S_{G,\delta}$.
\medskip

Now suppose that $\delta < \frac{1}{2}$. It follows from \cref{discreteklein} that there is a shortest \deltatour of $G$ that is either nice or contains at most 2 stopping points and such that all stopping points of $T$ can be expressed in the form $p(u,v,\lambda)$ for some $uv \in E(G)$ and $\lambda \in \{\delta,1-\delta,1-2\delta\}$. 

Let $p=p(u,v,\lambda)$ be a stopping point of $T$ with $\lambda \in S_\delta'$. If $\lambda=\delta$, we have $\lambda=\delta-\lfloor\delta \rfloor$, so $p \in S_{G,\delta}$. If $\lambda=1-\delta$, we have $1-\lambda=\delta-\lfloor\delta \rfloor$, so $p \in S_{G,\delta}$. If $\lambda=1-2\delta$, we have $1-\lambda=2 \delta=2\delta-\lfloor2\delta \rfloor$, so $p \in S_{G,\delta}$.
\end{proof}

\subsection{Immediate Consequences}\label{algosec}
In this section, we conclude~\cref{section:structural} with two useful algorithmic corollaries. 
The following result is an immediate consequence of Lemmas~\ref{nostopchar},~\ref{char1stop}, \cref{lem:tournice} and~\ref{char2stops}.
\begin{corollary}\label{check}
Given a graph~$G$, a constant $\delta$ and a tour~$T$ in $G$, we can decide in polynomial time whether $T$ is a \deltatour.
\end{corollary} 

We further obtain the following result.
\begin{corollary}\label{decide}
Given a graph~$G$ and a constant $\delta>0$, there is an algorithm that computes a shortest \deltatour in $G$ and runs in $f(n)$.
\end{corollary}
\begin{proof}
By \cref{lemma:discretization}, it suffices to consider tours all of whose stopping points are contained in $S_{G,\delta}$ where $S_{G,\delta}=\{p(u,v,\lambda) \mid uv \in E(G),\lambda \in S_\delta\}$. As a nice tour traverses every edge at most twice and stops at every point at most once, we obtain that it stops at most $\alpha n^2$ times for an absolute constant $\alpha$. We now enumerate all possible sequences of points in $S_{G,\delta}$ of length at most $\alpha n^2$. For each of these sequences, we first check whether it is a tour which is clearly possible in linear time. We next check whether it is a \deltatour which is possible by Corollary~\ref{check}. If this is the case, we compute its length. We output the shortest \deltatour we obtain during this procedure, which is a shortest \deltatour.
\end{proof}

\section{Approximation Algorithms}
\label{section:approximation-ub}
In this section, we deal with finding approximation algorithms for shortest \deltatour{}s. The behavior of this problem is strongly influenced by the range $\delta$ comes from.

The simplest case is when $\delta = 0$. 

Let $G$ be a connected graph. For a $0$-tour $T$ of $G$, we have $\dist_G(p, T) = 0$ for every
$p \in P(G)$. That is, $P(T) = P(G)$. A Chinese Postman tour is an integral tour in $G$ traversing every
edge at least once. The crucial insight for finding a shortest 0-tour in $G$ is that this problem is essentially equivalent
to finding a shortest Chinese Postman tour in $G$. Since a shortest Chinese Postman tour can
be found in polynomial time (see, e.g.,~\cite{3134208,EdmondsJ73}), the problem \deltatourprob[0] is 
also polynomial-time solvable. Formally, we prove the following result.
\begin{observation}\label{obs:chinese-postman-zero-tour}
	There is a polynomial-time algorithm solving \deltatourprob[0].
\end{observation}
\begin{proof}
Let $G$ be a connected graph.
We compute and output a shortest Chinese Postman Tour $T_{\text{CP}}$ of $G$.
This can be done in polynomial time; see for example~\cite{3134208,EdmondsJ73}.

Let $\opt_{\text{CP}}$ be the length of $T_{\text{CP}}$,
	and let $\opttour[0]$ be the length of a shortest \deltatour[0].
Now we observe that a Chinese Postman tour given by a sequence
	of adjacent vertices that traverses every edge is a $0$-tour of the same
	length, so $\opttour[0] \leq  \opt_{\text{CP}}$.
Conversely, it follows directly from \cref{discreteklein} that there is a
shortest \deltatour[0] that is a Chinese Postman tour, so $\opt_{\text{CP}} \leq \opttour[0]$.
\end{proof}

For any $\delta > 0$, the problem \deltatourprob unfortunately becomes
\np-hard and even \apx-hard as we show in~\cite{FreiGHHM24}.
Therefore, we resort to the design of approximation algorithms for each $\delta > 0$.

We consider the ranges in increasing order. More concretely, in Sections~\ref{sec:appro1},~\ref{sec:appro2},~\ref{sec:appro3},~\ref{sec:appro4} and~\ref{sec:appro5}, we consider the cases $\delta \in (0,1/6]$, $\delta \in (1/6,1/2)$, $\delta=1/2$, $\delta \in (1/2,33/40)$, and $\delta \in [33/40,3/2)$, respectively. The remaining part of this section is dedicated to dealing with large $\delta$. In Section~\ref{sec:appro6}, we give the construction of an auxiliary graph which will be helpful for this purpose. We exploit this in Section~\ref{sec:appro7}, where we give an approximation result for fixed, large $\delta$ and in Section~\ref{sec:appro8}, where we deal with the case that $\delta$ is part of the input.

\subsection[\texorpdfstring{Covering Range $\delta \in (0, 1/6]$}{Covering Range delta in (0, 1/6]}]{Covering Range \boldmath$\delta \in (0, 1/6]$}\label{sec:appro1}

We start with the smallest covering range, that is, $\delta \in (0,1/6]$. It turns out that in this range,
	a shortest Chinese Postman Tour is a good approximation of a \deltatour.
That is, our algorithm consists of computing a shortest Chinese Postman tour.
Intuitively speaking, for this small $\delta$, there is a nice shortest \deltatour
	that stops at every vertex and passes large parts of every edge. We can hence show that
	a shortest Chinese Postman Tour is only larger by small factor than such a \deltatour.
Observe that the approximation ratio of our algorithm approaches $1$ when $\delta$ goes to $0$. More precisely, we prove Theorem \ref{ThmApproxUbZeroSixth}, which we restate here for convenience.

\ThmApproxUbZeroSixth*\label\thisthm
\begin{proof}

Let $G$ be a connected graph. By Corollary \ref{decide}, we may suppose that $|V(G)|\geq 3$.  We compute and output a shortest Chinese Postman Tour
	$T_{\text{CP}}$ of $G$.
Let $\opt_{\text{CP}}$ be the length of $T_{\text{CP}}$,
	and let $\opttour$ be the length of a shortest \deltatour.

	By definition, $T_{\text{CP}}$ traverses every edge and hence is a \deltatour.
	It remains to bound its length.
	By \cref{discreteklein} and as $|V(G)|\geq 3$, there is a
		minimum length nice \deltatour $T=\seq{p_0&p_1&\dots&p_z}$ which, for every edge $uv \in E(G)$,
		either traverses $uv$ or contains the segment $\seq{u'&p(u',v,'\lambda)&u'}$
		for some $u',v'$ with $\{u',v'\}=\{u,v\}$ and some $\lambda \geq 1-2\delta$.
It follows that the extension $\ceil{T}$ of $T$ is a Chinese Postman tour and that
	$\len(\ceil{T})\leq \frac{1}{1-2\delta}\len(T_{\text{CP}})$ holds.
Hence $\len(T_{\text{CP}}) = \opt_{\text{CP}}  \leq \len(\ceil{T}) \leq \len(T)/({1-2\delta})=
	{\opttour}/({1-2\delta})$.
\end{proof}

\subsection[\texorpdfstring{Covering Range $\delta \in (1/6, 1/2)$}{Covering Range, delta in (1/6, 1/2)}]{Covering Range \boldmath$\delta \in (1/6, 1/2)$}\label{sec:appro2}

We next consider the case that $\delta\in(1/6, 1/2)$.
It turns out that if $\delta$ is in this range, we can benefit from a close connection to a deeply studied related algorithmic problem, namely \MetricTSP (where TSP stands for Traveling Salesman Problem).
Formally, given an edge weighted graph~$(H,w)$ and a tour $T=\seq{p_0&\dots&p_z}$ in $H$ such that $p_i \in V(H)$ for all $i \in [z]$,
we define the length of a tour as $\len(T)=\sum_{i \in [z]}w(p_{i-1}p_i)$.
The problem can now be defined as follows.
\begin{myproblem}[\MetricTSP]
\label{prob:metrictsp}%
Instance&A connected graph~$H$ with a weight function $w\colon E(H)\rightarrow \mathbb{R}_{\geq 0}$.\\
Solution&Any tour $T$ in $H$ stopping at all vertices of $V(H)$.\\
Goal& Minimize the tour length $\len(T)$.
\end{myproblem}
We acknowledge that \MetricTSP is often defined in a slightly different form where $H$ is restricted to be complete, $w$ is required to be metric and a Hamiltonian cycle is sought rather than an arbitrary tour stopping at all vertices. However, both versions are easily seen to be equivalent and so we use the above version which is more convenient for our purposes.

We heavily rely on a well-known result of Christofides stating the existence of a $1.5$-approximation algorithm for metric \TSP. We use the following more formal restatement of the result of Christofides.
\begin{lemma}[\cite{Christofides76, Christofides2022}]
\label{lem:tsp_approx}
There is a polynomial-time algorithm that computes a TSP tour~$T$ of a given
weighted connected graph~$(H,w)$ such that $\len(T) \leq 1.5 \cdot \opttsp$,
where $\opttsp$ is the length of a shortest TSP tour of $(H,w)$.
\end{lemma}

We are now ready to describe our algorithm which is based on a reduction to an instance of \MetricTSP and applying Lemma~\ref{lem:tsp_approx}. Some results from Section~\ref{section:structural} will be convenient in the analysis of the quality of the algorithm. Formally, we prove Theorem~\ref{thm:approx:ub:sixth_half}, which we restate here for convenience.
\ThmApproxUbSixthHalf*\label\thisthm

\begin{proof}
Let a connected graph~$G$ and $\delta \in (1/6,1/2)$ be given. By Corollary~\ref{decide}, we may suppose that $|V(G)|\geq 3$. Let $V_1$ be the vertices $v \in V(G)$ with $d_G(v)=1$ and let $V_{\geq 2}=V(G)-V_1$. We now construct an auxiliary weighted graph~$(H,w)$. First, we let $V(H)$ contain $V_{\geq 2}$. Next, for every $uv \in E(G)$ with $u \in V_{\geq 2}$ and $v \in V_1$, we let $V(H)$ contain the point $p(u,v,1-\delta)$ and we let $E(H)$ contain the edge $up(u,v,1-\delta)$. For the remaining description of $H$, we need to make a finer distinction of the range $\delta$ comes from. First suppose that $\delta<\frac{1}{4}$. For every $uv \in E(G)$ with $\{u,v\}\subseteq V_{\geq 2}$, we let $V(H)$ contain $p(u,v,2 \delta)$ and $p(v,u,2 \delta)$ and we let $E(H)$ contain $up(u,v,2\delta), p(u,v,2\delta)p(v,u,2\delta)$, and $p(v,u,2\delta)v$. Now suppose that $\delta \geq \frac{1}{4}$ and let $\{v_1,\ldots,v_q\}$ be an arbitrary ordering of $V_{\geq 2}$. For every $i,j \in [q]$ with $i<j$ and $v_iv_j \in E(G)$, we let $V(H)$ contain $p(v_i,v_j,2\delta)$ and we let $E(H)$ contain $v_ip(v_i,v_j,2\delta)$ and $p(v_i,v_j,2\delta)v_j$. Finally, we define $w\colon E(H)\rightarrow \mathbb{R}_{\geq 0}$ by $w(pp')=\dist_G(p,p')$ for all $pp' \in E(H)$. This finishes the description of $(H,w)$. 

The tight relationship between $\delta$-tours in $G$ and \TSP-tours in $H$ is described in the following two claims.

\begin{claim}\label{deltatsp}
Let $T$ be a $\delta$-tour in $G$. Then there is a \TSP-tour~$T_0$ in $H$ with $\len(T_0)\leq \len(T)$.
\end{claim}
\begin{claimproof}
By Lemma~\ref{discreteklein} and the assumption that $|V(G)|\geq 3$, we may suppose that one of Lemma~\ref{discreteklein}$(b),(c)$, and $(d)$ holds for every $uv \in E(G)$ for $T$, so in particular $T$ is nice. We now obtain a tour~$T_0$ from $T$ in the following way. If $\delta<\frac{1}{4}$, for all edges $uv \in E(G)$ for which $(c)$ holds, we replace all segments~$\seq{u&v}$ and $\seq{v&u}$ of $T$ by $\seq{u&p(u,v,2 \delta)&p(v,u,2\delta)&v}$ and $\seq{v&p(v,u,2 \delta)&p(u,v,2 \delta)&u}$, respectively. If $\delta\geq \frac{1}{4}$, for all edges $v_iv_j \in E(G)$ with $i<j$ for which $(c)$ holds, we replace all segments~$\seq{v_i&v_j}$ and $\seq{v_j&v_i}$ of $T$ by $\seq{v_i&p(v_i,v_j,2\delta)&v_j}$ and $\seq{v_j&p(v_i,v_j,2\delta)&v_i}$, respectively. If $\delta<\frac{1}{4}$ and $(d)$ holds for an edge $uv$ of $G$, say $T$ contains the segment~$\seq{u&p(v,u,2 \delta)&u}$, we replace the segment~$\seq{u&p(v,u,2 \delta)&u}$ by $\seq{u&p(u,v,2\delta)&p(v,u,2 \delta)&p(u,v,2\delta)&u}$. Finally, if $\delta \geq \frac{1}{4}$ and for some $v_iv_j \in E(G)$ with $i<j$, $(d)$ holds and $T$ contains the segments~$\seq{v_i&p(v_j,v_i,2 \delta)&v_i}$ and $\seq{v_j}$, then we replace these segments by $\seq{v_i}$ and $\seq{v_j&p(v_i,v_j,2 \delta)&v_j}$, respectively. This finishes the description of $T_0$. By construction, we have $\len(T_0)=\len(T)$. Observe that, as $T$ is nice, so is $T_0$. Further, by Lemmas~\ref{char1stop} and~\ref{char2stops} and construction, we obtain that $T_0$ is a \TSP-tour in $H$.
\end{claimproof}

\begin{claim}\label{tspdelta}
Let $T_0$ be \TSP-tour in $H$. Then $T_0$ is a \deltatour of the same length in $G$.
\end{claim}
\begin{claimproof}

It follows directly by construction that $T_0$ is a tour of the same length in $G$. In order to see that $T_0$ is a \deltatour in $G$, it suffices to prove that for every $p \in P(G)$ which is not passed by $T$, there is some $q \in V(H)$ with $\dist_G(p,q)\leq \delta$. First suppose that $p=p(u,v,\lambda)$ for some $u \in V_{\geq 2}$ and $v \in V_1$ with $uv \in E(G)$ and some $\lambda \in [0,1]$. As $\deg_H(p(u,v,1-\delta))=1$, we obtain that $T_0$ contains the segment~$\seq{u&p(u,v,1-\delta)&u}$. Hence, as $p$ is not passed by $T_0$, it follows that $\lambda>1-\delta$. Therefore, $\dist_G(p,p(u,v,1-\delta))\leq \delta$ holds. Now suppose that $p=p(u,v,\lambda)$ for some $u,v \in V_{\geq 2}$ with $uv \in E(G)$. First suppose that $\delta<\frac{1}{4}$. If $\lambda\leq \delta$, we have $\dist_G(p,u)\leq \delta$. If $\lambda \in (\delta,\frac{1}{2})$, as $\delta \geq \frac{1}{6}$, we have $\dist_G(p,p(u,v,2\delta))\leq \delta$. If $\lambda \in [\frac{1}{2},1-\delta)$, as $\delta \geq \frac{1}{6}$, we have $\dist_G(p,p(v,u,2\delta))\leq \delta$. If $\lambda \geq 1-\delta$, we have $\dist_G(p,v)\leq \delta$. Now suppose that $\delta \geq \frac{1}{4}$. By symmetry, we may suppose that $u=v_i$ and $v=v_j$ for some $i,j \in [q]$ with $i<j$. If $\lambda \leq \delta$, we have $\dist_G(p,u)\leq \delta$. If $\lambda \geq \delta$ and $\lambda \leq 1-\delta$, as $\delta \geq \frac{1}{4}$, we have $\dist_G(p,p(u,v,2 \delta))\leq \delta$. Finally, if $\lambda \geq 1-\delta$, we have $\dist_G(p,v)\leq \delta$. Hence $T_0$ is a \deltatour. 
\end{claimproof}
We are now ready to conclude Theorem~\ref{thm:approx:ub:sixth_half}. Let
$\opttour$ be the length of a shortest \deltatour in $G$. By
Lemma~\ref{lem:tsp_approx}, in polynomial time, we can compute a TSP tour~$T$ of
$H$ such that  $\len(T) \leq 1.5 \cdot \opttsp$, where $\opttsp$ is the
length of a shortest TSP tour of $(H,w)$. By Claim~\ref{tspdelta}, we have that
$T$ is a \deltatour in $G$. Moreover, by Claim~\ref{deltatsp}, we have
$\opttsp\leq \opttour$. This yields $\len(T)\leq 1.5 \cdot \opttsp\leq
1.5 \cdot \opttour$.
\end{proof}
 
\subsection[\texorpdfstring{Covering Range $\delta = 1/2$}{Covering Range delta = 1/2}]{Covering Range \boldmath$\delta = 1/2$}\label{sec:appro3}
In this section, we show that for $\delta=\frac{1}{2}$, we can give an approximation algorithm with a slightly better bound. This is due to an approximation result on a more restricted version of \MetricTSP. Namely, we denote by {\it Graphic TSP} the restriction of \MetricTSP to instances $(H,w)$ in which $w$ is the unit weight function. Throughout this, section, when we speak of a TSP tour of a certain length, we refer to the unit weight functions. We rely on the following result of Seb\H o and Vygen.

\begin{lemma}[\cite{SeboV14}]\label{lem:graphtsp_approx}
There is a polynomial-time algorithm that computes a TSP tour~$T$ of a given
connected graph~$H$ such that $\len(T) \leq 1.4 \cdot \opttsp$,
where $\opttsp$ is the length of a shortest TSP tour of $H$.
\end{lemma}

We are now ready to give an approximation algorithm for $\frac12$-tour which attains the same approximation ratio as the algorithm in Lemma~\ref{lem:graphtsp_approx}. More formally, we prove Theorem~\ref{thm:approx:ub:half}, which we restate here for convenience. 
\ThmApproxUbHalf*\label\thisthm
\begin{proof}
Let $G$ be a connected graph. By Corollary~\ref{decide}, we may suppose that $|V(G)|\geq 3$. We denote by $V_1$ the set of vertices $v \in V(G)$ with $d_G(v)=1$ and we use $V_{\geq 2}=V(G)-V_1$. Further, we set $H=G[V_{\geq 2}]$. The idea of our proof is that every $\frac12$-tour in $G$ consists of a TSP tour in $H$ together with some detours for the vertices in $V_1$. We can efficiently approximate the first part by Lemma~\ref{lem:graphtsp_approx} while the detours result in a positive additive constant. In the following two claims, this relations is proven more formally.

\begin{claim}\label{einhalbtsp2}
Let $T$ be a $\frac12$-tour in $G$. Then there is a TSP-tour~$T_0$ in $H$ with $\len(T_0)=\len(T)-|V_1|$.
\end{claim}
\begin{claimproof}
By Lemma~\ref{discreteklein}, we may suppose that $T$ stops at every vertex in $V_{\geq 2}$ and that for every $v \in V_1$, we have that $T$ contains the segment~$\seq{x_v&p(x_v,v,\frac{1}{2})&x_v}$ where $x_v$ is the unique neighbor of $v$ in $G$. We now obtain $T_0$ from $T$ by recursively replacing the segment~$\seq{x_v&p(x_v,v,\frac{1}{2})&x_v}$ by $\seq{x_v}$ for all $v \in V_1$. As $T$ stops at every vertex of $V_{\geq 2}$, we obtain that $T_0$ is a TSP tour in $H$. Further, by construction, we have $\len(T_0)=\len(T)-|V_1|$.
\end{claimproof}

\begin{claim}\label{einhalbtsp}
Let  a TSP tour $T_0$ in $H$ be given. Then, in polynomial time, we can compute a $\frac12$-tour~$T$ in $G$ with $\len(T)=\len(T_0)+|V_1|$.
\end{claim}
\begin{claimproof}
Clearly, we may suppose that all stopping points of $T_0$ are vertices in $V(G)$. For every $v \in V_1$, let $x_v$ be the unique neighbor of $v$ in $G$. As $T_0$ is a \deltatour in $H$, we know that $T_0$ stops at $x_v$ for all $v \in V_1$. We now obtain $T$ from $T_0$ by, for every $v \in V_1$ choosing an arbitrary occurrence of $x_v$ and replacing it by $x_vp(x_v,v,\frac{1}{2})x_v$. We do this recursively for an arbitrary ordering of $V_1$. Clearly, $T$ can be computed in polynomial time given $T_0$. Further, as all stopping points of $T_0$ are in $V(G)$, we obtain that $T$ is nice. It follows from \Cref{char1stop,char2stops} that $T$ is a $\frac12$-tour in $G$. Further, we clearly have $\len(T)=\len(T_0)+|V_1|$.
\end{claimproof}
We are now ready to conclude Theorem~\ref{thm:approx:ub:half}. First, by
Lemma~\ref{lem:graphtsp_approx}, we can compute a TSP tour~$T_0$ in $H$ which
satisfies $\len(T_0)\leq \frac{7}{5}\opttsp$ where $\opttsp$ is the length of a
shortest TSP tour in $H$. Then by Claim~\ref{einhalbtsp}, in polynomial time, we
can compute a $\frac12$-tour~$T$ in $G$ which satisfies
$\len(T)=\len(T_0)+|V_1|$. Finally, by Claim~\ref{einhalbtsp2}, we have
$\opttsp\leq \opttour-|V_1|$ where $\opttour$ is the length of a shortest $\frac12$-tour in $G$.
We obtain $\len(T)=\len(T_0)+|V_1|\leq 1.4 \cdot \opttsp+|V_1|\leq 1.4 (\opttour-|V_1|)+|V_1|\leq
1.4 \cdot \opttour$. This finishes the proof. 
\end{proof}
\subsection[\texorpdfstring{Covering Range $\delta\in(1/2, 33/40)$}{Covering Range delta in (1/2, 33/40)}]{Covering Range \boldmath$\delta\in(1/2, 33/40)$}\label{sec:appro4}

Here we show that computing a $\frac12$-tour using
\Cref{thm:approx:ub:half}, yields a good approximation of a \deltatour.

We first need the following result relating the lengths of $\delta$-tours and $\frac12$-tours.
\begin{proposition}\label{dtrtzfzu}
Let $T_\delta$ be a $\delta$-tour in a connected graph~$G$ for some $\delta\in(1/2, 33/40)$. Then there is a $\frac12$-tour~$T_{1/2}$ of $G$ that satisfies $\len(T_{1/2})\leq \frac{1}{2(1-\delta)}\len (T_\delta)$.
\end{proposition}
\begin{proof}
If $\alpha(T_\delta)\leq 2$, we obtain by Lemmas~\ref{nostopchar} and~\ref{char1stop} that $|V(G)|\leq 2$. In that case, there clearly exists a $\frac12$-tour of length 0 in $G$ and so the statement follows. We may hence suppose that $\alpha(T_\delta)\geq 3$.
Therefore, by Lemma~\ref{lem:discretization:gt:one}, we may suppose that $T_\delta$ is nice and, as $\floor{2 \delta}=1$, that for all non-integral stopping points $p$ of $T_\delta$, there are some $uv\in E(G)$ and $\lambda \in \{1-\delta,2-2\delta,3/2-\delta\}$ such that $p=p(u,v,\lambda)$ and $T_\delta$ contains the segment~$\seq{u&p&u}$.
 We obtain $T_{1/2}$ from $T_\delta$ by replacing every segment of $T_\delta$ of the form $up(u,v,1-\delta)u$ for some $u,v \in V(G)$ by $up(u,v,\frac{1}{2})u$ and replacing every segment of $T_\delta$ of the form $up(u,v,\lambda)u$ for some $u,v \in V(G)$ and some $\lambda \in \{2-2\delta,\frac{3}{2}-\delta\}$by $uvu$. Clearly, $T_{1/2}$ is a tour in $G$ and $\len(T_{1/2})\leq \frac{1}{2(1-\delta)}\len (T_\delta)$. We still need to prove that $T_{1/2}$ is a $\frac12$-tour in $G$. First observe that, as $T_\delta$ is nice, so is $T_{1/2}$. Let $uv \in E(G)$. If $T_\delta$ stops at both $u$ and $v$, then $T_{1/2}$ stops at both $u$ and $v$, and so $uv$ is covered by $T_{1/2}$ by Lemma~\ref{char2stops}. Next suppose that $T_\delta$ stops at exactly one of $u$ and $v$, say $u$. By Lemma~\ref{char1stop}, one of Lemma~\ref{char1stop} $(i)$, and $(ii)$ holds for $T_\delta$. If Lemma~\ref{char1stop}$(ii)$  holds, then $T_{1/2}$ stops at $u$ and $p(u,v,\lambda)$ for some $\lambda \in \{\frac{1}{2},1\}.$ It follows from Lemmas~\ref{char1stop} and~\ref{char2stops} that $uv$ is covered by $T_{1/2}$. If Lemma~\ref{char1stop}$(i)$ holds and $(ii)$ does not hold, we obtain that $T_\delta$ contains the segment~$\seq{w&p(w,v,\lambda)&w}$ for some $w \in N_G(v)$ and some $\lambda \in \{2-2\delta,3/2-\delta\}$. By construction, it follows that $T_{1/2}$ stops at $u$ and $v$ and so $uv$ is covered by $T_{1/2}$ by Lemma~\ref{char2stops}. Finally, suppose that $T_\delta$ stops at none of $u$ and $v$. By Lemma~\ref{nostopchar}, there are stopping points $p(u',u,\lambda_u)$ and $p(v',v,\lambda_v)$ of $T_\delta$ with $\lambda_u+\lambda_v \geq 3-2 \delta$. By assumption, we obtain that $\lambda_u=\lambda_v=\frac{3}{2}-\delta$. It follows by construction that $T_{1/2}$ stops at $u$ and $v$. We obtain by Lemma~\ref{char2stops} that $uv$ is covered by $T_{1/2}$. It follows that $T_{1/2}$ is a $\frac12$-tour in $G$. 
\end{proof}
We are now ready to prove the main result of this section, which is Theorem~\ref{thm:approx:ub:half:threequarters}. It is restated here for convenience.
\ThmApproxUbHalfThreeQuarters*\label\thisthm
\begin{proof}
We use Theorem~\ref{thm:approx:ub:half} to compute a
$\frac12$-tour~$T_{1/2}$ of $G$ which satisfies
$\len(T_{1/2})\leq \frac{7}{5}\len (T_{1/2}^*)$ where
$T_{1/2}^*$ is a shortest $\frac12$-tour in $G$. Observe that
$T_{1/2}$ is in particular a $\delta$-tour in $G$. In order to prove the
quality of $T_{1/2}$, let $T_\delta^*$ be a shortest $\delta$-tour in $G$.
By Proposition~\ref{dtrtzfzu}, there is a
$\frac12$-tour~$T_{1/2}'$ with $\len(T_{1/2}')\leq \frac{1}{2(1-\delta)}\len (T_\delta^*)$.
We obtain $\len(T_{1/2})\leq 1.4 \cdot \len(T_{1/2}^*)\leq  1.4 \cdot \len(T_{1/2}')\leq
\frac{1.4}{2(1-\delta)}\len (T_\delta^*)$, so the statement follows.
\end{proof}

\subsection[\texorpdfstring{Covering Range $\delta \in [33/40, 3/2 )$}{Covering Range delta in [33/40, 3/2)}]{Covering Range \boldmath$\delta \in \left[33/40, 3/2 \right)$}
\label{sec:appro5}

As mentioned in \cref{section:overview},
in this range, we use a different approach based on a linear program (LP)
first considered by Könemann et al.~\cite{KonemannKPS03} for 
	computing a shortest vertex cover tour, which is a tour such that the vertices this tour stops at form a vertex cover of the input graph.
We first review the LP formulation and then show how it can be used to obtain
approximation algorithms for \deltatourprob in this range.

Given a graph~$G$, we let $\mathcal{F}(G)$ be the set
of subsets of $V(G)$ such that both $G[S]$ and $G[V(G) \setminus S]$
contain at least one edge.
For some $F\subseteq V(G)$, let $C_G(F)$ denote the set of edges in $G$ with exactly one endpoint
in $F$.
The LP is then formulated in~\cite{KonemannKPS03} as follows.

\begin{equation*}
\label{lp:toc} \tag{1}
\boxed{\begin{array}{ll@{}ll}
\text{Minimize} & \displaystyle\sum\limits_{e \in E(G)} z_e &\\
\text{subject to}& \displaystyle\sum\limits_{e \in C_G(F)}z_e\geq 2 &\text{ for all $F \in \mathcal{F}(G)$}\text{ and }\\
0 \leq z_e \leq 2 &\text{ for all $e \in E(G)$.}
\end{array}}
\end{equation*}

For some graph~$G$, we denote by $\optlp(G)$ the optimum value of~\ref{lp:toc} defined with respect to $G$. We heavily rely on
\Cref{cor:koenemann_ovrvw} which was proven by Könemann et al.~in \cite{KonemannKPS03}, which we restate below for convenience.

\KoenemannCor*\label\thisthm

To show that the vertex cover tour computed via \Cref{cor:koenemann_ovrvw}
yields a $3$-approximation for $1$-\tour,
the main observation is that for a given connected graph~$G$, we have that the length of any 1-tour in $G$ is at least $\optlp(G)$.
More precisely, we prove the following lemma.

\begin{lemma}
\label{lplowerbound}
Let $G$ be a connected graph and $T_1$ a 1-tour in $G$. Then $\len(T_1)\geq\optlp(G)$.
\end{lemma}

\begin{proof}
Let $T_1 = (p_0, \dots, p_k)$. If $k=0$, we obtain that $G$ contains a single vertex that dominates $V(G)$. In that case, the all-zeros vector is a feasible solution for ~(\ref{lp:toc}) and so the statement clearly follows. If $k=2$ and $G$ does not contain a vertex that dominates $V(G)$, as $T_1$ is a 1-tour, we obtain that $T_1=uvu$ for an edge $uv \in E(G)$ such that $\{u,v\}$ is a dominating set of $G$. It follows that the vector defined by $x_{uv}=2$ and $x_e=0$ for all $e \in E(G)-uv$ is a feasible solution for ~(\ref{lp:toc}). As $\len(T_1)=2$, the statement follows. 

We may hence suppose that $k\geq 3$ and hence, by Lemma~\ref{lem:tournice},  that $T_1$ is nice.

For every edge $uv \in E(G)$, we define 
\[
	\Lambda_{e} \coloneqq \sum\limits_{i \in [k]\colon P(p_{i-1},p_{i}) \subseteq P(u,v)} 
		d_G(p_i, p_{i+1})\text{.}
\]
\begin{claim}
Let $z \in [0,2]^{E(G)}$ be the vector defined by $z_e = (\min(2, \Lambda_e))$ for all $e \in E(G)$. Then $z$ is feasible for~(\ref{lp:toc}).
\end{claim}
\begin{claimproof}
By construction, we have $z \in [0,2]^{E(G)}$.

Now fix an arbitrary $F \in \mathcal{F}(G)$. We make a distinction between two
cases:

\setcounter{Case}{0}
\begin{Case}
$T_1$ traverses an edge $uv \in C_G(F)$.
\end{Case}

By symmetry, we may suppose that $u \in F$
	and $v \notin F$. Then, as $T_1$ is a
	tour, $T_1$ either traverses $uv$ twice or traverses some edge $u'v' \in C_G(F)-uv$. In the first case, we have $z_{uv}= 2$ and in the second case, we have $\min\{z_{uv},z_{u'v'}\}\geq 1$. In either case, we obtain $\sum\limits_{e \in C_G(F)}z_e\geq 2$.
\begin{Case}
$T_{\deltatour[1]}$ does not traverse any edge of $C_G(F)$.
\end{Case}
We obtain that the set of vertices in $V(G)$ which are stopping points of $T_1$ is fully contained in one of $F$ and $V(G)-F$. Without
	loss of generality, we may suppose that the former is the case.
	Fix some edge $uv \in E(G[V(G) \setminus F])$, which exists
	by definition of $\mathcal{F}$.
	It follows from \cref{nostopchar} that
	$T_1$ stops at points $p(u', u, \lambda_1)$
	and $p(v', v', \lambda_2)$ for some $u' \in N_G(u),
	v' \in N_G(v)$ such that 
	$\lambda_1 + \lambda_2 \geq 1$. By $k\geq 2$ and assumption, we obtain that $\{u',v'\}\subseteq F$ and hence $\{uu',vv'\}\subseteq C_G(F)$.
	Thus, $\Lambda_{u' u} + \Lambda_{v' v}  \geq 2$,
	so 
	$\sum\limits_{e \in C_G(F)}{z_e} \geq
	\min(2, \Lambda_{u' u}) + \min(2, \Lambda_{v' v}) \geq 2$.
\end{claimproof}

Finally, observe that
$\len(T_1) =
	\sum_{i \in [k]}{\dist_G(p_{i-1}, p_{i})} \geq
	\sum_{e \in C_G(F)}{\Lambda_e} \geq \sum_{e \in C_G(F)}{z_e}\geq \optlp(G)$.
\end{proof}

We are now ready to prove our result for $\delta=1$. This result will be the crucial ingredient for the proof of Theorems~\ref{thm:approx:ub:one:threehalves} and~\ref{thm:approx:ub:threequarters:one}.

\begin{lemma}
\label{lem:one_tour_approx}
There is a polynomial-time $3$-approximation algorithm for $1$-\tour.
\end{lemma}
\begin{proof}
Let a connected graph~$G$ be given. By \cref{cor:koenemann_ovrvw}, in polynomial time, we can compute a vertex cover tour~$T$ of $G$ that satisfies $\len(T)\leq 3 \optlp(G)$. Clearly, we may suppose that $T$ is nice. It follows directly from Lemmas~\ref{char1stop} and~\ref{char2stops} that $T$ is a 1-tour in $G$. Moreover, by Lemma~\ref{lplowerbound}, we have that $\optlp(G)\leq \opttour[1]$, where $\opttour[1]$ is the length of a shortest $1$-tour in $G$. This yields $\len(T)\leq 3 \optlp(G)\leq 3 \opttour[1]$.
\end{proof}

In order to generalize \cref{lem:one_tour_approx} for 
$\delta > 1$, we prove the following relation between shortest 1-tours and shortest $\delta$-tours for more general $\delta$.

\begin{lemma}
\label{lem:make1} 
Let $G$ be a connected graph, $\delta \in (1, 3/2)$ be a real,
and $T_\delta$ be a nice shortest \deltatour in a $G$ with $\alpha(T_\delta) \geq 3$.
Then, there is a \deltatour[1] $T_1$ of $G$ of length
$\len(T_1) \leq \len(T_\delta)/(3-2\delta)$.
\end{lemma}
\begin{proof}
By \cref{lem:discretization:gt:one}, and
since $\floor{2\delta} = 2$, we may assume that for all non-integral stopping points
$p$ of $T_\delta$, there are some $uv \in E(G)$ and $\lambda \in \{\frac{3}{2} - \delta, 3-2\delta, 2-\delta\}$ such that $p=p(u,v,\lambda)$ and $T_\delta$ contains the segment
$\seq{u&p&u}$.
We obtain $T_1$ from $T_\delta$ by replacing every segment of $T_\delta$ of
the form $\seq{u&p(u, v, \frac{3}{2}-\delta)&u}$ for some $uv \in E(G)$ by
$\seq{u&p(u, v, \frac{1}{2})&u}$ and replacing every segment of $T_\delta$
of the form $\seq{u&p(u, v, \lambda)&u}$ where 
$\lambda \in \{3-2\delta, 2-\delta\}$ by $\seq{u&v&u}$.
Note that
$\len(T_1) \leq
	\max\{1/(3-2\delta), 1/(2-\delta)\} \len(T_\delta) =
	\len(T_\delta)/(3-2\delta)$,
where the last inequality is due to $\delta > 1$.

Clearly, we have that $T_1$ is a tour. It remains to show that $T_1$ is a \deltatour[1].
Fix an arbitrary edge $uv \in E(G)$.
If $T_\delta$ stops at $u$ or $v$, then so does $T_1$ and hence $uv$ is covered.
Suppose then that $T_\delta$ stops at neither $u$ nor $v$.
By \cref{nostopchar}, as $\alpha(T_\delta) \geq 3$,
for all $i \in \{1,2\}$, there are
stopping points $p_i = p(u_i, v_i, \lambda_i)$ on $T_\delta$ with $\lambda_i
\in [0, 1)$ and $\lambda_1 + \lambda_2 \geq \dist_G(u, v_1) + \dist_G(v, v_2) +
3 - 2 \delta$.
Observe that $\dist_G(u, v_1) + \dist_G(v, v_2) \leq \lambda_1 + \lambda_2 - 3 + 2\delta < 2$ as
$\delta < \frac{3}{2}$.
By symmetry, we may suppose that one of the following two cases occurs:
\begin{enumerate}
\item $\dist_G(u, v_1) = \dist_G(v, v_2) = 0$, so $u = v_1$ and $v = v_2$.
	In this case, $\lambda_1 + \lambda_2 \geq 3-2\delta$. Since
	$\lambda_i \in \{\frac{3}{2} - \delta, 3-2\delta, 2-\delta\}$,
	$T_1$ stops at points $p_1' = p(u_1, u, \lambda_1')$ and
	$p_2' = p(u_2, v, \lambda_2')$, where
	$(\lambda_1', \lambda_2') \in \{
	(\frac{1}{2}, \frac{1}{2}), 
	(0, 1), 
	(1, 0), 
	(1, 1)\}$. If $\lambda_1' = \lambda_2' = \frac{1}{2}$, then
	$uv$ is covered by \cref{nostopchar} (i) as $\lambda_1' + \lambda_2' \geq 1$.
	In all other cases, $T_1$ stops at $u$ or $v$, covering $uv$ as
	$\delta > 1$.
\item $\dist_G(u, v_1) = 0$ and $\dist_G(v, v_2) = 1$, so $u = v_1$. In this
	case, $\lambda_1 + \lambda_2 \geq 4-2\delta$.
	Since $\max{\{\frac{3}{2} - \delta, 3-2\delta, 2-\delta\}} = 2-\delta$, we have
	$\lambda_1 = \lambda_2 = 2-\delta$. Thus, $T_1$ stops at
	$u$, so $T_1$ covers $uv$.
\end{enumerate}

It follows that $T_1$ covers $uv$ and is, hence, a \deltatour[1] of $G$.
\end{proof}

This yields the following theorem.

\ThmApproxUbOneThreeHalves*\label\thisthm
\begin{proof}
	Given a connected graph~$G$, let $T_\delta$ be a 
	shortest \deltatour of $G$ all of whose stopping points are in $S_\delta$ whose existence is implied by \cref{lem:discretization:gt:one}.

	If $\alpha(T_\delta) \leq 2$, then there is an edge $uv \in E(G)$ and
	points $p_0 = p(u, v, \lambda_0)$ and $p_1 = p(u, v, \lambda_1)$
	with $\lambda_0, \lambda_1 \in S_\delta$ such that the tour
	$T_\delta = (p_0, p_1, p_0)$. Hence, we can check if such a tour exists and, if this is the case, compute one in $\Oh(|E| \cdot |S_\delta|^2) = \Oh(|E|)$.

	We may hence assume that $\alpha(T_\delta) \geq 3$. By \cref{lem:one_tour_approx}, 
	we can hence compute and output a \deltatour[1] $T$ 
	such that $\len(T) \leq 3 \cdot \opttour[1]$, where $\opttour[1]$ is
	the length of a shortest \deltatour[1] in $G$.
	By \cref{lem:make1}, we have that $\opttour[1] \leq \len(T_\delta)/(3-2\delta)$.
	As $\delta \geq 1$, clearly, $T$ is a $\delta$-tour, which is of length
	$\len(T) \leq 3/(3-2\delta) \cdot \len(T_\delta)$.
\end{proof}

For $\delta \in [33/40, 1)$, the idea, roughly speaking, is to use a
\deltatour[1] computed by \cref{lem:one_tour_approx} and add a collection of
short tour segments to complete it into a \deltatour for the desired
$\delta \in [30/40,1)$.
The following result will be used to bound the cost of these detours in comparison
to the cost of the whole tour.

\begin{lemma}\label{lem:dreivbound}
Let $\delta \in [33/40,1)$, let $G$ be a connected graph, let $V_1$ be the vertices of degree 1 in $G$, and
let $T$ be a \deltatour in $G$ with $\alpha(T)\geq 3$.
Then $\len(T) \geq 2(1-\delta)\abs{V_1} + 4(1-\delta)(\abs{V(G)}-\abs{V_1})$.
\end{lemma} 
\begin{proof}
By Lemma~\ref{lem:tournice}, we may suppose that $T$ is nice. For every $v \in V(G)$, we define a real $\Lambda_v$ as follows.
We initialize $\Lambda_v$ by $0$.
Then, for every tour segment of $T$ of the form $\seq{u&v}$ for some $u \in N_G(v)$,
we add $1$ to $\Lambda_v$ and for every segment of the form
$\seq{u&p(u,v,\lambda)&u}$ for some $u \in N_G(v)$, we add $2 \lambda$ to
$\Lambda_v$. We do this for every $v \in V(G)$.

 \begin{claim}\label{clm:zvgross}
We have $\Lambda_v \geq 2 (1-\delta)$ for all $v \in V_1$ and $\Lambda_v \geq 4 (1-\delta)$ for all $v \in V(G)-V_1$.
\end{claim}
\begin{claimproof}
 Let $v \in V(G)$. If $T$ stops at $v$, then, as $\alpha(T)\geq 3$ and $T$ is nice, we have that $T$ contains the segment~$\seq{u&v}$ for some
 $u \in N_G(v)$. As $\delta \geq \frac{3}{4}$, this yields
 $\Lambda_v \geq 1 \geq 4 (1-\delta)$.
 We may hence suppose that $T$ does not stop at $v$.
 If $v \in V_1$, let $u$ be the unique vertex in $N_G(v)$.
As $T$ is a nice tour with $\alpha(T)\geq 3$, we have that $u \notin V_1$.
 By \cref{char1stop} and \cref{char2stops}, it follows that $T$ contains a segment of the
 form $\seq{u&p(u,v,\lambda)&u}$ for some $\lambda \geq 1-\delta$.
 By construction, we obtain that $\Lambda_v \geq 2\lambda \geq 2(1-\delta)$.
 We may in the following suppose that $v \in V(G) \setminus V_1$.

First suppose that there is some $u \in N_G(v)$ such that $T$ does not stop at $u$. By Lemma~\ref{nostopchar}, as $\alpha(T)\geq 3$ and as $\delta<1$, there are $u' \in N_G(u)$, $v' \in N_G(v)$ and $\lambda_1,\lambda_2 \in [0,1)$ with $\lambda_1+\lambda_2\geq 3-2\delta$ such that $T$ contains the segments~$\seq{u'&p(u',u,\lambda_1)&u'}$ and $\seq{v&p(v',v,\lambda_2)&v'}$. We obtain that $\Lambda_v\geq 2 \lambda_2\geq 2(3-2\delta-\lambda_1)\geq 2(3-2\delta-1)=4(1-\delta)$.
Now suppose that $T$ stops at $u$ for all $u \in N_G(v)$ and for all $u \in N_G(v)$, let $\lambda_u$ be the largest real such that $T$ stops at $p(u,v,\lambda_u)$. By Lemma~\ref{char1stop}, we obtain that one Lemma~\ref{char1stop} $(i)$ and $(ii)$ holds for every $u \in N_G(v)$. If Lemma~\ref{char1stop} $(i)$ holds for $uv$ for some $u \in N_G(v)$, then there is some $u' \in N_G(v)-u$ such that $\lambda_u+\lambda_{u'}\geq 2-2\delta$. We obtain $\Lambda_v\geq 2(\lambda_u+\lambda_v)=4(1-\delta)$. Otherwise, Lemma~\ref{char1stop} $(ii)$ holds for $uv$ for every $u \in N_G(v)$.  As $v \notin V_1$, we obtain $\Lambda_v\geq 2(1-\delta)|N_G(v)|\geq 4(1-\delta)$.

\end{claimproof}
By construction and \cref{clm:zvgross}, we obtain
$\len(T) = \sum_{i \in [k]}\dist_G(p_{i}, p_{i+1}) = \sum_{v \in V(G)} \Lambda_v
 \geq 2(1-\delta)\abs{V_1}+4(1-\delta)(\abs{V(G)}-\abs{V_1})$.
\end{proof}

We can now prove that \deltatourprob admits a $4$-approximation
algorithm for $\delta \in [33/40, 1)$.

\ThmApproxUbThreeQuartersOne*\label\thisthm

\begin{proof}
Let $G$ be a connected graph and let $V_1$ be the set of vertices of degree 1 of $G$. By Corollary~\ref{decide}, we may suppose that $|V(G)|\geq 3$ and hence, that every $\delta$-tour of $G$ is of discrete length at least 3. Let \opttour[1] and \opttour denote the length of a shortest 1-tour and a shortest \deltatour in $G$, respectively.
By \cref{lem:one_tour_approx}, in polynomial time, we can compute a vertex cover tour
$T$ in $G$ that satisfies $\len(T) \leq 3 \cdot \opttour[1]$. By \cref{lem:tournice}, we may suppose that $T$ is nice.
Let $V_0 \subseteq V(G)$ be the set of vertices which are not stopping points of $T$,
let $v_1, \dots, v_q$ be an arbitrary ordering of $V_0$,
and for $i \in [q]$, let $w_i$ be an arbitrary neighbor of $v_i$ in $G$.
As $T$ stops at vertices forming a vertex cover, we obtain that $T$ stops at $w_i$ for all $i \in [q]$.
We now recursively construct a sequence of tours $T^0,\dots,T^q$ in $G$ all of which stop at all stopping points of $V(G) \setminus V_0$.
First we set $T^0=T$. Then, for any $i \in [q]$, if $v_i \in V_1$, we construct
$T_i$ from $T_{i-1}$ by choosing an arbitrary occurrence of $w_i$ in $T$ and
replacing it by the segment~$\seq{w_i&p(w_i,v_i,1-\delta)&w_i}$.
Otherwise, we construct $T_i$ from $T_{i-1}$ by choosing an arbitrary
occurrence of $w_i$ in $T$ and replacing it by the segment
$\seq{w_i&p(w_i,v_i,2(1-\delta))&w_i}$.  We let $T_\delta = T^q$. 
\begin{claim}
$T_\delta$ is a \deltatour in $G$.
\end{claim}
\begin{claimproof}
As $T$ is nice and by construction, so is $T_\delta$. Let $xy \in E(G)$. If $T_\delta$ stops at both $x$ and $y$, then $xy$ is covered by $T_\delta$ by \cref{char2stops}. By construction and symmetry, we may hence suppose that $y \in V_0$. If $y \in V_1$, by construction, we obtain that $T_\delta$ contains the segment~$\seq{x&p(x,y,1-\delta)&x}$, so $xy$ is covered by $T_\delta$ by \cref{char1stop}. Otherwise, we obtain by construction that there is some $w \in N_G(y)$ such that $T_\delta$ contains the segment~$\seq{w&p(w,y,2(1-\delta))&w}$. It follows by \cref{char1stop} that $xy$ is covered by $T_\delta$.
\end{claimproof}
From the definition of $T$, the fact that every $\delta$-tour is in particular a $1$-tour, the
construction and \cref{lem:dreivbound}, it follows that

\begin{align*}
\len(T_\delta)&=\len(T)+2(1-\delta)\abs{V_0\cap V_1}+4(1-\delta)\abs{V_0-V_1}\\
	      &\leq 3 \cdot \opttour[1]+2(1-\delta)\abs{ V_1}+4(1-\delta)\abs{V(G)-V_1}\\
	      &\leq 3 \cdot \opttour+\opttour = 4 \cdot \opttour.
\end{align*}
\end{proof}

\subsection[\texorpdfstring{Covering Range $\delta > 3/2$}{Covering Range delta > 3/2}]{Covering Range \boldmath$\delta > 3/2$}\label{sec:appro6}
\label{sec:approx:large_delta}

In this section, we consider the problem of finding approximation algorithms for $\delta$-\tour if $\delta$ is large. We move away from the design of constant-factor approximation
algorithms and instead give approximation algorithms with  
$\polylog(n)$-factor guarantees.
In~\cite{FreiGHHM24}, these results are complemented by a lower bound showing that, assuming $\p \neq \np$, an $o(\log{n})$-approximation algorithm for $\delta$-\tour does not exist, even for fixed $\delta \geq \frac{3}{2}$.

We consider two different settings:
In the first one, we wish to compute a shortest \deltatour for some fixed $\delta \geq \frac{3}{2}$ and in the second one, we let $\delta$ be part of the input. Note that in the latter case, $\delta$ may be linear in 
$n$; thus, approximations factors of the form $\Oh(f(\delta) \log{n})$,
which yield the desired result in the first setting may be worthless in the second setting.

We now give a construction which is the crucial ingredient for the proofs of Theorems~\ref{thm:approx:ub:threehalves:logn} and~\ref{thm:approx:ub:threehalves:lognpthree}.
Namely, for a connected graph~$G$ and a constant $\delta$, we construct an auxiliary graph~$\Gamma(G,\delta)$.
  We then show that there is a close relationship among $\delta$-tours in $G$ and dominating sets in $\Gamma(G,\delta)$. In Section~\ref{sec:appro7}, we exploit this directly by designing an $O(\log n)$-approximation algorithm for arbitrary fixed $\delta$ using an approximation algorithm for computing a smallest dominating set in a given graph with a similar guarantee for the solution quality. While this approach does not work for $\delta$ being part of the input, in Section~\ref{sec:appro8}, we obtain an $O(\log^{3}n)$-approximation for this setting by relying on previous results on a slightly different problem, called {\it Minimum Dominating Tree}. 

We now detail the construction of $\Gamma(G,\delta)$. We fix a connected graph~$G$ and some $\delta>0$. We first need some intermediate definitions. 
Let $S_\delta$ be the set of edge positions defined in \cref{lem:discretization:gt:one}; recall that $|S_\delta| = \Oh(1)$.
Let $P_\delta(G) \coloneqq \{p(u, v, \lambda) \mid
	uv \in E(G), \lambda \in S_\delta\}$.

Further, we define
$Q_\delta(G)$ to consist of $V(G)$ and the set of points $q \in P(G)$ for which there is some $p \in  P_\delta(G))$ such that $\dist_G(p,q)=\delta$.

Now for some $uv \in E(G)$, let $\Lambda^{uv}=\{0,\lambda_1^{uv},\ldots,\lambda_{k^{uv}}^{uv},1\}$ be the collection of reals, in increasing order, such that $p(u,v,\lambda)\in Q_\delta(G)$ if and only if $\lambda \in \Lambda^{uv}$. We now let $\mathcal{I}^{uv}_{G,\delta}$ be the set of segments consisting of $P(uv)$ if $\Lambda^{uv}=\{0,1\}$ and otherwise of the segments~$P(u,p(u,v,\lambda_1^{uv}))$ and $P(p(u,v,\lambda_i^{uv}),p(u,v,\lambda_{i+1}^{uv}))$ for $i \in [k^{uv}-1]$, and $P(p(u,v,\lambda_{k^{uv}}^{uv}),v)$. We further define $\mathcal{I}_{G,\delta}=\bigcup_{uv \in E(G)}\mathcal{I}^{uv}_{G,\delta}$.

We are now ready to give the definition of $\Gamma(G,\delta)$.
We let $V(\Gamma(G,\delta))$ consist of $P_\delta(G)$ and a vertex $x_I$ for every $I \in \mathcal{I}_{G,\delta}$. 
We next let $E(\Gamma(G,\delta))$ contain an edge linking $p$ and $p'$ for all distinct $p,p' \in P_\delta(G)$. 
Finally, we let $E(\Gamma(G,\delta))$ contain an edge linking some $p \in P_\delta(G)$ and $x_I$ for some 
$I=P(q_1,q_2) \in \mathcal{I}_{G,\delta}$ if $\dist_G(p,q_i)<\delta$ holds for some $i \in \{1,2\}$. This finishes 
the description of $\Gamma(G,\delta)$.

In order to algorithmically make use of this construction, we need the following easy result.
\begin{proposition}\label{compgamma}
We can compute $\Gamma(G,\delta)$ in polynomial time given $G$ and $\delta>1$ and we have $|V(\Gamma(G,\delta))|=O(n^4)$.
\end{proposition}
\begin{proof}

As the set $S_\delta$ is given explicitly in Lemma~\ref{lem:discretization:gt:one}, we can compute $P_\delta(G)$ in polynomial time and we have $|P_\delta(G)|=O(n^2)$. Now consider some $p \in P_\delta(G)$ and some $uv \in E(G)$. If $p$ is on $uv$, then, as $\delta>1$, there is no point $q$ on $uv$ with $\dist_G(p,q)=\delta$. It follows that every point $q$ on $uv$ with $\dist_G(p,q)=\delta$ satisfies $\dist_G(q,u)=\delta-\dist_G(p,u)$ or $\dist_G(q,v)=\delta-\dist_G(p,v)$. It follows that $P(u,v)$ contains at most 2 points which are of distance exactly $\delta$ to $p$ and that they can be computed in polynomial time. We now execute this computation for every $p \in P_\delta(G)$ and every $uv \in E(G)$. By this procedure, we can compute $Q_\delta(G)$ in polynomial time and we obtain $|Q_\delta(G)|\leq |V(G)|+2|P_\delta(G)||E(G)|=O(n^4)$. Form this, we can clearly compute $\mathcal{I}_{G,\delta}$ in polynomial time and we obtain $|\mathcal{I}_{G,\delta}|\leq |Q_\delta(G)|=O(n^4)$. Finally, we can compute the edge set of $\Gamma(G,\delta)$ by some shortest walk computations.
\end{proof}
We now prove the key lemma, which is Lemma~\ref{lem:domsettour} and restated here for convenience. It shows that there is a close relationship between $\delta$-tours in $G$ and dominating sets in $\Gamma(G,\delta)$.

\DomSetTour*\label\thisthm
\begin{proof}

First suppose that the set $P_T$ of stopping points of $T$ is a dominating set of $\Gamma(G,\delta)$. In order to see that $T$ is a $\delta$-tour, let $p=p(u,v,\lambda)\in P(G)$. Then, by construction, there are some $p_0,p_1 \in Q_\delta(G)$ such that $p\in P(p_0,p_1)$ and $I\in \mathcal{I}_{G,\delta}$ where $I=P(p_0,p_1)$. As $P_T$ is a dominating set in $\Gamma(G,\delta)$, there is some $q \in P_T$ such that $qx_I \in E(G)$. By the definition of $\Gamma(G,\delta)$ and by symmetry, we may suppose that $\dist_G(q,p_0)<\delta$. If $\dist_G(q,p)>\delta$, then there are some $p'\in P(p,p_0)-p_0$ such that $\dist_G(q,p')=\delta$. As $p'\in P(p_0,p_1)-\{p_0,p_1\}$, we obtain a contradiction to the fact that $P(p_0,p_1)\in \mathcal{I}_{G,\delta}$.

Now suppose that $T$ is a $\delta$-tour all of whose stopping points are in $P_\delta(G)$, let $P_T$ be the set of stopping points of $T$, and consider some $I \in \mathcal{I}_{G,\delta}$. Then, by definition, there are some unique $p,p' \in Q_{\delta}(G)$ such that $I=P(p,p')$. Let $u,v \in V(G)$ and $\lambda,\lambda' \in [0,1]$ such that $p=p(u,v,\lambda)$ and $p'=p(u,v,\lambda')$. If some point in $P(p,p')$ is passed by $T$, let $v \in V(G)\cap P_T$ be chosen such that $\dist_G(p,v)$ is minimized. As $\delta >1$, we have $\dist_G(p,v)\leq 1<\delta$, so $vx_I \in E(\Gamma(G,\delta))$. We may hence suppose that no point in $P(p,p')$ is passed by $\delta$. By Proposition~\ref{nearstop}, we can now choose some $q,q'\in P_T$ such that $\dist_G(p,q)=\dist_G(p,T)$ and $\dist_G(p',q')=\dist_G(p',T)$. As $T$ is a $\delta$-tour, we have $\dist_G(p,q)\leq \delta$ and $\dist_G(p',q')\leq \delta$. If $\dist_G(p,q)=\dist_G(p',q')=\delta$ holds, let $p''=(u,v,\frac{1}{2}(\lambda+\lambda'))$. As no point in $P(p,p')$ is passed by $\delta$, we obtain that $\dist_G(p'',T)=\min\{\dist_G(p'',p)+\dist_G(p,T), \dist_G(p'',p')+\dist_G(p',T)\}>\min\{\dist_G(p,T), \dist_G(p',T)\}=\delta$. This contradicts $T$ being a $\delta$-tour. We may hence suppose by symmetry that $\dist_G(p,q)<\delta$ holds. We hence have $qx_I \in E(\Gamma(G,\delta))$. As $I \in \mathcal{I}_{G,\delta}$ was chosen arbitrarily, we obtain that $V_T$ is a dominating set of $\Gamma(G,\delta)$.
\end{proof}

\subsubsection[\texorpdfstring{Fixed Covering Range $\delta$}{Fixed Covering Range delta}]{Fixed Covering Range \boldmath$\delta$}\label{sec:appro7}
The purpose of this section is to give an approximation algorithm for the case that some $\delta \geq 3/2$ is fixed. Our algorithm is based on the computation of a dominating set in the auxiliary graph defined above. We first need the following result that shows that a suitable dominating set of this auxiliary graph can be efficiently connected into a tour in $G$.
\begin{proposition}\label{drtfgzuih}
Let $G$ be a connected graph, $\delta\geq 3/2$ a real, and $Y \subseteq P_\delta(G)$ a dominating set of $\Gamma(G,\delta)$. Then, in polynomial time, we can compute a tour in $G$ that stops at all points of $Y$ and whose length is at most $4\delta|Y|$.
\end{proposition}
\begin{proof}
Let $H$ be the graph with $V(H)=Y$ and where $E(H)$ contains an edge linking two points $y_1$ and $y_2$ if $\dist_G(y_1,y_2)\leq 2 \delta$.
\begin{claim}
$H$ is connected.
\end{claim}
\begin{claimproof}
Suppose otherwise and let $C$ be a component of $H$. We further choose some $v_1 \in V(C)$ and some $v_2 \in V(H)-V(C)$ such that $\dist_G(v_1,v_2)$ is minimized. Let $W$ be a shortest walk from $v_1$ to $v_2$ in $G$. As $v_1$ and $v_2$ are in distinct connected components of $H$, we have that $E(H)$ does not contain an edge linking $v_1$ and $v_2$. Hence we have $\len(W)>2\delta$. It follows that there is a point $p$ passed by $W$ that satisfies $\min\{\dist_G(v_1,p),\dist_G(v_2,p)\}>\delta$. If there is some $v_1'\in V(C)$ such that $\dist_G(v_1',p)\leq \delta$, then we have $\dist_G(v_1',v_2)\leq \dist_G(v_1',p)+\dist_G(v_2,p)<\dist_G(v_1,p)+\dist_G(v_2,p)=\len(W)=\dist_G(v_1,v_2)$, a contradiction to the choice of $v_1$ and $v_2$. A similar argument shows that $\dist_G(v_2',p)>\delta$ for all $v_2'\in V(H)-V(C)$. By Lemma~\ref{lem:domsettour}, we obtain that $Y$ is not a dominating set of $\Gamma(G,\delta)$, a contradiction.
\end{claimproof}
Now let $U$ be the multigraph obtained from an arbitrary spanning tree of $H$ by replacing every edge by two parallel copies of itself. By Proposition~\ref{euler}, we obtain that there is an Euler tour~$T_0=p_0\ldots p_z$ of $U$. We now obtain $T$ from $T_0$ by replacing the segment~$\seq{p_{i-1}&p_{i}}$ by a shortest walk from $p_{i-1}$ to $p_i$ for $i \in [z]$. Observe that $\len(T)=\sum_{i \in [z]}\dist_G(p_{i-1},p_i)\leq 2\delta z=4\delta (|V(H)|-1)\leq 4 \delta |Y|$. Further, as $T$ stops at all vertices in $Y$ and $Y$ is a dominating set of $\Gamma(G,\delta)$, we obtain that $T$ is a $\delta$-tour in $G$ by Lemma~\ref{lem:domsettour}. 
\end{proof}
Our proof is based on the following well-known result on the approximation of the problem of finding a smallest dominating set; see for example \cite{JOHNSON1974256}.
\begin{proposition}\label{approxdom}
There is an $O(\log n)$-approximation algorithm for finding a minimum dominating set in a given graph.
\end{proposition}
We further need the following simple observation. For some $\delta\geq 0$, let $s_\delta=\min\{|s_1-s_2|\mid \{s_1,1-s_1\} \cap S_\delta\neq \emptyset\neq \{s_2,1-s_2\} \cap S_\delta, s_1 \notin \{s_2,1-s_2\} \}$.
\begin{proposition}\label{lentour}
Let $T$ be a tour in a graph $G$ all of whose stopping points are in $P_\delta(G)$. Then $\alpha(T)\leq \lceil\len(T)/s_\delta\rceil$.
\end{proposition}
\begin{proof}
Let $T=p_0\ldots p_z$.  As $p_i \in P_\delta(G)$ for $i\in[z]$, we obtain $\dist_G(p_{i-1},p_i)\geq s_\delta$ for $i \in [z]$. This yields $\len(T)=\sum_{i\in[z]}\dist_G(p_{i-1},p_i)\geq \alpha(T)s_\delta$.
\end{proof}
We are now ready to give the proof of Theorem~\ref{thm:approx:ub:threehalves:logn}, which we restate here. 
\ThmApproxUbThreeHalvesLogN *\label\thisthm
\begin{proof}
Let $G$ be a connected graph.
By Corollary~\ref{2points}, we may suppose that $G$ does not admit a shortest $\delta$-tour with at most two 
stopping points. We first compute $\Gamma(G,\delta)$ which is possible in polynomial time by 
Proposition~\ref{compgamma}. Now, by Proposition~\ref{approxdom}, we can compute a dominating set $Y$ of 
$\Gamma(G,\delta)$ that satisfies $|Y|\leq \alpha \log(|V(\Gamma(G,\delta))|)|Y^*|$ where $\alpha$ is an absolute 
constant and $Y^*$ is a minimum size dominating set of $\Gamma(G,\delta)$. By Proposition~\ref{compgamma}, we have 
that $|V(\Gamma(G,\delta))|=O(n^4)$ and hence $|Y|\leq \alpha' \log(n)|Y^*|$ for some absolute constant $\alpha'$. 
We now obtain $Y'$ from $Y$ by replacing every $v \in Y-P_{\delta}(G)$ by some arbitrary $v' \in 
N_{\Gamma(G,\delta)}(v)$, keeping only one copy of multiple elements in $Y'$. As $\Gamma(G,\delta)[P_{\delta}(G)]$ 
is a clique and $\Gamma(G,\delta)-P_{\delta}(G)$ is an independent set by construction, we obtain that $Y'$ is a 
dominating set of $\Gamma(G,\delta)$. Further, by Proposition~\ref{drtfgzuih}, in polynomial time, we can compute a 
tour~$T$ in $G$ whose length is at most $4 \delta |Y'|$ and that stops at all points of $Y'$. We now output $T$. 
Observe that $T$ can be computed in polynomial time. As $T$ stops at all points of $Y'$, we obtain by 
Lemma~\ref{lem:domsettour} that $T$ is a $\delta$-tour in $G$. In order to bound the length of $T$, let $T^*$ be a 
shortest $\delta$-tour in $G$. Let $P_{T^*}$ be the set of stopping points of $T^*$. By 
Lemma~\ref{lem:discretization:gt:one}, we may suppose that $P_{T^*}$ is contained in $P_{G,\delta}$. It hence 
follows by Lemma~\ref{lem:domsettour} that $P_{T^*}$ is a dominating set of $\Gamma(G,\delta)$.  We obtain that $\len(T^*)\geq s_\delta|P_{T^*}|$ by \cref{lentour}. This 
yields 
\begin{align*}
\len(T)&\leq 4 \delta |Y'|\\
&\leq 4 \delta |Y|\\
&\leq 4 \delta \alpha' \log(n)|Y^*|\\
&\leq 4 \delta \alpha' \log(n)|P_{T^*}|\\
&\leq 4 \delta \frac{1}{s_\delta}\alpha' \log(n)\ell(T^*).
\end{align*} 

This finishes the proof.
\end{proof}

\subsubsection[\texorpdfstring{Covering Range $\delta$ as Part of the Input}{Covering Range delta as Part of the Input}]{Covering Range \boldmath$\delta$ as Part of the Input}\label{sec:appro8}

In the setting where $\delta$ is provided as part of the input, the  algorithm from the previous section does not yield any non-trivial approximation guarantee. This is due to two reasons: first, we get an additional factor of roughly $\delta$ when connecting the dominating set into a tour and second, our estimation for the length of a shortest tour gives another factor of $\frac{1}{s_\delta}$ which may be large for certain values of $\delta$.

In an attempt to overcome this problem, we use a result for a slightly different problem that also deals with dominating sets, but, in addition, also takes into account the distance of the points in the dominating set that is sought for.

Formally, given a graph~$H$, a {\it dominating} tree $U$ of $H$ is a subgraph of $H$ that is a tree and such that $V(U)$ is a dominating set of $H$.  Further, given a weight function $w\colon E(H)\rightarrow \mathbb{R}_{\geq 0}$, the weight $w(U)$ of a subgraph~$U$ of $H$ is defined to be $\sum_{e \in E(U)}w(e)$.

We use the following result due to Kutiel~\cite{Kutiel18}.

\begin{lemma}[{\cite{Kutiel18}}]
\label{lem:mdt_approx}
	There is a polynomial-time algorithm that,
	given a weighted graph~$G$ of order $n$,
	computes a $(\log n)^3$-approximation of a minimum weight
	dominating tree of $G$.
\end{lemma}

We now fix a connected graph~$G$ and $\delta >1$ and define a weight function $w\colon E(\Gamma(G,\delta))\rightarrow \mathbb{R}_{\geq 0}$. For all $p,p'\in P_{\delta}(G)$, we set $w(pp')=\dist_G(p,p')$ and for all $p\in P_{\delta}(G)$ and $I \in \mathcal{I}_{G,\delta}$ such that $px_I\in E(G(\Gamma,\delta))$, we set $w(px_I)=n^3$.

We now give the following key lemmas relating $\delta$-tours in $G$ and dominating trees of $\Gamma(G,\delta)$.

\begin{lemma}\label{treetour}
Let $U$ be a dominating tree of $\Gamma(G,\delta)$ with $V(U)\subseteq P_\delta(G)$. Then, in polynomial time, we can compute a $\delta$-tour~$T$ of $G$ with $\len(T)\leq 2w(U)$.
\end{lemma}
\begin{proof}
Let $U'$ be the multigraph obtained from $U$ by replacing every edge by two copies of itself. By Proposition~\ref{euler}, we obtain that $U'$ has an Euler tour~$T_0=p_0\ldots p_z$. We now obtain $T$ from $T_0$ by replacing the segment~$\seq{p_{i-1}&p_i}$ by a shortest walk in $G$ from $p_{i-1}$ to $p_i$. Observe that $T$ can be constructed in polynomial time. Further, we have that $T$ stops at all points in $V(U)$. It hence follows from Lemma~\ref{lem:domsettour} that $T$ is a $\delta$-tour in $G$. Further, by the definition of $w$, we have $\len(T)=\sum_{i \in [z]}\dist_G(p_{i-1},p_i)=\sum_{i \in [z]}w(p_{i-1}p_i)=2w(U)$.
\end{proof}

\begin{lemma}\label{tourtree}
Let $T$ be a $\delta$-tour in $G$ all of whose stopping points are in $P_\delta(G)$. Then there is a dominating tree $U$ in $\Gamma(G,\delta)$ with $w(U)\leq \len(T)$.
\end{lemma}
\begin{proof}
Let $V_T$ be the set of stopping points of $T$.
It follows by Lemma~\ref{lem:domsettour}
that $V_T$ is a dominating set in $\Gamma(G,\delta)$.
Now let $U_0$ be the graph on $V_T$ that contains an edge linking two points
$p,p' \in P_\delta(G)$ whenever there is some $i \in [z]$ such that $\{p_{i-1},p_i\}=\{p,p'\}$.
Observe that $U_0$ is a subgraph of $\Gamma(G, \delta)$.
As $T$ is a tour, we have that $U_0$ is connected.
Now let $U$ be an arbitrary spanning tree of $U_0$. As $V_T$ is a dominating set of $\Gamma(G,\delta)$,
we obtain that $U$ is a dominating tree of $\Gamma(G,\delta)$. Further, we have $w(U)\leq \len(T)$.
\end{proof}

We are now ready to prove Theorem~\ref{thm:approx:ub:threehalves:lognpthree}, which we restate here for convenience.
\ThmApproxUbThreeHalvesLogNpThree*\label\thisthm
\begin{proof}
By Corollary~\ref{2points}, we may suppose that $G$ does not admit a $\delta$-tour with at most two stopping points. Next, by Theorems~\ref{thm:approx:ub:zero_sixth},~\ref{thm:approx:ub:sixth_half},~\ref{thm:approx:ub:half},~\ref{thm:approx:ub:half:threequarters}, and 
		\ref{thm:approx:ub:threequarters:one}, we may suppose that $\delta >1$.
Finally, by \cref{decide}, we may suppose that $n \geq 6$. We first compute $\Gamma(G,\delta)$ which is possible in 
polynomial time by Proposition~\ref{compgamma}. We further compute the weight function $w$ which is possible in 
polynomial time as $|V(\Gamma(G,\delta))|=\Oh(n^4)$ by Proposition~\ref{compgamma}. Now, by 
Lemma~\ref{lem:mdt_approx}, we can compute a dominating tree $U$ of $\Gamma(G,\delta)$ that satisfies $w(U)\leq 
\log^3(|V(\Gamma(G,\delta))|)w(U^*)$ where $U^*$ is a minimum weight dominating tree of $\Gamma(G,\delta)$. By 
Proposition~\ref{compgamma}, we have that $|V(\Gamma(G,\delta))|=O(n^4)$ and hence $w(U)\leq 64 \log^3(n)w(U^*)$. 
Clearly, a spanning tree of $G$ corresponds to a dominating tree of $\Gamma(G,\delta)$ of the same weight. This 
yields that $w(U^*)\leq n-1$. If $V(U)-P_{\delta}(G)$ is nonempty, as $n \geq 6$, we have $w(U)\geq n^3 > 4 
\log^3(n)(n-1)\geq 4 \log(n)w(U^*)$, a contradiction. We hence obtain $V(U)\subseteq P_\delta(G)$. Now, by 
Proposition~\ref{treetour}, in polynomial time, we can compute a tour~$T$ in $G$ whose length is at most $2w(U)$ and 
that stops at all points of $V(U)$. We now output $T$. Observe that $T$ can be computed in polynomial time. As $T$ 
stops at all points of $V(U)$ and $V(U)$ is a dominating set of $\Gamma(G,\delta)$, we obtain by 
Lemma~\ref{lem:domsettour} that $T$ is a $\delta$-tour in $G$. In order to bound the length of $T$, let $T^*$ be a 
shortest $\delta$-tour in $G$. Let $P_{T^*}$ be the set of stopping points of $T^*$. By 
Lemma~\ref{lem:discretization:gt:one}, we may suppose that all stopping points in $P_{T^*}$ are contained in 
$P_{G,\delta}$. It hence follows by Lemma~\ref{tourtree} that there is a dominating tree $U_0$ of 
$\Gamma(G,\delta)$ with $w(U_0)\leq \len(T^*)$. This can conclude the proof with the following inequality chain.  
\begin{align*}
\len(T)&\leq 2w(U)\\
&\leq 64 \log^3(n)w(U^*)\\
&\leq 64 \log^3(n)w(U_0)\\
&\leq 64 \log^3(n)\len(T^*).
\end{align*} 
\end{proof}

\bibliography{lit}

\clearpage
\tableofcontents
\label{toc}

\end{document}